\documentclass[11pt]{article}
\usepackage{float}

 \usepackage{basket_Fourier}
 \usepackage[T1]{fontenc}
 \usepackage[utf8]{inputenc}
 \usepackage{lmodern}
 \usepackage{microtype}
 \usepackage{cmap} 
 
 \input{glyphtounicode}
\makeatletter
\def\BState{\State\hskip-\ALG@thistlm}
\makeatother

\usepackage[thinlines]{easytable}
\renewcommand{\arraystretch}{1.3}
\title{Data-Driven Stochastic Optimal Control for Intraday Electricity Trading by Renewable Producers}

\author[1]{Chiheb Ben Hammouda}
\author[2]{Michael Samet\thanks{samet@uq.rwth-aachen.de}}
\author[3]{Ra\'ul Tempone}
\affil[1]{Mathematical Institute, Utrecht University, Utrecht, the Netherlands}
\affil[2]{Chair of Mathematics for Uncertainty Quantification, RWTH Aachen University, Aachen, Germany.}
\affil[3]{King Abdullah University of Science and Technology (KAUST), Computer, Electrical and Mathematical Sciences \& Engineering Division (CEMSE), Thuwal, Saudi Arabia.}

\usepackage{comment}
\usepackage{tikz}

\usetikzlibrary{calc}
\usetikzlibrary{shapes.geometric, arrows}
\tikzstyle{startstop} = [rectangle, rounded corners, minimum width=3cm, minimum height=1cm, text centered, text width=4.5cm, draw=black, fill=green!30]
\tikzstyle{io} = [trapezium, trapezium left angle=70, trapezium right angle=110, minimum width=3cm, minimum height=1cm, text centered, draw=black, fill=blue!30]
\tikzstyle{process} = [rectangle, rounded corners, minimum width=3cm, minimum height=1cm, text centered, text width=10cm, draw=black, fill=red!10]
\tikzstyle{decision} = [rectangle, minimum width=3cm, minimum height=1cm, text centered, text width=3cm, draw=black, fill=blue!30]
\tikzstyle{alternative} = [rectangle, minimum width=5cm, minimum height=1cm, text width=5.5cm, draw=black, fill=blue!10]
\tikzstyle{alternative2} = [rectangle, minimum width=5cm, minimum height=1cm, text width=5.5cm, draw=black, fill=green!10, text centered]
\tikzstyle{arrow} = [thick,->,>=stealth]
\usepackage{xcolor}

\usetikzlibrary{arrows.meta, positioning, decorations.pathreplacing, shapes.arrows}

\definecolor{forwardcolor}{RGB}{27, 158, 119}
\definecolor{dayaheadcolor}{RGB}{217, 95, 2}
\definecolor{intradaycolor}{RGB}{117, 112, 179}
\definecolor{balancingcolor}{RGB}{231, 41, 138}
\definecolor{deliverycolor}{RGB}{102, 166, 30}

\begin{document}
	\date{}
\maketitle
\begin{abstract}
The rapid growth of weather-dependent renewable generation increases
price volatility and imbalance penalty risk in power markets, creating the need for advanced quantitative
trading strategies. We
develop a data-driven continuous-time stochastic optimal control framework for intraday electricity trading using stochastic differential equations with drift terms ensuring mean reversion to deterministic forecast trajectories. Production follows a Jacobi diffusion, while prices follow an asymmetric jump-diffusion to reflect the heavy-tailed behavior observed in intraday markets. The framework accounts for realistic market features by incorporating gate closure and energy-based imbalance settlement over the delivery window, where the path-dependent imbalance cost is handled by state augmentation to preserve the Markovian structure. The value function is characterized via the dynamic programming principle by a three-stage sequence of two linear Kolmogorov backward equations and a nonlinear Hamilton-Jacobi-Bellman partial integro-differential equation. To solve this problem efficiently, we propose a monotone IMEX finite-difference scheme with operator splitting, semi-implicit linearization, and a differential formulation for the jump operator. Numerical experiments based on German market data indicate that, under the provided forecasts, the computed strategy outperforms the TWAP benchmark and approaches the perfect-foresight benchmark.  Sensitivity experiments further show how jump intensity, delivery-window length, and trading horizon affect the trading policy and the resulting profit-and-loss distribution.

\end{abstract}


\tableofcontents

\section{Introduction}
The increasing penetration of variable renewable generation, particularly wind and solar, has become a primary source of uncertainty in electricity markets. Forecast errors in renewable generation, combined with volatile intraday prices and liquidity frictions expose producers to imbalance charges levied by TSOs when the committed trading position deviates from the actual metered production. Continuous intraday markets allow producers to revise committed positions up to a fixed exchange-specific lead time before the start of physical delivery, providing a mechanism to reduce imbalance exposure. Designing profitable intraday trading strategies that revise previously committed day-ahead positions therefore requires realistic continuous-time stochastic models for both power production and prices, and a representation reflecting the imbalance settlement mechanism appropriately.  

Building on the optimal execution framework of \cite{almgren2001optimal}, \cite{aid2016optimal} formulate an intraday electricity trading problem in which producers trade toward a stochastic terminal target representing the residual demand (the net difference between load and renewable production), so that the final inventory position is aligned with realized production shortfall. In their formulation, both intraday prices and residual demand forecasts follow arithmetic Brownian motion (ABM) dynamics, and linear temporary impact together with a quadratic terminal imbalance penalty lead to a linear-quadratic control problem. As a result, the associated Hamilton-Jacobi-Bellman (HJB) equation admits a closed-form solution. While this approach provides analytical insight, the modeling assumptions are largely driven by tractability and limit the ability of the model to reproduce several empirically observed features of intraday electricity markets. A related formulation is studied in \cite{tan2018optimal}, where a wind producer trades across forward, day-ahead, and intraday markets through a single aggregated price process. Their model combines linear market impact with a utility-based objective. Prices and wind speed are described by ABMs, and wind speed is converted into power production through a deterministic power curve. The resulting HJB equation admits a closed-form solution, which allows for an explicit analysis of how market conditions and risk-preference parameters affect the optimal trading strategy.

Subsequent work by \cite{glas2020intraday} extends the framework of \cite{aid2016optimal} to a producer operating renewable and multiple thermal generation units and incorporates data-driven, time-varying liquidity and production costs calibrated from limit-order-book data. The richer modeling of market impact and production costs renders the associated HJB analytically intractable, and the authors therefore resort to a numerical solution based on policy iteration. Despite these advances in modeling execution costs, residual demand and intraday price  are still described by ABM dynamics and assumed to be independent in contrast to econometric evidence documenting dependence between renewable production and intraday price forecasts \cite{kiesel2017econometric}. 

A related strand of the literature shifts focus from production uncertainty to price dynamics and execution risk. In this direction, \cite{cartea2022optimal} study intraday cross-border electricity trading in a European power network, where prices follow mean-reverting jump-diffusion dynamics with country-specific compound Poisson spikes. Jump sizes are modeled as Gaussian, leading to symmetric price spikes, and the associated HJB equation was solved explicitly via a Riccati-type ansatz. A different line of work focuses on execution risk and market impact at the microstructural level of intraday trading. Inspired by Hawkes-process-based execution models in equity markets \cite{alfonsi2016dynamic}, \cite{chatziandreou2026optimal} study optimal execution problems in which order flow is modeled by self-exciting point processes. In this framework, trading intensity depends on past order arrivals through a decay kernel, capturing endogenous clustering and transient impact effects. While Hawkes processes are non-Markovian in general, restricting the kernel to an exponential form yields a finite-dimensional Markovian representation, allowing the associated HJB equation to be solved analytically and leading to a tractable characterization of optimal execution strategies. The modeling of transient market impact is further generalized in \cite{abi2025optimal}, where the price impact is described through  general propagator functions, allowing the effect of past trades to persist over time in a more general way. In this setting, the optimal control problem no longer admits a finite-dimensional  HJB equation, and the authors reformulate the problem as a stochastic Fredholm integral equation, which they solve numerically. 

Despite methodological differences, the aforementioned works share several structural limitations. First, imbalance charges are typically modeled by penalty functions applied only at a single delivery instant, rather than through settlement based on energy integrated over an imbalance settlement period \cite{brusco2025rolling}. Second, the exchange-specific lead time inherent to intraday markets is not modeled explicitly. Finally, power production and price dynamics are often simplified, either by treating production deterministically or by relying on stylized stochastic models that fail to reproduce key empirical features of intraday electricity markets (see  \cite{aid2015electricity} for an overview).

We address these limitations by formulating a continuous-time stochastic optimal control framework that models intraday electricity markets more realistically. The key contributions are:
\begin{enumerate}

	\item  \textbf{Market mechanism and operational timeline:} We model imbalance penalties as path-dependent costs on cumulative metered
energy delivered over a settlement window and explicitly account for gate
closure and the lead time between the end of trading and the start of
physical delivery.
 To the best of our knowledge, existing continuous-time stochastic control formulations do not simultaneously capture gate closure and delivery over a settlement window within a unified framework. This yields a novel three-stage formulation involving the sequence of two linear Kolmogorov backward equations and a nonlinear HJB partial integro-differential equation (PIDE).
	
\item \textbf{Forecast-driven stochastic modeling of prices and production:} Building on the data-driven parametric SDE framework of \cite{caballero2021quantifying,ben2025data} for modeling forecast errors in wind and solar power production, we model both prices and production using SDEs that mean-revert around deterministic forecast trajectories. In contrast to commonly used specifications, forecast-driven dynamics incorporate information directly contained in the forecast paths themselves, such as seasonality and intraday patterns. Embedding the forecast trajectories in the drift renders the dynamics time-inhomogeneous and leads to a non-autonomous stochastic control problem, in which the optimal policy adjusts to anticipated future production and price movements rather than relying solely on current observations. Moreover, we enrich the intraday price dynamics by incorporating  an asymmetric jump component. While \cite{cartea2022optimal} assumes symmetric jump size distributions, we allow for asymmetric jumps, reflecting the empirical observation that upward and downward price movements differ in frequency and magnitude. This modeling choice is supported by empirical evidence of directionally unbalanced price jumps in intraday electricity markets \cite{kostrzewski2021impact} and allows the model to capture abrupt market events such as sudden generation outages, grid incidents, or large forecast errors.

	\item \textbf{Numerical methodology for nonlinear HJB-PIDE:} The presence of jumps in the price dynamics renders the associated
HJB equation nonlocal, leading to a PIDE.
Consequently, the resulting SOC problem must be solved numerically. We tackle this by developing
a monotone finite-difference scheme based on operator splitting and implicit-explicit (IMEX) time
discretization, with a suitable linearization of the Hamiltonian followed by Picard iterations, together with
a structure-exploiting discretization of the jump operator that avoids off-grid interpolation and
preserves monotonicity under the CFL condition \eqref{eq:cfl}. 
\end{enumerate}
Numerical experiments on German market data illustrate how penalty specifications, liquidity costs, and jump intensities shape the optimal trading policy. Under the provided forecasts, the optimal strategy achieves improved profit-and-loss distributions relative to the TWAP benchmark and approaches the perfect-foresight benchmark.

The remainder of the paper is organized as follows. Section \ref{sec:problem_setting_pricing_framework} introduces the market model for wind production, intraday prices, and inventory. Section \ref{sec:SOC_formulation} formulates the three-stage stochastic control problem and derives the associated P(I)DEs to solve. Section \ref{sec:numerical_methodology} presents the monotone IMEX operator-splitting scheme, the discretization of the jump operator, and the associated numerical analysis, including domain truncation, boundary treatment, and monotonicity properties. Section \ref{sec:num_exp_results} presents the numerical experiments and results on German market data. The appendices contain auxiliary proofs.

\section{Stochastic Market Model}\label{sec:problem_setting_pricing_framework}
Let $(\Omega,\mathcal{F},\mathbb{F}=(\mathcal{F}_t)_{t\in[0,T]},\mathbb{P})$ be a filtered probability space satisfying the right-continuity and completeness conditions \cite{oksendal2007applied}. We denote by $T_{\mathrm{gc}} > 0$ the gate closure time, i.e., the last time at which trading is permitted, by $h > 0$ the exchange-specific lead time between the gate closure and the start of physical delivery, and by $L > 0$ the product-specific delivery window length.

The problem horizon is then $T := T_{\mathrm{gc}} + h + L$, corresponding to the delivery end time. These quantities define three consecutive intervals: the trading window $[0, T_{\mathrm{gc}}]$, the lead-time period $[T_{\mathrm{gc}},\ T_{\mathrm{gc}}+h]$, and the imbalance settlement period $[T - L, T]$. The market timeline structure is summarized schematically in Figure~\ref{fig:electricity_trading_timeline}.
\begin{figure}[H]
	\centering
	\caption{Timeline of German electricity markets with delivery window and gate closure.}
	\label{fig:electricity_trading_timeline}
	\resizebox{\linewidth}{!}{%
		\begin{tikzpicture}[font=\sffamily,>=Stealth,scale=0.9,transform shape]
			\definecolor{daC}{RGB}{210,120,40}
			\definecolor{idC}{RGB}{90,90,180}
			
			\coordinate (xmin)       at (-3.5,0);
			\coordinate (xmax)       at (11.3,0);
			\coordinate (DAopen)     at (-2.4,0);
			\coordinate (DAclose)    at (-0.6,0);   
			\coordinate (IDopen)     at (0.6,0);    
			\coordinate (DayDstart)  at (2.0,0);    
			\coordinate (Tzero)      at (0.0,0);    
			\coordinate (Tbar)       at (5.0,0);    
			\coordinate (TminusL)    at (7.0,0);    
			\coordinate (Tend)       at (9.4,0);    
			\coordinate (DayDend)    at (10.6,0);   
			
			\draw[-{Stealth[length=2.4mm]}] (xmin) -- (xmax) node[right, font=\tiny]{Time};
			
			\node[fill=daC!20, rounded corners, minimum height=8mm,
			text width=32mm, align=center]
			at (-2.3,1.45) {Day-Ahead (DA) Hourly Bidding};
			
			\node[fill=idC!20, rounded corners, minimum height=8mm,
			text width=41mm, align=center]
			at (2.8,1.45) {Intraday (ID) \\ Continuous Trading};
			
			\node[fill=black!12, rounded corners, minimum height=8mm,
			text width=22mm, align=center]
			at (8.2,1.45) {Balancing\\Market};
			
			\fill[idC!10] (TminusL) rectangle ($(Tend)+(0,0.45)$);
			
			\draw[dashed] (DAclose) -- ++(0,1.65);
			\draw[dashed] (IDopen) -- ++(0,1.65);
			\draw[dashed] (DayDstart) -- ++(0,1.65);
			\draw[dashed] (Tbar) -- ++(0,1.65);
			\draw[dashed] (TminusL) -- ++(0,1.65);
			\draw[dashed] (Tend) -- ++(0,1.65);
			\draw[dashed] (DayDend) -- ++(0,1.65);
			
			\draw[<->] ($(Tzero)+(0.6,0.3)$) -- node[above=2pt, font=\tiny]{Trading horizon ($T_{gc}$)} ($(Tbar)+(0,0.3)$);
			\draw[<->] ($(Tbar)+(0,0.3)$) -- node[above=2pt, font=\tiny]{Lead time ( $h$ )} ($(TminusL)+(0,0.3)$);
			\draw[<->] ($(TminusL)+(0,0.3)$) -- node[above=2pt, font=\tiny]{Delivery ( $L$ )} ($(Tend)+(0,0.3)$);
			
			\node[below=2pt, font=\tiny, rotate=90, anchor=east] at (DAclose)   {DA Closes on Day D-1 (12:00)};
			\node[below=2pt, font=\tiny, rotate=90, anchor=east] at (IDopen)    {ID Opens (15:00)};
			\node[below=2pt, font=\tiny, rotate=90, anchor=east] at (DayDstart) {Start of Day D (00:00)};
			\node[below=2pt, font=\tiny, rotate=90, anchor=east] at (Tbar)      {Gate closure ($T_{gc}$)};
			\node[below=2pt, font=\tiny, rotate=90, anchor=east] at (TminusL)   {Start of Delivery ($T_{gc} + h$) };
			\node[below=2pt, font=\tiny, rotate=90, anchor=east] at (Tend)      {End of Delivery ($T := T_{gc} + h + L$ )};
			\node[below=2pt, font=\tiny, rotate=90, anchor=east] at (DayDend)   {End of Day D (00:00)};
			
	\end{tikzpicture}}
\end{figure}
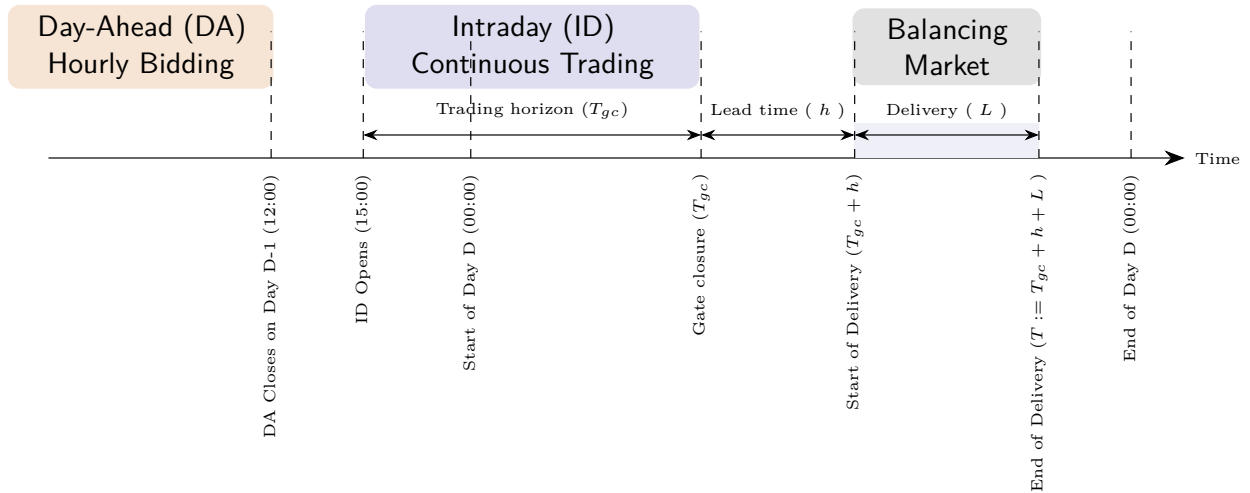

\subsection{Data-Driven Production Dynamics}
\label{sec:wind_dynamics}

In this work, we model the production as a mean-reverting process around a
deterministic forecast curve {$p_X$}
(see Section~\ref{sec:data_description_parameters} for more details).
{The normalized forecast observations are first projected onto
$[\varepsilon_{\mathrm{tr}},1-\varepsilon_{\mathrm{tr}}]$, with
$0<\varepsilon_{\mathrm{tr}}\ll1$, and then extended to $[0,T]$ using a
shape-preserving cubic Hermite interpolation. We denote the resulting
trajectory by $p_X^{\varepsilon_{\mathrm{tr}}}$. By construction,}
\begin{equation}
\label{eq:X_forecast_truncation}
{
p_X^{\varepsilon_{\mathrm{tr}}}\in C^1([0,T]),
\qquad
\varepsilon_{\mathrm{tr}}
\leq
p_X^{\varepsilon_{\mathrm{tr}}}(t)
\leq
1-\varepsilon_{\mathrm{tr}},
\qquad t\in[0,T].
}
\end{equation}
{The derivative
$\dot p_X^{\varepsilon_{\mathrm{tr}}}$ is obtained directly from the
interpolant. Henceforth, to lighten the notation, we write
$p_X:=p_X^{\varepsilon_{\mathrm{tr}}}$.}

The production dynamics feature a time-dependent mean-reversion rate,
$\theta(t)$, and a nonlinear state-dependent diffusion, given by:
\begin{equation}
	\label{eq:wind_sde}
	\left\{
	\begin{array}{l}
		\mathrm{d}X_t=\big({\dot p_X(t)}
		-\theta(t)\,(X_t-{p_X(t)}) \big)\,\mathrm{d}t
		+\sqrt{2\alpha\theta_0\,X_t\big(1-X_t\big)}\,\mathrm{d}B^X_t,\quad t\in[0,T],\\[4pt]
		X_0= p_X(0),
	\end{array}
	\right.
\end{equation}
where $\theta_0, \alpha >0$. The mean-reversion rate $\theta(t)$ is then defined as
\begin{equation}
	\label{eq:mean_reversion_condition}
	\theta(t):=\max \left\{\theta_0,
	\frac{\left|{\dot{p}_X(t)}\right|}
	{\min \left({p_X(t)}, 1-{p_X(t)}\right)}
	\right\}.
\end{equation}
The constraint in \eqref{eq:mean_reversion_condition} ensures that
the process $X_t$ in \eqref{eq:wind_sde} is almost surely (a.s.) bounded in $[0,1]$
\cite{caballero2021quantifying}.
{Moreover, \eqref{eq:X_forecast_truncation} implies that
$\dot p_X$ is bounded and that the denominator in
\eqref{eq:mean_reversion_condition} is bounded below by
$\varepsilon_{\mathrm{tr}}$. Hence, $\theta(\cdot)$ is continuous and bounded on
$[0,T]$.}
The diffusion term in \eqref{eq:wind_sde} qualitatively captures the fact that forecast
uncertainty is smaller near physical capacity limits, since we cannot
exceed them. The presence of $\theta_0$ in the diffusion form is to
follow the standard form of Jacobi-type SDEs, we refer to
\cite{iacus2008simulation} for more details.

Figure~\ref{fig:wind_paths} illustrates the production model on the
Amprion control zone (Germany) for 2024-04-11. The day-ahead forecast
$p_X(t)$, the realized production, and a sample of simulated paths from the calibrated Jacobi diffusion~\eqref{eq:wind_sde}. Model parameters are as in
Table~\ref{tab:parameters}. Further details on the data 
 are provided in Section~\ref{sec:data_description_parameters}.

\begin{figure}[!ht]
	\centering
	\includegraphics[width=0.55\linewidth]{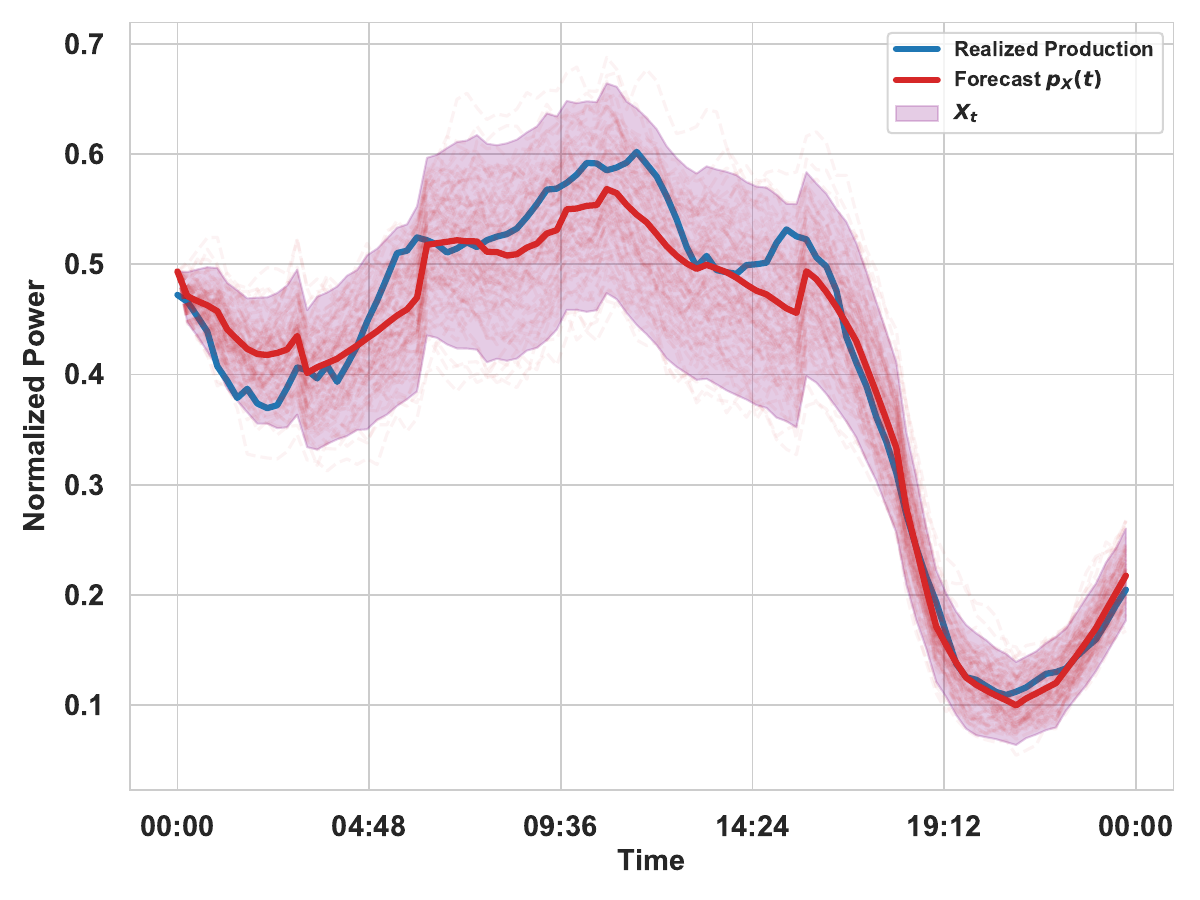}
	\caption{Normalized wind power production on 2024-04-11 in the
		Amprion control zone, Germany. Day-ahead forecast $p_X(t)$,
		realized production, and simulated paths from the calibrated
		Jacobi diffusion~\eqref{eq:wind_sde}. Parameters as in
		Table~\ref{tab:parameters}.}
	\label{fig:wind_paths}
\end{figure}

\begin{remark}
	A key distinction from related works \cite{aid2016optimal,tan2018optimal,glas2020intraday} is that these studies model the production  process \(X_t\) as the evolution of successive forecast updates over time. In contrast, we model the actual power production and deliberately exclude forecast revisions within the trading horizon. While power producers often receive updated forecasts in practice, incorporating such updates would require the stochastic optimal control problem to be re-solved when new information becomes available. In \cite{tan2018optimal}, this formalism is mathematically consistent because analytical solutions are obtained under filtrations that explicitly account for the arrival of updated forecasts, and the control is adapted to that enlarged information flow. In a fully numerical setting, however, this approach cannot be implemented in a single backward dynamic programming solve because each forecast update changes the drift of the SDE, hence the infinitesimal generator and the associated HJB equation. A consistent numerical treatment would thus necessitate a rolling-horizon procedure in which the control problem is re-initialized and re-solved upon every forecast update, which is left for future investigation.
	
\end{remark}

\subsection{Data-Driven Intraday Price Dynamics}
\label{sec:price_dynamics}
Similarly to the wind power dynamics described in Section \ref{sec:wind_dynamics}, we model the intraday quoted  price process $Y = (Y_t)_{t \in [0,T]}$ as a real-valued stochastic process that mean-reverts around a deterministic forecast trajectory $(p_Y(t))_{t \in [0,T_{gc}]} \subset \mathbb{R}$. We refer to Section \ref{sec:data_description_parameters} for details on the construction of $p_Y(t)$. The variable $Y_t$ denotes the quoted forward price at time $t$ for delivery of one unit of electricity during the interval $[T_{gc} + h,\; T_{gc} + h + L]$\footnote{For notational simplicity, we omit the explicit dependence of $Y_t$ on $(T,h,L)$, since the control problem is formulated for a fixed delivery product.}, expressed in [EUR/MWh]. This formulation is consistent with the empirical findings of \cite{chatziandreou2026optimal}, who report that quoted prices corresponding to different delivery periods follow distinct stochastic dynamics, motivating the separate modeling of each delivery product.

Empirical analyses of European intraday markets indicate that price returns display both heavy tails and discrete jumps \cite{han2022complexity}. These abrupt movements are typically linked to forecast errors in renewable generation, liquidity shocks, or operational incidents such as grid contingencies and plant outages. Moreover, the jump distribution is empirically asymmetric, for instance, upward spikes may occur less frequently but can be of larger magnitude than downward ones \cite{kostrzewski2021impact}. To reproduce these stylized features within a tractable framework, we enrich the mean-reverting diffusion model with a compensated compound Poisson process whose jump sizes follow the asymmetric double-exponential law, proposed in financial mathematics literature by \cite{kou2002jump}. This specification jointly captures leptokurtic price dynamics and the asymmetry observed in electricity price data.

{Similarly to the wind production process, the intraday price $Y_t$
is assumed to mean-revert toward the forecast curve $p_Y(t)$, obtained
using a shape-preserving cubic Hermite interpolation. Hence,
$p_Y\in C^1([0,T_{\mathrm{gc}}])$, and the derivative $\dot p_Y$ is
obtained directly from the interpolant.} The dynamics combine (i) a mean-reverting drift term, (ii) a constant volatility term as in \cite{tan2018optimal,aid2016optimal,glas2020intraday}, and (iii) the compensated jump component described above, yielding
\begin{equation}
	\label{eq:price_sde}
	\left\{
	\begin{array}{l}
		\mathrm{d}Y_t = \big(\dot p_Y(t) - \kappa(Y_t - p_Y(t)) \big)\,\mathrm{d}t 
		+ \sigma\,\mathrm{d}B_t^Y + \mathrm{d}J_t, \quad t \in [0,T_{gc}],\\[4pt]
		Y_t = Y_{T_{gc}}, \quad t \in [T_{gc},\,T],\\[4pt]
		Y_0 = p_Y(0),
	\end{array}
	\right.
\end{equation}
where $\kappa > 0$ denotes the mean-reversion rate, $\sigma > 0$ the volatility, and $B^Y=(B_t^Y)_{t\in[0,T]}$ a standard $\mathbb{F}$-Brownian motion, possibly correlated with $B^X$, with $\mathrm{d}\langle B^X, B^Y\rangle_t = \rho\,\mathrm{d}t$ for some $\rho\in[-1,1]$, typically $\rho<0$ \cite{tan2018optimal}. The negative correlation reflects the merit-order effect whereby higher renewable output decreases prices.

\begin{remark}
	The process $Y_t$ has a physical interpretation only over the trading horizon $[0,T_{gc}]$, during which the contract for delivery on $[T_{gc} + h,\; T_{gc} + h + L]$ remains tradable. Beyond $T_{gc}$, the dynamics in \eqref{eq:price_sde} are extended  to ensure a well-posed stochastic control formulation (see Section~\ref{sec:SOC_formulation}), for $t \geq T_{gc}$, $Y_t$ neither influences the control nor enters the objective functional (see Equation \eqref{eq:J_extended}).
\end{remark}

The jump component in \eqref{eq:price_sde} is defined as a compensated compound Poisson process,
\begin{equation}
	\label{eq:jump_process}
	J_t = \sum_{i=1}^{{\mathcal N_t}} Z_i - \lambda\,\mathbb{E}[Z]\,t,
\end{equation}
where ${\mathcal N_t}$ is a Poisson process with intensity $\lambda>0$, and $\{Z_i\}_{i\ge1}$ are i.i.d.\ jump sizes with density
$$
f_Z(z)=
\begin{cases}
	p_+\,\eta_+ e^{-\eta_+ z}, & z>0,\\[4pt]
	p_-\,\eta_- e^{\eta_- z}, & z<0,
\end{cases}
\qquad p_+ + p_- = 1,\quad \eta_+,\eta_- > 0.
$$
The associated compensated Poisson random measure $\tilde{N} (\cdot,\cdot)$ is related to the Poisson random measure $N(\cdot,\cdot)$ via the Lévy measure $\nu(\cdot)$ as follows
\begin{equation}
	\label{eq:compensated_measure}
	\widetilde N(\mathrm{d}t,\mathrm{d}z)
	= N(\mathrm{d}t,\mathrm{d}z) - \nu(\mathrm{d}z)\,\mathrm{d}t,
	\qquad \nu(\mathrm{d}z) = \lambda f_Z(z)\,\mathrm{d}z.
\end{equation}
This construction ensures the forecast-tracking condition i.e.,  $\mathbb{E}[Y_t]=p_Y(t)$ for all $t\in[0,T_{gc}]$ (see \cite{caballero2021quantifying}), while reproducing the empirical jump characteristics of intraday electricity prices reported in \cite{han2022complexity,kostrzewski2021impact}.

Figure~\ref{fig:price_paths} illustrates the price model for the same
day (2024-04-11): the day-ahead hourly price, the smoothed forecast
$p_Y(t)$, and simulated paths from the 
jump-diffusion SDE in \eqref{eq:price_sde}. The construction of the deterministic forecast trajectory
$p_Y(t)$ from the available market data is described in detail in
Section~\ref{sec:data_description_parameters}, and the used model parameters
are given in Table~\ref{tab:parameters}.
\FloatBarrier
\begin{figure}[!ht]
	\centering
	\includegraphics[width=0.55\linewidth]{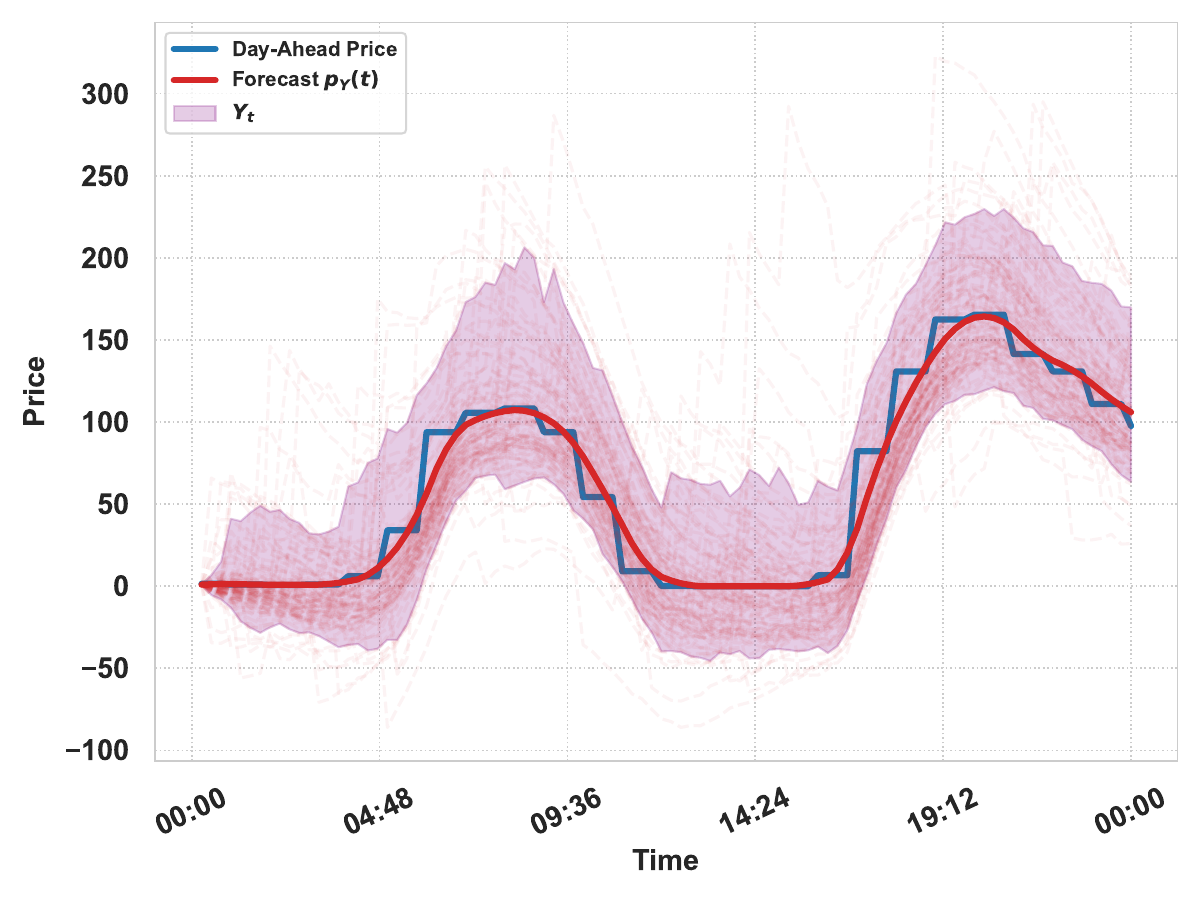}
	\caption{Day-ahead price, smoothed forecast $p_Y(t)$, and simulated
		paths from the calibrated jump-diffusion~\eqref{eq:price_sde}
		on 2024-04-11 in the Amprion control zone, Germany. Parameters
		as in Table~\ref{tab:parameters}.}
	\label{fig:price_paths}
\end{figure}
\FloatBarrier
\begin{remark}[Market Structure and Trading Horizon of the Price Process \(Y_t\)]
	In this work, prices in the intraday continuous market arise from a limit order book (LOB) containing multiple bids and asks for the same delivery period rather than a single quoted price \cite{graf2021modeling}. Hence, the stochastic price process \((Y_t)_{t\in[0,T_{gc}]}\) in Equation \eqref{eq:price_sde} is interpreted as the mid-price (the average of best bid and best ask), following  the standard  LOB modeling as in \cite{glas2020intraday}. Moreover,  we note that $Y_t$ models the evolution of the price for one fixed delivery product. In all numerical experiments, we use $h=5$ minutes for the German product studied here, matching the local EPEX SPOT rule set documented for the relevant product and data period in~\cite{epex_spot_trading}.  However, the proposed framework accommodates other product specifications by simply adjusting the values of $h$ and $L$.

\end{remark}

\paragraph{Notation.}
For the remainder of the paper, we define the following notation for the drift and diffusion coefficients of the aforementioned SDEs\footnote{We remark that $\mu_Y$ includes the compensation of the jump component.}:

$$
\mu_X(t,x)=\dot p_X(t)-\theta(t)\big(x-p_X(t)\big),\qquad
\sigma_X(x)=\sqrt{2\alpha\theta_0\,x(1-x)},
$$
$$
\mu_Y(t,y)=\dot p_Y(t)-\kappa\big(y-p_Y(t)\big)-\lambda\,\mathbb E[Z],\qquad
\sigma_Y=\sigma.
$$

\subsection{Inventory Dynamics}
We denote by $Q=\left(Q_t\right)_{t \in[0, T]}$ the inventory process, representing the cumulative amount of energy sold or purchased in [MWh] up to time $t \leq T_{\mathrm{gc}}$ for delivery during $[T_{\mathrm{gc}}+h,\, T_{\mathrm{gc}}+h+L]$. A positive value $Q_t>0$ indicates a net selling position, while $Q_t<0$ corresponds to a net buying position.

The inventory evolves according to
$$
\left\{\begin{array}{l}
	\mathrm{d} Q_t=\psi_t\,\mathrm{d} t, \quad t \in[0, T], \\
	Q_0=0 ,
\end{array}\right.
$$

{where the trading rate $\psi$, measured in $[\mathrm{MWh}/\mathrm{h}]$,
belongs to set of admissible controls, $\mathcal{A}_{0,T}$, which
 we define for $t\in[0,T]$,
{
\begin{equation}
\label{eq:admissible_control_set}
\mathcal{A}_{t,T}
:=
\left\{
\psi=(\psi_s)_{s\in[t,T]}:
\psi \text{ is } \mathbb{F}\text{-progressively measurable and takes values in }
[\psi_{\min},\psi_{\max}]
\right\}.
\end{equation}
}
The bounds $\psi_{\min},\psi_{\max}\in\mathbb{R}$ are discussed in
Section~\ref{sec:domain_truncation}. Since trading stops at
$T_{\mathrm{gc}}$, we impose $\psi_t=0$ for all
$t\in(T_{\mathrm{gc}},T]$, $\mathbb{P}$-a.s. Consequently, the inventory
remains constant thereafter, i.e.,
$Q_t=Q_{T_{\mathrm{gc}}}$ for all $t\in[T_{\mathrm{gc}},T]$.}

\section{Stochastic Optimal Control Formulation of the Trading Problem}
\label{sec:SOC_formulation}

\subsection{Cost Functional and Value Function}

In most existing works \cite{aid2016optimal,tan2018optimal,glas2020intraday}, the terminal imbalance penalty is applied in a simplified way at a single delivery time based on the instantaneous power production. However, in practice TSOs impose imbalance charges on the energy delivered over an imbalance settlement period (ISP), typically 15 or 30 minutes on European markets \cite{brusco2025rolling}. To reflect this operational feature, we propose a path-dependent imbalance penalty based on cumulative metered energy over $[T_{\mathrm{gc}}+h, T_{\mathrm{gc}}+h + L]$.

\paragraph{Standard formulation.}
The standard way to define the electricity trading problem  \cite{aid2016optimal,tan2018optimal,glas2020intraday} is to include the gains/losses from trading activity, liquidity/execution costs accumulated on this window and a terminal imbalance penalty at a single time point at gate closure $T_{\mathrm{gc}}$. {For an admissible control $\psi\in\mathcal{A}_{t,T}$} and the state vector $(Q_t,X_t,Y_t)=(q,x,y)$, the cost functional is given by
\begin{align}
	\label{eq:J_classic}
	J^{\mathrm{std}}(t,x,y,q;\psi)
	:=
	\mathbb{E}\Big[
	-\!\int_{t}^{T_{gc}}\psi_s\,Y_s\,\mathrm{d}s
	+\frac{\gamma}{2}\!\int_{t}^{T_{gc}}\psi_s^{2}\,\mathrm{d}s
	+\,g\!\big(Q_{T_{gc}}-P_{\max} X_{T_{gc}}\big)
	\ \Big|\ Q_t=q,\ X_t=x,\ Y_t=y\Big],
\end{align}
where:
\begin{itemize}
	\item $-\int_{t}^{T_{gc}}\psi_s\,Y_s\,\mathrm{d}s$ represents the cumulative trading cash flow over $[t,T_{gc}]$, with the sign convention $\psi_s>0$ for net sales and $\psi_s<0$ for net purchases.
	\item The  term $\frac{\gamma}{2}\int_{t}^{T_{gc}}\psi_s^{2}\,\mathrm{d}s$ models the temporary market impact, with $\gamma>0$ scaling liquidity frictions and penalizing aggressive trading rates. This is a modeling choice aligned with \cite{tan2018optimal}, rather than incorporating the control in the price SDE  \eqref{eq:price_sde} as in \cite{aid2016optimal,glas2020intraday}.
\item The terminal penalty term
$g\!\left(Q_{T_{\mathrm{gc}}}-P_{\max}X_{T_{\mathrm{gc}}}\right)$
is an imbalance penalty at gate closure, comparing the inventory fixed at gate closure
$Q_{T_{\mathrm{gc}}}$ with instantaneous normalized production
$X_{T_{\mathrm{gc}}}$ scaled by maximum capacity $P_{\max}$.
A typical choice for $g$ is a quadratic function. For instance,
\cite{tan2018optimal,aid2016optimal} consider a symmetric quadratic penalty
$g(\xi)=\frac{\beta}{2}\xi^2$ with $\beta>0$, $\xi\in\mathbb{R}$,
penalizing over- and underdelivery equally. This choice is driven by
analytical tractability: the quadratic structure ensures that the value function
is quadratic in the state variables, yielding a linear optimal control via a
Riccati ansatz and hence a closed-form solution. However, the quadratic penalty
has no clear empirical basis in electricity markets, where TSOs levy imbalance
charges linearly per unit of energy deviation. By contrast,
\cite{glas2020intraday} construct a data-driven penalty calibrated from
limit-order-book data. In this work, we adopt the two-sided linear penalty
$g(\xi)=\beta|\xi|$, which penalizes both shortfall and surplus proportionally
to the absolute deviation. This provides a stylized representation of imbalance
settlement, capturing its economic effect through a symmetric linear
deviation penalty. Since the resulting HJB equation is solved numerically, no
restriction to quadratic penalties is required. In practice, $\beta$ varies
across delivery periods depending on overall grid balancing costs.
\end{itemize}

\begin{remark}[Initial Inventory]
	In practice, a pre-committed day-ahead position may serve as the initial inventory, that is, $Q_0 = q_0 \neq 0$. This setting can transform the problem into an intraday adjustment of a previously optimized day-ahead bidding position, as in \cite{helgren2024feature}.
\end{remark}

\paragraph{Proposed extended formulation.}
In reality, imbalance settlement depends on  energy, that is to say, the cumulative energy delivered on $[T-L,T]$. Therefore, we replace the pointwise penalty by an energy-based penalty and augment the state space of the optimal control problem as follows
\begin{align}
	\label{eq:J_extended}
	J(t,x,y,q,m;\psi)
	:=
	\mathbb{E}\Big[
	&-\!\int_{t}^{T_{gc}}\psi_s Y_s\,\mathrm{d}s
	+\frac{\gamma}{2}\!\int_{t}^{ T_{gc}}\psi_s^2\,\mathrm{d}s \notag\\
	&\hspace{1.8em}
	+\,g\!\big(Q_{T}-P_{\max} \underbrace{\int_{T-L}^{T}X_s\,\mathrm{d}s }_{:=M_T}\big)
	\ \Big|\ X_t=x,\ Y_t=y,\ Q_t=q,\ M_t=m\Big],
\end{align}
where $\psi_s\equiv 0$ for $s>T_{gc}$, hence $Q_T=Q_{T_{gc}}$, and the state space is augmented compared to the standard formulation defining the metered energy during delivery by 
$$
\mathrm{d}M_t =
\begin{cases}
	0, & t \in [0,\,T-L),\\[2pt]
	X_t\,\mathrm{d}t, & t \in [T-L,\,T],
\end{cases}
\qquad M_0=0.
$$

The proposed extended formulation can thus be summarized in three stages as follows:
\begin{equation}
\label{eq:three_stage_schematic}
(X_t,\, Y_t,\, Q_t)
\xrightarrow[\text{Stage III}]{[0,\, T_{\mathrm{gc}}]}
(X_t,\, Q_{T_{\mathrm{gc}}})
\xrightarrow[\text{Stage II}]{[T_{\mathrm{gc}},\, T_{\mathrm{gc}}+h]}
(X_t,\, M_t)
\xrightarrow[\text{Stage I}]{[T-L,\, T]}
g(Q_{T_{\mathrm{gc}}} - P_{\max} M_T).
\end{equation}
The equations defining each stage are derived in Section~\ref{sec:three_stage_definition}.

\paragraph{Value function.}
The value function is defined by \eqref{eq:value_function_def}
\begin{equation}
	\label{eq:value_function_def}
	v(t,x,y,q,m):=\inf_{\psi\in\mathcal{A}_{t, T}} J(t,x,y,q,m;\psi),
\end{equation}

\subsection{Three-Stage KBE-HJB P(I)DEs }
\label{sec:three_stage_definition}

By the dynamic programming principle, the value function on the state space \((x,y,q,m)\) satisfies a stage-wise system of sequential backward HJB/Kolmogorov PIDEs on the intervals \([T\!-\!L,T]\), \((T_{gc},T\!-\!L]\), and \([0,T_{gc}]\).  For the remainder of this paper, we define the following notation for the partial derivatives of a function
$$\quad v_x:=\frac{\partial v}{\partial x}, \quad v_{x x}:=\frac{\partial^2 v}{\partial x^2}, \quad v_{xy}:=\frac{\partial^2 v}{\partial x\,\partial y}.$$

\paragraph{Stage I (delivery of energy).}
On this interval the control is inactive ($\psi\equiv0$) hence the inventory is known $Q_t = Q_{T_{gc}} = q $ for all $t \in [T-L,T]$,  by \eqref{eq:price_sde}, the price process is constant. Consequently, the cost functional can be written as follows for  $t\in[T-L,T]$
\begin{align}
	\label{eq:J_stageI}
	J^I(t,x,m;\psi)
	:=
	\mathbb{E}\Big[
	\,g\!\big(q-P_{\max} M_T\big)
	\ \Big|\ \ X_t=x,\ M_t=m\Big],
\end{align}
Hence the value function depends only on $(x,m)$ and solves a Kolmogorov backward PDE, parametrized in $q$ via the terminal imbalance condition, which is given as follows
\begin{equation}\label{eq:stageI_PDE}
	\begin{aligned}
		&	v_t^{\mathrm{I}}
		+ \mu_X(t,x)\,v_x^{\mathrm{I}}
		+ \tfrac12\,\sigma_X^2(x)\,v_{xx}^{\mathrm{I}}
		+ x\,v_m^{\mathrm{I}}
		= 0, 
		\qquad t\in[T-L,T),\\[4pt]
		&	v^{\mathrm{I}}(T,x,m;q)=g\big(q-P_{\max}m\big).
	\end{aligned}
\end{equation}

\paragraph{Stage II (lead time).}
During the lead time period, the control is inactive ($\psi\equiv0$), the metered energy state is identically zero ($M_t = M_{T-L}\equiv0$), and by \eqref{eq:price_sde} the price process $Y_t$ is constant ($Y_t = Y_{T_{gc}} \equiv y$), and the inventory is fixed ($Q_s\equiv q$). Hence the cost functional can be written using the tower property as follows for $t\in(T_{gc},T_{gc} + h)$
\begin{align*}
	J^{\mathrm{II}}(t,x;\,q)
	&= \mathbb{E}\!\left[
	g\!\big(Q_{T_{gc}}-P_{\max}M_T\big)
	\,\middle|\,X_t=x
	\right] \\[4pt]
	&= \mathbb{E}\!\left[
	\mathbb{E}\!\left[
	g\!\big(Q_{T_{gc}}-P_{\max}M_T\big)
	\,\middle|\,\mathcal F_{T-L}
	\right]
	\,\middle|\,X_t=x
	\right] \\[4pt]
	&= \mathbb{E}\!\left[
	v^{\mathrm I}(T\!-\!L,\,X_{T-L},\,0;\,q)
	\,\middle|\,X_t=x
	\right].
\end{align*}

Consequently, the value function  depends only on the variable $x$, but remains parametrized in the variable $q$ through the terminal condition of the stage I. Hence,  $v^{\mathrm{II}}$ is the solution to the following Kolmogorov backward PDE  
\begin{equation}\label{eq:stageII_KBE_corrected}
	\begin{aligned}
		& v_t^{\mathrm{II}}
		+ \mu_X(t,x)\,v_x^{\mathrm{II}}
		+ \tfrac12\,\sigma_X^2(x)\,v_{xx}^{\mathrm{II}}
		= 0, \qquad t\in(T_{gc},T_{gc} + h),\\[4pt]
		& v^{\mathrm{II}}(T_{gc} + h,x;\,q)=v^{\mathrm{I}}(T_{gc} + h,x,0;\,q).
	\end{aligned}
\end{equation}

\paragraph{Stage III (trading window).}
During trading stage, the state vector is $(x,y,q)$, and using the tower property, the cost functional can be written as follows for  $t\in[0,T_{gc}]$
\begin{align*}
	J^{\mathrm{III}}(t,x,y,q;\psi)
	&:= \mathbb{E}\!\left[
	-\!\int_{t}^{T_{gc}}\!\psi_s Y_s\,\mathrm{d}s
	+ \frac{\gamma}{2}\!\int_{t}^{T_{gc}}\!\psi_s^2\,\mathrm{d}s
	+ v^{\mathrm{II}}(T_{gc}, X_{T_{gc}};\,Q_{T_{gc}})
	\,\middle|\, X_t=x,\ Y_t=y,\ Q_t=q
	\right]
\end{align*}

Hence, the corresponding HJB PIDE for $v^{\mathrm{III}}(t,x,y,q)$ is given by
\begin{equation}\label{eq:stageIII_HJB_corrected}
	\begin{aligned}
		& v_t^{\mathrm{III}}
		+ \mu_X(t,x)\,v_x^{\mathrm{III}}
		+ \mu_Y(t,y)\,v_y^{\mathrm{III}}
		+ \tfrac12\,\sigma_X^2(x)\,v_{xx}^{\mathrm{III}}
		+ \rho\,\sigma_X(x)\,\sigma_Y\,v_{xy}^{\mathrm{III}}
		+ \tfrac12\,\sigma_Y^2\,v_{yy}^{\mathrm{III}} \\
		& 	+ H_q\!\big(y, v_q^{\mathrm{III}}\big) +\, I_y\!\big[v^{\mathrm{III}}\big](t,x,y,q)
		\;=\; 0, \qquad t\in[0,T_{gc}), \\[4pt]
		& v^{\mathrm{III}}(T_{gc},x,y, q) \;=\; v^{\mathrm{II}}(T_{gc},x;q),
	\end{aligned}
\end{equation}
where the nonlocal jump operator acting on the price coordinate is defined as

\begin{equation}
\label{eq:jump_operator_def}
    I_y[v^{\mathrm{III}}](t,x,y,q):=\int_{\mathbb{R}} \big(\ v^{\mathrm{III}} (t,x,y+\delta,q)-v^{\mathrm{III}}(t,x,y,q)\big)\,\lambda\,f_Z(\delta)\,\mathrm{d}\delta.
\end{equation}
and where the control-dependent part of the Hamiltonian is given by:
\begin{equation}
	\label{eq:control_Hamiltonian}
	H_q(y,p):=\inf_{\psi\in[\psi_{\min},\psi_{\max}]}\Big\{-\psi\,y+\tfrac{\gamma}{2}\psi^2+\psi\,p\Big\}.
\end{equation}
The unconstrained minimizer of \eqref{eq:control_Hamiltonian}  is given by
\begin{equation}
	\label{eq:unconstrained_minimizer}
	\psi^{0}(y,p):=\frac{y-p}{\gamma},
\end{equation}
and the projection of \eqref{eq:unconstrained_minimizer} onto $[\psi_{\min},\psi_{\max}]$ can be written as
\[
\psi^{*}(y,p) := \Pi_{[\psi_{\min},\psi_{\max}]}\!\Big( 	\psi^{0}(y,p) \Big) 
=\min\!\big\{\psi_{\max},\max\{\psi_{\min},\,\psi^{0}(y,p)\}\big\}.
\]
Hence, the Hamiltonian in \eqref{eq:control_Hamiltonian} admits an explicit representation, and $H_q$ reads as:
\begin{equation}
	\label{eq:hamiltoniain_piecewise}
	H_q(y,p)=
	\begin{cases}
		-\dfrac{(y-p)^2}{2\gamma}, & \text{if }\ \psi_{\min}\le \dfrac{y-p}{\gamma}\le \psi_{\max},\\[8pt]
		\dfrac{\gamma}{2}\,\psi_{\min}^2+\psi_{\min}\,(p-y), & \text{if }\ \dfrac{y-p}{\gamma}<\psi_{\min},\\[8pt]
		\dfrac{\gamma}{2}\,\psi_{\max}^2+\psi_{\max}\,(p-y), & \text{if }\ \dfrac{y-p}{\gamma}>\psi_{\max}.
	\end{cases}
\end{equation}
The global value function is defined piecewise across the three stages of the trading problem by
$$
v(t,x,y,q,m)=
\begin{cases}
	v^{I}(t,x,m;q), & t\in[T_{\mathrm{gc}}+h,\, T],\\[4pt]
	v^{II}(t,x;q), & t\in[T_{\mathrm{gc}},\, T_{\mathrm{gc}}+h],\\[4pt]
	v^{III}(t,x,y,q), & t\in[0,\, T_{\mathrm{gc}}].
\end{cases}
$$


\subsection{Well-Posedness of the Regularized Stage~III HJB-PIDE}
\label{sec:well_posedness}

In this section we study the regularized Stage~III problem obtained by
replacing the Jacobi diffusion coefficient $\sigma_X$ with the
regularized coefficient $\sigma_X^\varepsilon$ defined below. The
motivation for this modification is that
$$
\sigma_X(x)=\sqrt{2\alpha\theta_0\,x(1-x)}
$$
is not globally Lipschitz on $[0,1]$. Consequently, the original Stage~III
problem does not fit directly into the viscosity solution framework of
\cite{pham1998optimal}. The regularization restores global Lipschitz
continuity of the diffusion coefficient and allows us to apply that
framework to the regularized problem. The regularized coefficient
$\sigma_X^\varepsilon$ is defined by
\begin{equation}
  \label{eq:regularized_diffusion}
  \sigma_X^\varepsilon(x)
  :=
  \begin{cases}
    \dfrac{\sigma_X(\varepsilon)}{\varepsilon}\,x,
      & 0\le x<\varepsilon,\\[6pt]
    \sigma_X(x),
      & \varepsilon\le x\le 1-\varepsilon,\\[6pt]
    \dfrac{\sigma_X(1-\varepsilon)}{\varepsilon}\,(1-x),
      & 1-\varepsilon<x\le 1,
  \end{cases}
\end{equation}
{for a fixed $\varepsilon\in(0,1/2)$.} By construction, $\sigma_X^\varepsilon$
is continuous on $[0,1]$, vanishes at $x\in\{0,1\}$, and is equal to
$\sigma_X$ on $[\varepsilon,1-\varepsilon]$. Moreover,
$\sigma_X^\varepsilon$ is globally Lipschitz on $[0,1]$. In the regularized setting, we denote by
$v^{\mathrm{I},\varepsilon}$ and $v^{\mathrm{II},\varepsilon}$ the
Stage I and Stage II value functions obtained by replacing
$\sigma_X$ with $\sigma_X^\varepsilon$ in the Stage I and Stage II
Kolmogorov backward equations \eqref{eq:stageI_PDE} and
\eqref{eq:stageII_KBE_corrected}, respectively, and by
$v^{\mathrm{III},\varepsilon}$ the value function of the corresponding
regularized Stage III control problem. In particular, the regularized
Stage III problem is posed with terminal datum
$v^{\mathrm{II},\varepsilon}(T_{\mathrm{gc}},x;\,q)$. The numerical solver is introduced in Section \ref{sec:numerical_methodology}.
{Illustrations of the regularization are provided in
Appendix~\ref{app:regularization}.}

\begin{assumption}[Lipschitz continuity of coefficients]
\label{ass:lipschitz}
There exists $K_1 > 0$ such that for all $t \in [0, T_{\mathrm{gc}}]$,
$x, x' \in [0,1]$, and $y, y' \in \mathbb{R}$:
\begin{enumerate}[label=\textup{(\roman*)}, ref=\ref*{ass:lipschitz}(\roman*)]
\item\label{ass:lipschitz_muX}
$|\mu_X(t,x) - \mu_X(t,x')| \le K_1\,|x - x'|$;
\item\label{ass:lipschitz_muY}
$|\mu_Y(t,y) - \mu_Y(t,y')| \le K_1\,|y - y'|$;
\item\label{ass:lipschitz_sigmaX}
$|\sigma_X^\varepsilon(x) - \sigma_X^\varepsilon(x')| \le K_1\,|x - x'|$.
\end{enumerate}
\end{assumption}

\begin{assumption}[Continuity of coefficients]
\label{ass:continuity}
The functions $\mu_X(t,x)$, $\mu_Y(t,y)$, $\sigma_X^\varepsilon(x)$
are continuous in their arguments.
\end{assumption}

\begin{assumption}[Jump amplitude and L\'evy measure]
\label{ass:levy}
Let $\boldsymbol{\xi} = (x, y, q) \in [0, 1] \times \mathbb{R}^2$  and $\psi \in \Psi = [\psi_{\min}, \psi_{\max}]$.
\begin{enumerate}[label=\textup{(\roman*)}, ref=\ref*{ass:levy}(\roman*)]
\item\label{ass:levy_amplitude}
The jump amplitude $\Gamma(z) := (0, z, 0)^\top$ satisfies
$|\Gamma(z)| = |z|$ and does not depend on $(t, \boldsymbol{\xi}, \psi)$.
\item\label{ass:levy_measure}
The L\'evy measure $\nu(\mathrm{d}z) = \lambda\,f_Z(z)\,\mathrm{d}z$
defined in~\eqref{eq:compensated_measure} has finite total mass and
finite second moment:
\begin{equation}\label{eq:levy_finiteness}
  \nu(\mathbb{R}) = \lambda < \infty,
  \qquad
  \int_{\mathbb{R}} |z|^2\,\nu(\mathrm{d}z)
  = \lambda\,\mathbb{E}[Z^2] < \infty.
\end{equation}
Moreover, with $\varrho(z) := |z|$,
\begin{equation}\label{eq:jump_amplitude_bound}
  |\Gamma(z)| \leq \varrho(z),
  \qquad
  \int_{\mathbb{R}} \varrho^2(z)\,\nu(\mathrm{d}z) < \infty.
\end{equation}
\end{enumerate}
\end{assumption}

\begin{assumption}[Lipschitz regularity of running and terminal costs]
\label{ass:cost_lip}
Let $\boldsymbol{\xi} = (x, y, q) \in [0, 1] \times \mathbb{R}^2$ and $\psi \in \Psi = [\psi_{\min}, \psi_{\max}]$.
There exists $K_2 > 0$ such that for all $t \in [0, T_{\mathrm{gc}}]$:
\begin{enumerate}[label=\textup{(\roman*)}, ref=\ref*{ass:cost_lip}(\roman*)]
\item\label{ass:cost_lip_running}
the running cost $\ell(y, \psi) := -\psi\,y + \tfrac{\gamma}{2}\psi^2$ satisfies
$$
|\ell(y, \psi) - \ell(y', \psi)| \le K_2\,|y - y'|;
$$
\item\label{ass:cost_lip_terminal}
the terminal cost $v^{\mathrm{II},\varepsilon}(T_{\mathrm{gc}}, x;\, q)$ satisfies
$$
|v^{\mathrm{II},\varepsilon}(T_{\mathrm{gc}}, x;\, q)
- v^{\mathrm{II},\varepsilon}(T_{\mathrm{gc}}, x';\, q')|
\le K_2\,(|x - x'| + |q - q'|).
$$
\end{enumerate}
\end{assumption}

\begin{lemma}[Lipschitz continuity of the regularized Stage~III terminal cost]
\label{lem:stageII_terminal_lipschitz}
{Let $\varepsilon\in(0,1/2)$ be fixed,} and assume that Assumptions~\ref{ass:lipschitz_muX} and~\ref{ass:lipschitz_sigmaX} hold.
Then there exists a constant $C>0$ such that
$$
|v^{\mathrm{II},\varepsilon}(T_{\mathrm{gc}},x;\,q)
-
v^{\mathrm{II},\varepsilon}(T_{\mathrm{gc}},x';\,q')|
\le
C\,(|x-x'|+|q-q'|)
$$
for all $x,x'\in[0,1]$ and all $q,q'\in\mathbb{R}$.
\end{lemma}
\begin{proof}
The proof is given in Appendix~\ref{app:proof_stageII_terminal_lipschitz}.
\end{proof}

\begin{theorem}[Viscosity characterization of the regularized problem]
\label{thm:viscosity_reg}
{Let $\varepsilon\in(0,1/2)$ be fixed.} Under
Assumptions~\ref{ass:lipschitz}-\ref{ass:cost_lip}, the value function
$v^{\mathrm{III},\varepsilon}$
of the regularized Stage III control problem is the unique continuous
viscosity solution, in the class of continuous functions with at most
linear growth that are uniformly continuous in $(x,y,q)$ uniformly in
$t$, of the HJB-PIDE
\begin{equation}\label{eq:HJB_reg}
\begin{aligned}
  &\partial_t v^{\mathrm{III},\varepsilon}
  + \mu_X(t,x)\,v^{\mathrm{III},\varepsilon}_x
  + \mu_Y(t,y)\,v^{\mathrm{III},\varepsilon}_y
  + \tfrac{1}{2}\bigl(\sigma_X^\varepsilon(x)\bigr)^2\,
    v^{\mathrm{III},\varepsilon}_{xx}
  + \rho\,\sigma_X^\varepsilon(x)\,\sigma\,v^{\mathrm{III},\varepsilon}_{xy}
  + \tfrac{1}{2}\sigma^2\,v^{\mathrm{III},\varepsilon}_{yy}\\[4pt]
  &\quad
  + H_q\!\bigl(y,\,v^{\mathrm{III},\varepsilon}_q\bigr)
  +\mathcal{I}_y[v^{\mathrm{III},\varepsilon}](t,x,y,q)
  = 0,
  \qquad (t,x,y,q)\in[0,T_{\mathrm{gc}})\times[0,1]
    \times\mathbb{R}^2,
  \\[6pt]
  &v^{\mathrm{III},\varepsilon}(T_{\mathrm{gc}},x,y,q)
  =v^{\mathrm{II},\varepsilon}(T_{\mathrm{gc}},x;\,q),
\end{aligned}
\end{equation}
where $H_q$ is defined in~\eqref{eq:control_Hamiltonian} and
$\mathcal{I}_y[\cdot]$ is the nonlocal operator defined
in~\eqref{eq:jump_operator_def}.
\end{theorem}

\begin{proof}
{The verification that Assumptions~3.2--3.5 are satisfied is provided in
Appendix~B. The proof of the theorem then follows from
\cite{pham1998optimal}.}
\end{proof}

{
We next establish convergence of the regularized value functions to
their counterparts under the original production dynamics \eqref{eq:wind_sde}. For
$0\leq t\leq s\leq T$ and $x\in[0,1]$, let $X_s^{t,x}$ and
$X_s^{\varepsilon,t,x}$ denote the original and regularized production
processes, respectively, driven by the same Brownian motion.

\begin{lemma}[Convergence of the regularized production process]
\label{lem:production_process_convergence}
Under~\eqref{eq:mean_reversion_condition}, we have
$$
\sup_{\substack{0\leq t\leq s\leq T\\x\in[0,1]}}
\mathbb{E}\left[
\left|X_s^{\varepsilon,t,x}-X_s^{t,x}\right|
\right]
\longrightarrow0
\qquad\text{as }\varepsilon\downarrow0.
$$
\end{lemma}

\begin{proof}
The proof is given in
Appendix~\ref{app:proof_production_process_convergence}.
\end{proof}
We next combine Lemma~\ref{lem:production_process_convergence} with the
conditional expectation representations for Stages I and~II and Stage~III to prove
Proposition~\ref{prop:value_function_regularization}, which establishes
uniform convergence of the regularized value functions to the original value functions in all three
stages.

\begin{proposition}[Convergence of the regularized value functions]
\label{prop:value_function_regularization}
The value functions of the three stages satisfy
$$
\sup_{\substack{t\in[T-L,T],\ x\in[0,1]\\
                  m\in[0,L],\ q\in\mathbb{R}}}
\left|v^{\mathrm{I},\varepsilon}(t,x,m;\,q)
-v^{\mathrm{I}}(t,x,m;\,q)\right|
\longrightarrow0,
$$
$$
\sup_{\substack{t\in[T_{\mathrm{gc}},T-L],\ x\in[0,1]\\
                  q\in\mathbb{R}}}
\left|v^{\mathrm{II},\varepsilon}(t,x;\,q)
-v^{\mathrm{II}}(t,x;\,q)\right|
\longrightarrow0,
$$
and
$$
\sup_{\substack{t\in[0,T_{\mathrm{gc}}],\ x\in[0,1]\\
                  (y,q)\in\mathbb{R}^2}}
\left|v^{\mathrm{III},\varepsilon}(t,x,y,q)
-v^{\mathrm{III}}(t,x,y,q)\right|
\longrightarrow0
$$
as $\varepsilon\downarrow0$.
\end{proposition}

\begin{proof}
The proof is given in
Appendix~\ref{app:proof_value_function_regularization}.
\end{proof}
}

\begin{remark}[Scope of the well-posedness result]
\label{rem:scope}
Theorem~\ref{thm:viscosity_reg} establishes well-posedness  for
the regularized problem. The original Jacobi
coefficient~$\sigma_X$ does not satisfy
Assumption \ref{ass:lipschitz_sigmaX}. The
underlying SDE is nevertheless well posed
under \eqref{eq:mean_reversion_condition} (see \cite{caballero2021quantifying}), and the regularization
modifies the dynamics only in the boundary layers
$[0,\varepsilon)\cup(1{-}\varepsilon,1]$.
{Proposition~\ref{prop:value_function_regularization}
establishes uniform convergence of the regularized value functions to
their original counterparts as $\varepsilon\downarrow0$.}
Moreover, the value function corresponding
to the Stage III HJB-PIDE for the non-regularized problem is visualized
in Section \ref{sec:value_function_structure}.
Consequently, the numerical scheme developed in
Section~\ref{sec:numerical_methodology} and the experiments reported
in Section~\ref{sec:num_exp_results} consider the non-regularized
Stage~III HJB-PIDE. Theorem~\ref{thm:viscosity_reg} therefore serves
as a well-posedness result for the regularized approximate problem.
{Uniqueness of $v^{\mathrm{III}}$ as a viscosity solution
of the original degenerate HJB-PIDE is left for future work.}
\end{remark}

\section{Numerical Scheme and Analysis}

\label{sec:numerical_methodology}
The SOC problem is formulated as a sequence of two linear Kolmogorov PDEs in Stages~I and II and a fully nonlinear four-dimensional HJB PIDE in Stage~III, all of which are solved numerically. Since the numerical methodology required for the linear PDEs is contained within that of the Stage~III HJB PIDE, we focus on the Stage~III finite difference discretization, the schemes for Stages~I and II are obtained by a direct simplification. {For the Stage~III time discretization, we use the reversed time variable $\tau=T_{\mathrm{gc}}-t$, so the terminal condition at $t=T_{\mathrm{gc}}$ corresponds to $\tau=0$ and the numerical scheme advances forward in $\tau$.}

The numerical scheme is guided by the monotone, stable, and consistent approximation principle of \cite{barles1991convergence} for fully nonlinear viscosity problems. For the nonlocal jump operator, the discretization builds on the finite difference framework for jump-diffusion PIDEs of \cite{cont2005finite,wang2008maximal}, specialized to the double-exponential jump structure via the differential formulation of \cite{carr2007numerical}. At the continuous level, the Stage~III PIDE is interpreted within the viscosity-solution framework for controlled jump-diffusions of \cite{pham1998optimal}, which underlies the well-posedness analysis of Stage III HJB-PIDE.

Section~\ref{sec:numerical_methodology} is organized as follows. Section~\ref{sec:domain_truncation} specifies the truncated computational domains and the spatial grid discretization, while Section~\ref{sec:boundary_conditions} describes the numerical boundary conditions. Section~\ref{sec:monotone_scheme} introduces the monotone finite difference discretization, with Section \ref{sec:semi_implicit_Hamiltonian} describing the semi-implicit scheme for the nonlinear part of the Hamiltonian and Section~\ref{sec:xy_implicit_discretization} detailing the implicit scheme for drift and diffusion terms in price and wind power state variables. Section~\ref{sec:jump_discretization} is devoted to the explicit discretization of the nonlocal jump operator.

\subsection{{Domain Truncation}}
\label{sec:domain_truncation}
In this section, we first specify the truncation bounds defining the computational domains on which the three-stage equations are solved. 
Then, we introduce the spatial grid discretization.

\paragraph{Probabilistic truncation of the price domain}

We derive probabilistic bounds on the forecast error
between the quoted price process $Y_t$ and its forecast trajectory $p_Y(\cdot)$.

\begin{proposition}[Deviation bound for $Y_t$]\label{prop:price_bound}
Let $D_t := Y_t - p_Y(t)$ denote the deviation of the price process
from its forecast, then $D_t$ satisfies  for $t\in[0,T_{\mathrm{gc}}]$ 
\[
\mathrm{d}D_t = -\kappa\,D_t\,\mathrm{d}t
+ \sigma\,\mathrm{d}B_t^Y + \mathrm{d}J_t,
\qquad D_0 = 0,
\]
and admits the decomposition $D_t = U_t + \widetilde{J}_t$, where
\[
U_t := \sigma\!\int_0^t e^{-\kappa(t-s)}\,\mathrm{d}B_s^Y,
\qquad
\widetilde{J}_t := \int_0^t e^{-\kappa(t-s)}\,\mathrm{d}J_s.
\]

For any tolerance $\varepsilon \in (0,1)$ and parameters
$\alpha_+ \in (0, \eta_+)$, $\alpha_- \in (0, \eta_-)$, let
$$
K_U(\varepsilon)
:= m_U + \sigma_{\sup}\,\sqrt{2\ln\!\bigl(\tfrac{4}{\varepsilon}\bigr)},
$$
where
$$
m_U := \mathbb{E}\!\left[\sup_{0 \le t \le T_{\mathrm{gc}}} U_t\right],
\qquad
\sigma_{\sup}^2
:= \sup_{0 \le t \le T_{\mathrm{gc}}} \mathrm{Var}[U_t]
= \frac{\sigma^2}{2\kappa}\bigl(1 - e^{-2\kappa T_{\mathrm{gc}}}\bigr),
$$
and
$$
K_J(\varepsilon)
:= 2\max\!\left\{
\frac{\Lambda_c(\alpha_+)\,T_{\mathrm{gc}}
      - \ln(\varepsilon/4)}{\alpha_+},\;
\frac{\Lambda_c(-\alpha_-)\,T_{\mathrm{gc}}
      - \ln(\varepsilon/4)}{\alpha_-}
\right\},
$$
with the compensated cumulant generating function given by
$$
\Lambda_c(\alpha)
:= \lambda\bigl(\mathbb{E}[e^{\alpha Z}] - 1
    - \alpha\,\mathbb{E}[Z]\bigr),
\qquad \alpha \in (-\eta_-, \eta_+).
$$

Then, with $K(\varepsilon) := K_U(\varepsilon) + K_J(\varepsilon)$,
$$
\mathbb{P}\!\left(
\sup_{0 \le t \le T_{\mathrm{gc}}} |D_t| > K(\varepsilon)
\right)
\le \varepsilon.
$$
\end{proposition}

\begin{proof}
{The proof is given in Appendix~\ref{sec:proof_prob_bound_Y}.}
\end{proof}

Consequently, for fixed $ 0 <\varepsilon \ll  1$, we define the upper and lower bounds for the price variable as follows
$$
y_{\min}:=\min_{t\in[0,T_{\mathrm{gc}}]} p_Y(t)-K(\varepsilon), 
\qquad
y_{\max}:=\max_{t\in[0,T_{\mathrm{gc}}]} p_Y(t)+K(\varepsilon).
$$

\begin{remark}
In the implementation, $m_U$ is estimated via Monte Carlo simulation
of the process $U_t$ on $[0,T_{\mathrm{gc}}]$.
\end{remark}

\paragraph{The domain of $Q_t$.}
Since no physical constraints are imposed on the inventory variable $Q_t$, we
define the computational interval $[q_{\min},q_{\max}]$ for $Q_t$ by means of
the method of characteristics applied to an auxiliary reduced first-order
equation in the inventory variable $q$.

This yields, for any $\epsilon>0$,
\begin{equation}
\label{eq:q_bounds}
q_{\min} := T_{\mathrm{gc}} \min\left\{\frac{y_{\min}-\beta}{\gamma},0\right\} -\epsilon, \qquad q_{\max} := T_{\mathrm{gc}}
\max\left\{\frac{y_{\max}+\beta}{\gamma},0\right\}
+\epsilon.
\end{equation}

The detailed steps to derive \eqref{eq:q_bounds} are provided in
Appendix~\ref{sec:characteristics}.

\paragraph{The domain of $\psi_t$.}
{Since there are no exogenous restrictions on the trading rate, we set}
\begin{equation}
\label{eq:psi_bounds}
\psi_{\min}:=\frac{q_{\min}}{T_{\mathrm{gc}}},
\qquad
\psi_{\max}:=\frac{q_{\max}}{T_{\mathrm{gc}}}.
\end{equation}
{where $[q_{\min},q_{\max}]$ is determined by the method of
characteristics. This is a modeling choice, and the admissible trading rate
bounds may instead be chosen independently of the computational inventory
domain.}

\paragraph{Grid discretization.}
We introduce a uniform  grid over the time and state variables.
{The three stages use separate uniform time grids. For Stage~III, the reversed interval $\tau\in[0,T_{\mathrm{gc}}]$ is discretized as
$$
\Delta\tau=\frac{T_{\mathrm{gc}}}{N_t},
\qquad
\tau_n=n\Delta\tau,
\qquad n=0,\ldots,N_t.
$$
For Stage~II on $[T_{\mathrm{gc}},T-L]=[T_{\mathrm{gc}},T_{\mathrm{gc}}+h]$,
$$
\Delta t_{\mathrm{gc}}=\frac{h}{N_{\mathrm{gc}}},
\qquad
t_n^{\mathrm{gc}}=T_{\mathrm{gc}}+n\Delta t_{\mathrm{gc}},
\qquad n=0,\ldots,N_{\mathrm{gc}},
$$
and for Stage~I on $[T-L,T]$,
$$
\Delta t_{\mathrm d}=\frac{L}{N_{\mathrm d}},
\qquad
t_n^{\mathrm d}=T-L+n\Delta t_{\mathrm d},
\qquad n=0,\ldots,N_{\mathrm d}.
$$}
Each spatial variable
$\xi \in \{x, y, q, m\}$ is discretized uniformly on its respective
domain $[\xi_{\min}, \xi_{\max}]$ with $N_\xi + 1$ nodes, step size
$\Delta\xi = (\xi_{\max} - \xi_{\min})/N_\xi$, and grid points
$\xi_\ell = \xi_{\min} + \ell\,\Delta\xi$ for
$\ell = 0,\ldots,N_\xi$. For the production variable, $x \in [0,1]$, and
\begin{equation}
\label{eq:production_grid}
{x_i=i\Delta x,
\qquad
\Delta x=\frac{1}{N_x},
\qquad
i=0,\ldots,N_x.}
\end{equation}
{We denote the Stage~III numerical approximation by
$$
V^n_{i,j,k}\approx v^{III}(T_{\mathrm{gc}}-\tau_n,x_i,y_j,q_k),
\qquad n=0,\ldots,N_t.
$$
Thus $V^0$ contains the terminal data at $t=T_{\mathrm{gc}}$, while $V^{N_t}$ approximates $v^{III}(0,\cdot)$.}

\subsection{{Numerical Boundary Conditions}}
\label{sec:boundary_conditions}

{
We specify the numerical boundary conditions used at the boundary of the
computational domain.

\emph{Production variable.}
At the boundary of $x\in\{0,1\}$, the Jacobi coefficient satisfies
\begin{equation}
\label{eq:production_diffusion_boundary}
\sigma_X(0)=\sigma_X(1)=0.
\end{equation}
Hence the diffusion of $x$ and the mixed-derivative term
$\rho\sigma_X(x)\sigma v_{xy}$ vanish at $x=0$ and $x=1$. The boundary
rows in the $x$ direction are assembled from the remaining terms in the
$(x,y)$ equation, using one-sided upwind differences.
\cite{caballero2021quantifying} shows that the drift of $x$ points inward at the boundary i.e.,
\begin{equation} 
\label{eq:production_drift_boundary}
{\mu_X(t,0)\ge0},
\qquad
{\mu_X(t,1)\le0},
\qquad
{t\in[0,T_{\mathrm{gc}}]}.
\end{equation}
Condition \eqref{eq:production_drift_boundary} determines
the one-sided upwind direction used at $i=0$ and $i=N_x$
{on the production grid defined in \eqref{eq:production_grid}}. No exterior values
in the $x$ variable are needed.

\emph{Price variable.}
The grid endpoints $y_{\min}$ and $y_{\max}$ are chosen as in
Proposition~\ref{prop:price_bound}. We impose at these endpoints the
affine condition
\begin{equation}
\label{eq:price_farfield_condition}
{\partial_y v^{III}(T_{\mathrm{gc}}-\tau,x,y_{\max},q)} = a_+(\tau),
\qquad
{\partial_y v^{III}(T_{\mathrm{gc}}-\tau,x,y_{\min},q)} = a_-(\tau),
\end{equation}
where
\begin{equation}
\label{eq:price_farfield_coefficients}
a_+(\tau) = -\psi_{\max}\,\frac{1-e^{-\kappa\tau}}{\kappa},
\qquad
a_-(\tau) = -\psi_{\min}\,\frac{1-e^{-\kappa\tau}}{\kappa},
\end{equation}
with the convention $(1-e^{-\kappa\tau})/\kappa = \tau$ for $\kappa = 0$.
The expressions \eqref{eq:price_farfield_coefficients} are derived in
Appendix~\ref{app:farfield}.
\emph{Inventory variable.}
The bounds $[q_{\min},q_{\max}]$ in \eqref{eq:q_bounds} are chosen from
the characteristic bounds of the controlled inventory equation. The rows
at $q=q_{\min}$ and $q=q_{\max}$ are assembled with the upwind convention
used in the $q$-scheme, so no exterior values in the $q$ direction are
needed.

\emph{Metered energy variable.}
The variable $m$ appears only in the Stage I equation and satisfies
$m\in[0,L]$. The rows at $m=0$ and $m=L$ are assembled with the one-sided
upwind differences associated with the transport in $m$.}

\subsection{Monotone IMEX finite difference Scheme with Operator Splitting}
\label{sec:monotone_scheme}
To discretize the PIDEs, we adopt an implicit-explicit operator splitting strategy:  

\begin{itemize}
	\item The diffusion and drift terms in the $x,y,m$ variables are treated implicitly.  
	\item The $q$ variable, which introduces a fully nonlinear term in the Stage~III HJB-PIDE, is handled semi-implicitly by linearization of the Hamiltonian (see Section~\ref{sec:semi_implicit_Hamiltonian}).  
    \item {The mixed derivative term in $(x,y)$ is treated via an implicit positive-type generalized finite difference scheme inspired by \cite{bonnans2003consistency} (see Section \ref{sec:xy_implicit_discretization}).}
	\item The nonlocal jump operator is evaluated via a differential formulation that localizes the double-exponential integrals to the computational grid, avoiding both extrapolation beyond the grid boundaries and the formation of dense matrices (see Section~\ref{sec:jump_discretization}).
	\item In Stage~III, operator splitting is applied to the HJB-PIDE so that the update in the $q$-direction is performed first, followed by the implicit step in the $(x,y)$ variables at each time level.
\end{itemize}

The aforementioned choices yield a monotone discretization that remains stable and robust across a wide range of model parameters, in particular the imbalance penalty $\beta$ and the market impact parameter $\gamma$ (see Section \ref{sec:numerical_experiments}). {Treating the differential part of the PIDE fully explicitly would impose severe CFL restrictions on the admissible reversed-time step $\Delta\tau$, which are avoided by the proposed implicit and semi-implicit treatment.} In fact, as $\gamma \rightarrow 0$, the optimal control becomes unbounded, causing $\partial_p H_q(p, y)$ to become very large in magnitude even for finite but sufficiently small values of $\gamma$. In addition, large values of $\beta$ steepen the gradient values, $\partial_q v$, and hence magnify CFL-type restrictions in explicit schemes, which our semi-implicit treatment maneuvers effectively. {The only remaining time-step restriction is due to the explicit jump component and requires $\Delta\tau \leq 1/\lambda$; see Appendix \ref{app:monotonicity_jump}.}

To decouple the nonlinear and nonlocal components of the HJB PIDE, we employ a first-order Lie operator splitting combined with an IMEX splitting in reversed time. {Under $\tau=T_{\mathrm{gc}}-t$,
$$
\partial_\tau v^{III}(T_{\mathrm{gc}}-\tau,x,y,q)
=
\mathcal H_{xy}[v^{III}]
+
\mathcal H_q[v^{III}]
+
\mathcal I_y[v^{III}],
$$
where the operators on the right-hand side are evaluated at $(T_{\mathrm{gc}}-\tau,x,y,q)$. Starting from the terminal data $V^0$, the numerical scheme advances forward in reversed time for $n=0,\ldots,N_t-1$ according to
$$
V^n \xrightarrow[\text{Step 1}]{\mathcal H_q} U^{n+1}
\xrightarrow[\text{Step 2}]{\mathcal H_{xy}} W^{n+1}
\xrightarrow[\text{Step 3}]{\mathcal I_y} V^{n+1}.
$$}

\medskip
\noindent
\textbf{Step 1: $q$-substep (semi-implicit)}
$$
{U^{n+1}-\Delta\tau\,\mathcal H_q[U^{n+1}]=V^n.}
$$

\noindent
\textbf{Step 2: $(x,y)$-substep (implicit)}
$$
{W^{n+1}-\Delta\tau\,\mathcal H_{xy}[W^{n+1}]=U^{n+1}.}
$$

\noindent
\textbf{Step 3: jump substep in $y$ (explicit)}
$$
{V^{n+1}=W^{n+1}+\Delta\tau\,\mathcal I_y[W^{n+1}].}
$$

{The intermediate values $U^{n+1}$ and $W^{n+1}$ are associated with the update from $\tau_n$ to $\tau_{n+1}$.}

\medskip
Each subproblem involves only one spatial operator and can therefore be solved efficiently using dedicated linear solvers. The resulting Lie-splitting IMEX method is globally first order in time. {The $q$-substep is solved by linearized tridiagonal systems with Picard updates, while the $(x,y)$-substep is implicit; the only CFL restriction comes from the explicit jump step.}

In what follows, we specify the finite difference approximations employed in the numerical scheme.  
To unify notation across spatial directions, let $\xi\in\{x,y,q,m\}$ denote a generic spatial variable discretized on a uniform grid 
$\{\xi_\ell\}_{\ell=0}^{N_\xi}$ with step size $\Delta\xi>0$.  
The discrete value $V_\ell$ represents the numerical approximation of the continuous function $v(\xi_\ell)$ at grid point $\xi_\ell$.

\paragraph{First-order derivatives}
Let
$b:=b(\xi_\ell)$ be a transport coefficient, which may depend on different state variables  (e.g., $(t,x,y,q,m)$).  
We denote the backward, forward finite differences, and the centered finite difference, by
\[
\Delta_\xi^- V_\ell := \frac{V_\ell - V_{\ell-1}}{\Delta\xi},
\qquad
\Delta_\xi^+ V_\ell := \frac{V_{\ell+1} - V_\ell}{\Delta\xi},
\qquad
\Delta_\xi^{0} V_\ell := \frac{V_{\ell+1}-V_{\ell-1}}{2\Delta\xi},
\qquad 1 \le \ell \le N_\xi-1.
\]
By splitting $b_\ell$ into its positive and negative parts as follows:
\[
b_\ell^+ := \max(b_\ell,0),
\qquad
b_\ell^- := \min(b_\ell,0),
\]
the monotone upwind approximation of the advective term  
$b(\xi)\,\partial_\xi v(\xi)$ at grid point $\xi_\ell$  reads as
\begin{equation}\label{eq:upwind_first_order}
\big(b\,\partial_\xi v\big)(\xi_\ell)
\;\approx\; \Delta_{\xi,b}^{\mathrm{up}}V_{\ell}
:=
b_\ell^{+}\,\Delta_\xi^- V_\ell
+
b_\ell^{-}\,\Delta_\xi^+ V_\ell.
\end{equation}

This construction applies uniformly for each spatial direction 
$\xi\in\{x,y,q,m\}$ by substituting the corresponding drift function $b$ and grid index set.
Boundary nodes are treated using the extrapolation conditions described in 
Section~\ref{sec:boundary_conditions}.

\paragraph{Second-order total derivatives.}
We approximate second-order spatial derivatives using centered finite differences, 
following the recommendation in \cite{wang2008maximal}:
\begin{equation}\label{eq:centered_second_order}
\partial_{\xi\xi}v(\xi_\ell)
\;\approx\;
\Delta_{\xi\xi}V_\ell
:=
\frac{V_{\ell+1}-2V_\ell+V_{\ell-1}}{(\Delta\xi)^2},
\qquad 1\le \ell\le N_\xi-1.
\end{equation}
This symmetric stencil is second-order accurate.

In what follows, we explain how each of the subproblems is solved sequentially.

\subsection{Semi-Implicit Linearization of the Nonlinear Hamiltonian}
\label{sec:semi_implicit_Hamiltonian}

This section describes the numerical scheme for the nonlinear $q$-Hamiltonian in the Stage~III HJB-PIDE.
After Lie splitting, for each fixed $(x_i,y_j)$ the $q$-substep is a one-dimensional PDE in $q$.

{For $n=0,\ldots,N_t-1$, let $V^n_{i,j,k}$ denote the input at reversed time $\tau_n$.} An implicit Euler discretization in $\tau$ of the $q$-substep yields
\begin{equation}\label{eq:q_substep_fully_implicit}
    {U^{n+1}_{i,j,k}}
    =
    {V^n_{i,j,k}}
    +{\Delta\tau}\,\mathcal{H}_q\!\big(y_j,{p^{n+1}_{i,j,k}}\big),
    \qquad k=1,\dots,N_q-1,
\end{equation}
{where $p^{n+1}_{i,j,k}$ is the discrete $q$-slope associated with $U^{n+1}_{i,j,k}$.}
Solving \eqref{eq:q_substep_fully_implicit} directly would require a nonlinear solve for every $(i,j)$ at every time step. {We therefore use a one-stage Rosenbrock linearization \cite{lang2021rosenbrock}, so that each policy evaluation reduces to a tridiagonal linear system, and then use damped Picard iterations to recover self-consistency.}

Starting from the fully implicit subproblem in \eqref{eq:q_substep_fully_implicit}, we first linearize the mapping $p\mapsto \mathcal{H}_q(y_j,p)$ around the known slope $p^n$ at the previous reversed-time level,
\begin{equation}
    \label{eq:hamiltonian_expansion_qsub}
    \mathcal{H}_q\!\big(y_j,{p^{n+1}_{i,j,k}}\big)
    \approx
    \mathcal{H}_q\!\big(y_j,{p^{n}_{i,j,k}}\big)
    +\partial_p \mathcal{H}_q\!\big(y_j,{p^{n}_{i,j,k}}\big)\,
    \Big({p^{n+1}_{i,j,k}-p^{n}_{i,j,k}}\Big).
\end{equation}
We define the frozen characteristic speed by
\begin{equation}
    \label{eq:def_speed_a_qsub}
    {a^{n}_{i,j,k}}
    := \partial_p \mathcal{H}_q\!\big(y_j,\,{p^{n}_{i,j,k}}\big),
    \qquad k=1,\dots,N_q-1,
\end{equation}
{where the known slope at the previous reversed-time level is
$$
p^n_{i,j,k}:=\big(\Delta_q^0 V^n\big)_{i,j,k},
$$
which approximates $\partial_q v^{III}(T_{\mathrm{gc}}-\tau_n,x_i,y_j,q_k)$.}

Then the corresponding linear system is, for $k=1,\dots,N_q-1$,
\begin{equation}\label{eq:q_rosenbrock_eval_frozen_qsub}
    {U^{n+1}_{i,j,k}}
    -{\Delta\tau}\,{a^{n}_{i,j,k}}\,
    {\big(\Delta_{q,a^{n}}^{\mathrm{up}}U^{n+1}\big)_{i,j,k}}
    =
    {V^{n}_{i,j,k}}
    +{\Delta\tau}\Big(
    \mathcal{H}_q\!\big(y_j,{p^{n}_{i,j,k}}\big)
    -{a^{n}_{i,j,k}}\,{p^{n}_{i,j,k}}
    \Big).
\end{equation}

To reduce the error introduced by the frozen linearization, we perform damped Picard iterations that update the characteristic speed toward self-consistency at reversed time $\tau_{n+1}$. Let $r=0,1,\dots,R_{\max}-1$ denote the Picard iteration counter. Starting from
\begin{equation}
    \label{eq:init_picard_qsub}
    {U^{n+1,(0)}_{i,j,k}}:={V^{n}_{i,j,k}},\quad
    {p^{n+1,(0)}_{i,j,k}}:={p^n_{i,j,k}},\quad
    {a^{n+1,(0)}_{i,j,k}}:={a^n_{i,j,k}},\quad
    k=1,\dots,N_q-1,
\end{equation}
the Picard evaluation step solves, for $k=1,\dots,N_q-1$,
\begin{equation}\label{eq:q_rosenbrock_eval_picard_qsub}
    {U^{n+1,(r+1)}_{i,j,k}}
    -{\Delta\tau}\,{a^{n+1,(r)}_{i,j,k}}\,
    \big(\Delta_{q,a^{n+1,(r)}}^{\mathrm{up}}{U^{n+1,(r+1)}}\big)_{i,j,k}
    =
    {V^{n}_{i,j,k}}
    +{\Delta\tau}\Big(
    \mathcal{H}_q\!\big(y_j,{p^{n+1,(r)}_{i,j,k}}\big)
    -{a^{n+1,(r)}_{i,j,k}}\,{p^{n+1,(r)}_{i,j,k}}
    \Big).
\end{equation}
After solving \eqref{eq:q_rosenbrock_eval_picard_qsub}, we set
\begin{equation}
    \label{eq:def_p_picard_next_qsub}
    {p^{n+1,(r+1)}_{i,j,k}}
    :=\big(\Delta_{q,a^{n+1,(r)}}^{\mathrm{up}}{U^{n+1,(r+1)}}\big)_{i,j,k},
\end{equation}
and update the characteristic speed by
\begin{equation}
    \label{eq:def_a_hat_picard_qsub}
    {\widehat a^{n+1,(r+1)}_{i,j,k}}
    :=\Pi_{[\psi_{\min},\psi_{\max}]}\!\Big(\dfrac{y_j-{p^{n+1,(r+1)}_{i,j,k}}}{\gamma}\Big),
\end{equation}
followed by the damped update
\begin{equation}
    \label{eq:damping_picard_qsub}
    {a^{n+1,(r+1)}_{i,j,k}}
    :=(1-\omega)\,{a^{n+1,(r)}_{i,j,k}}+\omega\,{\widehat a^{n+1,(r+1)}_{i,j,k}},
    \qquad \omega\in(0,1].
\end{equation}
The iteration is terminated once
$\max_{1\le k\le N_q-1}|a^{n+1,(r+1)}_{i,j,k}-a^{n+1,(r)}_{i,j,k}|\le\varepsilon$.
After convergence, the output of Step~1 is $U^{n+1}$.
\FloatBarrier

\begin{remark}[Connection to policy iteration]
	Since the characteristic speed $a(y,p): =\partial_pH_q(y,p)=\psi^\ast(y,p)$ in the present model, freezing $a^{n+1,(r)}$
	is equivalent to freezing the optimal control. Hence, \eqref{eq:q_rosenbrock_eval_picard_qsub} coincides with a policy evaluation step under a frozen control,
	and the update step in Equation \eqref{eq:damping_picard_qsub} corresponds to policy improvement.
\end{remark}

{\FloatBarrier
\subsection{Implicit Discretization in Production and Price Variables}
\label{sec:xy_implicit_discretization}

After the semi-implicit $q$-update in Step~1, we perform the implicit
drift-diffusion substep in the production and price variables. At this stage,
${U^{n+1}_{i,j,k}}$ is known and serves as the input for the second
subproblem. For each fixed inventory grid point $q_k$, we solve
\begin{equation}
\label{eq:xy_implicit_subproblem}
{W^{n+1}_{i,j,k}}
-{\Delta\tau}\,\mathcal H_{xy}^{\Delta}\!\big[{W^{n+1}}\big]_{i,j,k}
={U^{n+1}_{i,j,k}}.
\end{equation}

The continuous second-order part of the $(x,y)$ operator can be written as
\begin{equation}
\label{eq:xy_diffusion_matrix}
\mathcal L^{(2)}_{xy}v
=\tfrac{1}{2}\operatorname{Tr}\!\big(\boldsymbol A(x)\nabla^2_{xy}v\big),
\qquad
\boldsymbol A(x)=
\begin{pmatrix}
\sigma_X(x)^2 & \rho\,\sigma_X(x)\sigma\\
\rho\,\sigma_X(x)\sigma & \sigma^2
\end{pmatrix},
\end{equation}
where $\nabla^2_{xy}v$ denotes the Hessian with respect to $(x,y)$. For the numerical experiments in Section \ref{sec:num_exp_results}, $\rho<0$. The case $\rho>0$
can be obtained by reflecting the sign of the oblique direction, and the case
$\rho=0$ requires no mixed-derivative. For $\rho<0$, the coefficient of the mixed derivative
$\rho\,\sigma_X(x)\sigma\,v_{xy}$ is negative. The standard centered
finite difference approximation of $v_{xy}$ then produces discrete
off-diagonal coefficients with both positive and negative signs, and
therefore does not, in general, yield a monotone discretization \cite{bonnans2003consistency,bonnans2004fast,ma2017unconditionally,chen2018monotone}. We therefore employ a
 specialization of the positive-type generalized
finite difference framework of \cite{bonnans2003consistency}, adapted to
the bounded  domains $x\in[0,1]$ and $ y \in [y_{\min},y_{\max}]$.  {The contribution of the present work is to specialize the construction of the wide stencil to a single oblique direction, derive the corresponding conditions for coefficient nonnegativity, and obtain an explicit condition, specific to the Jacobi diffusion, that ensures the oblique stencil remains inside the computational domain.}

\paragraph{Wide stencil construction.}
For an integer direction
$\boldsymbol r=(r_1,r_2)\in\mathbb Z^2\setminus\{\boldsymbol 0\}$, define
the symmetric directional second difference by
\begin{equation}
\label{eq:bz_directional_difference}
D_{\boldsymbol r}^2 W_{i,j,k}
:=W_{i+r_1,j+r_2,k}+W_{i-r_1,j-r_2,k}-2W_{i,j,k}.
\end{equation}
A positive-type approximation of $\mathcal L^{(2)}_{xy}$ takes the form
\begin{equation}
\label{eq:bz_positive_type_operator}
\mathcal D_i^{(2)}W_{i,j,k}
:=\!\!\sum_{\boldsymbol r\in \mathcal S_i}\!\!
w_{i,\boldsymbol r}\,D_{\boldsymbol r}^2W_{i,j,k},
\qquad w_{i,\boldsymbol r}\ge0,
\end{equation}
where $\mathcal S_i\subset\mathbb Z^2\setminus\{\boldsymbol 0\}$ is a
finite set of integer directions. The cone condition of
\cite{bonnans2003consistency} requires that the scaled diffusion matrix
at $x_i$ be representable as a nonnegative combination of the rank-one
matrices $\boldsymbol r\boldsymbol r^\top$,
$\boldsymbol r\in\mathcal S_i$. The identity \eqref{eq:bz_positive_type_operator} ensures the consistency of the
finite difference approximation \cite{bonnans2003consistency}, while the nonnegativity of the weights ensures the
monotonicity of the stencil (see Proposition \ref{prop:xy_mmatrix}).

At a fixed production grid point $x_i$, the scaled matrix to be represented
in the positive cone is
\begin{equation}
\label{eq:xy_scaled_matrix}
\boldsymbol{\tilde A}_i
:=
\begin{pmatrix}
\dfrac{\sigma_X(x_i)^2}{2(\Delta x)^2}
& \dfrac{\rho\,\sigma_X(x_i)\sigma}{2\Delta x\,\Delta y}
\\[8pt]
\dfrac{\rho\,\sigma_X(x_i)\sigma}{2\Delta x\,\Delta y}
& \dfrac{\sigma^2}{2(\Delta y)^2}
\end{pmatrix}.
\end{equation}
We use the three integer directions
\begin{equation}
\label{eq:xy_bz_directions}
\boldsymbol e_x=(1,0),
\qquad
\boldsymbol e_y=(0,1),
\qquad
\boldsymbol d_i=(-s_i,1),
\qquad
s_i\in\mathbb N,
\end{equation}
where the integer stencil width $s_i$ controls  the width of the
oblique stencil in the production direction.  The rank-one matrix associated with
$\boldsymbol d_i$ is
\begin{equation}
\label{eq:oblique_direction_rank_one}
\boldsymbol d_i\boldsymbol d_i^\top
=\begin{pmatrix} s_i^2 & -s_i\\ -s_i & 1\end{pmatrix},
\end{equation}
whose  off-diagonal entry can reproduce the negative off-diagonal
of $\boldsymbol{\tilde A}_i$ by setting
\begin{equation}
\label{eq:bz_positive_type_identity}
w_i^x\,\boldsymbol e_x\boldsymbol e_x^{\!\top}
+w_i^y\,\boldsymbol e_y\boldsymbol e_y^{\!\top}
+w_i^{xy}\,\boldsymbol d_i\boldsymbol d_i^{\!\top}
=\boldsymbol{\tilde A}_i
\end{equation}
and solving the resulting scalar system \eqref{eq:bz_positive_type_identity} yields
\begin{equation}
\label{eq:xy_bz_weights}
w_i^x=\frac{\chi_i^x}{(\Delta x)^2},
\qquad
w_i^y=\frac{\chi_i^y}{(\Delta y)^2},
\qquad
w_i^{xy}=\frac{\chi_i^{xy}}{(s_i\Delta x)^2+(\Delta y)^2},
\end{equation}
with
\begin{equation}
\label{eq:xy_bz_coefficients}
\begin{aligned}
\chi_i^x &= \tfrac12\sigma_X(x_i)^2
+\tfrac12\rho\,\sigma_X(x_i)\sigma\,\frac{s_i\Delta x}{\Delta y},\\[2pt]
\chi_i^y &= \tfrac12\sigma^2
+\tfrac12\rho\,\sigma_X(x_i)\sigma\,\frac{\Delta y}{s_i\Delta x},\\[2pt]
\chi_i^{xy} &= -\tfrac12\rho\,\sigma_X(x_i)\sigma\,
\frac{(s_i\Delta x)^2+(\Delta y)^2}{s_i\Delta x\,\Delta y}.
\end{aligned}
\end{equation}
The normalized oblique second difference associated with
$\boldsymbol d_i=(-s_i,1)$ is given by
\begin{equation}
\label{eq:xy_oblique_second_difference}
\Delta_{xy}^{s_i}W_{i,j,k}
:=\frac{W_{i+s_i,j-1,k}+W_{i-s_i,j+1,k}-2W_{i,j,k}}
{(s_i\Delta x)^2+(\Delta y)^2}.
\end{equation}
The second-order operator is thus approximated by
\begin{equation}
\label{eq:xy_second_order_discretization}
{\mathcal L^{(2)}_{xy}v^{III}(T_{\mathrm{gc}}-\tau_{n+1},x_i,y_j,q_k)}
\approx
\chi_i^x\,\Delta_{xx}{W^{n+1}_{i,j,k}}
+\chi_i^y\,\Delta_{yy}{W^{n+1}_{i,j,k}}
+\chi_i^{xy}\,\Delta_{xy}^{s_i}{W^{n+1}_{i,j,k}},
\end{equation}

\paragraph{Coefficient nonnegativity.}
We first consider the algebraic condition under which the representation
\eqref{eq:bz_positive_type_identity} has nonnegative coefficients. This
condition is stated for a fixed $i\in\{1,\dots,N_x-1\}$ and a fixed
integer stencil width $s_i$.

\begin{lemma}[Coefficient nonnegativity for a fixed stencil width]
\label{lem:coefficient_nonnegativity}
Let $\rho<0$, let $i\in\{1,\dots,N_x-1\}$, and let
$s_i\in\mathbb N$. Define
\begin{equation}
\label{eq:Ri_definition}
R_i
:=
\frac{\sigma_X(x_i)}{\sigma}\frac{\Delta y}{\Delta x}.
\end{equation}
For the directions
$\boldsymbol e_x=(1,0)$, $\boldsymbol e_y=(0,1)$, and
$\boldsymbol d_i=(-s_i,1)$, the representation
\begin{equation}
\label{eq:positive_type_representation}
\boldsymbol{\tilde A}_i
=
w_i^x\,\boldsymbol e_x\boldsymbol e_x^{\!\top}
+
w_i^y\,\boldsymbol e_y\boldsymbol e_y^{\!\top}
+
w_i^{xy}\,\boldsymbol d_i\boldsymbol d_i^{\!\top}
\end{equation}
has nonnegative coefficients if and only if
\begin{equation}
\label{eq:coefficient_nonnegativity_condition}
|\rho|R_i
\le
s_i
\le
\frac{R_i}{|\rho|}.
\end{equation}
\end{lemma}

\begin{proof}
{The proof is given in
Appendix~\ref{app:proof_coefficient_nonnegativity}.}
\end{proof}

Lemma~\ref{lem:coefficient_nonnegativity} is a fixed-width statement.
Since $s_i$ must be an integer, the coefficient nonnegativity condition
can be satisfied at $x_i$ if and only if
\begin{equation}
\label{eq:integer_stencil_width_condition}
\left\lceil |\rho|R_i\right\rceil
\le
\frac{R_i}{|\rho|}.
\end{equation}
When \eqref{eq:integer_stencil_width_condition} holds, we use the minimal
integer width
\begin{equation}
\label{eq:si_choice}
s_i:=\left\lceil |\rho|R_i\right\rceil,
\qquad
i=1,\dots,N_x-1.
\end{equation}
Then the lower bound in
\eqref{eq:coefficient_nonnegativity_condition} holds by construction, and
the remaining condition for coefficient nonnegativity is
$s_i\le R_i/|\rho|$.

For the Jacobi diffusion, the next lemma gives a 
sufficient condition which guarantees
\eqref{eq:integer_stencil_width_condition}.

\begin{lemma}[Sufficient condition for the Jacobi diffusion]
\label{lem:jacobi_coefficient_nonnegativity}
Assume $\rho<0$ and
$0<|\rho|\le 1/\sqrt2$. If
\begin{equation}
\label{eq:jacobi_lower_mesh_condition}
\frac{\Delta y}{\sqrt{\Delta x}}
\ge
\frac{|\rho|\sigma}{\sqrt{2\alpha\theta_0(1-\Delta x)}},
\end{equation}
then the coefficients in
\eqref{eq:positive_type_representation} are nonnegative for all
$i=1,\dots,N_x-1$ with $s_i$ given by \eqref{eq:si_choice}.
\end{lemma}

\begin{proof}
{The proof is given in
Appendix~\ref{app:proof_jacobi_coefficient_nonnegativity}.}
\end{proof}

\paragraph{{Stencil containment.}}
The preceding conditions concern the signs of the coefficients in the
positive type representation. A separate condition is needed because the
oblique second difference in \eqref{eq:xy_oblique_second_difference} uses
the production indices $i-s_i$ and $i+s_i$. Thus the oblique uses only points from the computational domain
 of $x$ and $y$ if and only if
\begin{equation}
\label{eq:oblique_stencil_condition}
0\le i-s_i,
\qquad
i+s_i\le N_x,
\qquad\text{equivalently}\qquad
s_i
\le
\min\{i,N_x-i\}.
\end{equation}
With the choice \eqref{eq:si_choice}, this condition is equivalent to
$\left\lceil |\rho|R_i\right\rceil\le\min\{i,N_x-i\}$.

\begin{lemma}[Wide stencil condition for the Jacobi diffusion]
\label{lem:jacobi_interior_oblique_stencil}
Assume $\rho<0$. Let
$s_i$ be given by \eqref{eq:si_choice}. Then
\eqref{eq:oblique_stencil_condition} holds for all
$i=1,\dots,N_x-1$ if and only if
\begin{equation}
\label{eq:jacobi_upper_mesh_condition}
\frac{\Delta y}{\sqrt{\Delta x}}
\le
\frac{\sigma}{|\rho|\sqrt{2\alpha\theta_0(1-\Delta x)}}.
\end{equation}
\end{lemma}

\begin{proof}
{The proof is given in
Appendix~\ref{app:proof_jacobi_interior_oblique_stencil}.}
\end{proof}

Combining Lemmas~\ref{lem:jacobi_coefficient_nonnegativity} and
\ref{lem:jacobi_interior_oblique_stencil} yields the sufficient  condition
\begin{equation}
\label{eq:mesh_sandwich}
\frac{|\rho|\sigma}{\sqrt{2\alpha\theta_0(1-\Delta x)}}
\le
\frac{\Delta y}{\sqrt{\Delta x}}
\le
\frac{\sigma}{|\rho|\sqrt{2\alpha\theta_0(1-\Delta x)}}.
\end{equation}
Under \eqref{eq:mesh_sandwich}, both
\eqref{eq:coefficient_nonnegativity_condition} and
\eqref{eq:oblique_stencil_condition} hold for all $i=1,\dots,N_x-1$ with
$s_i$ given by \eqref{eq:si_choice}.

In terms of $N_y$, the condition \eqref{eq:mesh_sandwich} is given by
\begin{equation}
\label{eq:admissible_Ny}
\left\lceil
\frac{|\rho| (y_{\max}-y_{\min})\sqrt{2\alpha\theta_0(1-\Delta x)}}{\sigma}\,
\sqrt{N_x}
\right\rceil
\;\le\; N_y\;\le\;
\left\lfloor
\frac{(y_{\max}-y_{\min})\sqrt{2\alpha\theta_0(1-\Delta x)}}{|\rho|\sigma}\,
\sqrt{N_x}
\right\rfloor.
\end{equation}

\paragraph{First-order terms.}
We now record the first-order terms entering the implicit drift and diffusion
step in the production and price variables. {For $\tau\in[0,T_{\mathrm{gc}}]$, define the reversed-time convection coefficients
$$
b_X(\tau,x):=-\mu_X(T_{\mathrm{gc}}-\tau,x),
\qquad
b_Y(\tau,y):=-\mu_Y(T_{\mathrm{gc}}-\tau,y).
$$
At the new reversed-time level $\tau_{n+1}$, the corresponding grid values are
\begin{equation}
\label{eq:xy_backward_velocities}
b_{X,i}^{n+1}
:=b_X(\tau_{n+1},x_i)
=-\mu_X(T_{\mathrm{gc}}-\tau_{n+1},x_i),
\qquad
b_{Y,j}^{n+1}
:=b_Y(\tau_{n+1},y_j)
=-\mu_Y(T_{\mathrm{gc}}-\tau_{n+1},y_j).
\end{equation}}
 We use the upwind
finite differences defined in \eqref{eq:upwind_first_order}.

For compactness, we introduce the nonnegative transport and diffusion coefficients respectively
\begin{equation}
\label{eq:xy_transport_rates}
a_i^{x,W}
:=
\frac{({b_{X,i}^{n+1}})^+}{\Delta x},
\qquad
a_i^{x,E}
:=
\frac{-({b_{X,i}^{n+1}})^-}{\Delta x},
\qquad
a_j^{y,S}
:=
\frac{({b_{Y,j}^{n+1}})^+}{\Delta y},
\qquad
a_j^{y,N}
:=
\frac{-({b_{Y,j}^{n+1}})^-}{\Delta y},
\end{equation}
\begin{equation}
\label{eq:xy_diffusion_rates}
\nu_i^x
:=
\frac{\chi_i^x}{(\Delta x)^2},
\qquad
\nu_i^y
:=
\frac{\chi_i^y}{(\Delta y)^2},
\qquad
\nu_i^{xy}
:=
\frac{\chi_i^{xy}}{(s_i\Delta x)^2+(\Delta y)^2}.
\end{equation}
The terms in \eqref{eq:xy_transport_rates} are nonnegative by
construction and the terms in \eqref{eq:xy_diffusion_rates} are
nonnegative if the   condition
\eqref{eq:coefficient_nonnegativity_condition} holds.

\paragraph{Matrix assembly at interior nodes.}
For $i\in\{1,\dots,N_x-1\}$ and $j\in\{1,\dots,N_y-1\}$, assuming the
interior oblique stencil condition \eqref{eq:oblique_stencil_condition},
the implicit $(x,y)$ substep is
\begin{equation}
\label{eq:xy_discrete_residual}
\begin{aligned}
&
\frac{{W^{n+1}_{i,j,k}}-{U^{n+1}_{i,j,k}}}{{\Delta\tau}}
+
\Delta_{x,{b_X}}^{\mathrm{up}}{W^{n+1}_{i,j,k}}
+
\Delta_{y,{b_Y}}^{\mathrm{up}}{W^{n+1}_{i,j,k}}
\\
&\qquad
-
\chi_i^x\Delta_{xx}{W^{n+1}_{i,j,k}}
-
\chi_i^y\Delta_{yy}{W^{n+1}_{i,j,k}}
-
\chi_i^{xy}\Delta_{xy}^{s_i}{W^{n+1}_{i,j,k}}
=
0.
\end{aligned}
\end{equation}
Equivalently, after multiplying by ${\Delta\tau}$, the row is
\begin{equation}
\label{eq:xy_matrix_row}
\begin{aligned}
&
C^C_{i,j}{W^{n+1}_{i,j,k}}
+
C^W_{i,j}{W^{n+1}_{i-1,j,k}}
+
C^E_{i,j}{W^{n+1}_{i+1,j,k}}
+
C^S_{i,j}{W^{n+1}_{i,j-1,k}}
+
C^N_{i,j}{W^{n+1}_{i,j+1,k}}
\\
&\qquad
+
C^{SE}_{i,j}{W^{n+1}_{i+s_i,j-1,k}}
+
C^{NW}_{i,j}{W^{n+1}_{i-s_i,j+1,k}}
=
{U^{n+1}_{i,j,k}},
\end{aligned}
\end{equation}
where

\begin{equation}
\label{eq:xy_matrix_coefficients}
\begin{aligned}
C^W_{i,j}
&=
-{\Delta\tau}
\left(
\frac{({b_{X,i}^{n+1}})^+}{\Delta x}
+
\frac{\chi_i^x}{(\Delta x)^2}
\right),
\quad
C^E_{i,j}
=
-{\Delta\tau}
\left(
\frac{-({b_{X,i}^{n+1}})^-}{\Delta x}
+
\frac{\chi_i^x}{(\Delta x)^2}
\right),
\\
C^S_{i,j}
&=
-{\Delta\tau}
\left(
\frac{({b_{Y,j}^{n+1}})^+}{\Delta y}
+
\frac{\chi_i^y}{(\Delta y)^2}
\right),
\quad
C^N_{i,j}
=
-{\Delta\tau}
\left(
\frac{-({b_{Y,j}^{n+1}})^-}{\Delta y}
+
\frac{\chi_i^y}{(\Delta y)^2}
\right),
\\
C^{SE}_{i,j}
&=
-{\Delta\tau}
\frac{\chi_i^{xy}}{(s_i\Delta x)^2+(\Delta y)^2},
\quad
C^{NW}_{i,j}
=
-{\Delta\tau}
\frac{\chi_i^{xy}}{(s_i\Delta x)^2+(\Delta y)^2},
\\
C^C_{i,j}
&=
1+{\Delta\tau}
\left(
\frac{({b_{X,i}^{n+1}})^+-({b_{X,i}^{n+1}})^-}{\Delta x}
+
\frac{({b_{Y,j}^{n+1}})^+-({b_{Y,j}^{n+1}})^-}{\Delta y}
+
\frac{2\chi_i^x}{(\Delta x)^2}
+
\frac{2\chi_i^y}{(\Delta y)^2}
+
\frac{2\chi_i^{xy}}{(s_i\Delta x)^2+(\Delta y)^2}
\right).
\end{aligned}
\end{equation}
In
\eqref{eq:xy_matrix_coefficients}, the off-diagonal matrix coefficients,
namely
$C^W_{i,j}$, $C^E_{i,j}$, $C^S_{i,j}$, $C^N_{i,j}$,
$C^{SE}_{i,j}$, and $C^{NW}_{i,j}$, are nonpositive whenever
\eqref{eq:coefficient_nonnegativity_condition} holds. If further
\eqref{eq:oblique_stencil_condition} holds, then the oblique production
indices $i-s_i$ and $i+s_i$ appearing in \eqref{eq:xy_matrix_row} belong
to $\{0,\dots,N_x\}$.

\paragraph{Boundary contributions }
We use the boundary properties stated in
\eqref{eq:production_diffusion_boundary} and
\eqref{eq:production_drift_boundary}. Since $\sigma_X$ vanishes at
$x=0$ and $x=1$, the rows at $i=0$ and $i=N_x$ contain neither the
diffusion of $x$ nor the mixed-derivative contribution
$\rho\sigma_X(x)\sigma v_{xy}$. They are assembled from the remaining
terms in the $(x,y)$ equation.

{The drift condition \eqref{eq:production_drift_boundary} implies
$b_{X,0}^{n+1}\le0$ and $b_{X,N_x}^{n+1}\ge0$. Hence, for $j=1,\ldots,N_y-1$, the finite difference discretization of
the implicit $(x,y)$ substep is given by
\begin{equation}
\label{eq:production_boundary_rows}
\frac{W^{n+1}_{i,j,k}-U^{n+1}_{i,j,k}}{\Delta\tau}
+
\Delta_{x,b_X}^{\mathrm{up}}W^{n+1}_{i,j,k}
+
\Delta_{y,b_Y}^{\mathrm{up}}W^{n+1}_{i,j,k}
-
\frac{\sigma^2}{2}\Delta_{yy}W^{n+1}_{i,j,k}
=0,
\qquad i\in\{0,N_x\},
\end{equation}
where
\begin{equation}
\Delta_{x,b_X}^{\mathrm{up}}W^{n+1}_{i,j,k}
=
\begin{cases}
\displaystyle
b_{X,0}^{n+1}
\frac{W^{n+1}_{1,j,k}-W^{n+1}_{0,j,k}}{\Delta x},
& i=0,\\[10pt]
\displaystyle
b_{X,N_x}^{n+1}
\frac{W^{n+1}_{N_x,j,k}-W^{n+1}_{N_x-1,j,k}}{\Delta x},
& i=N_x.
\end{cases}
\end{equation}}

At the boundary of $y$, we discretize
\eqref{eq:price_farfield_condition} by a backward difference at
$y=y_{\max}$ and a forward difference at $y=y_{\min}$. This gives
\begin{equation}
\label{eq:affine_boundary_rows}
\begin{aligned}
{W^{n+1}_{i,N_y,k}-W^{n+1}_{i,N_y-1,k}}
&=
{a_+(\tau_{n+1})}\Delta y,
\\
{W^{n+1}_{i,0,k}-W^{n+1}_{i,1,k}}
&=
-{a_-(\tau_{n+1})}\Delta y,
\end{aligned}
\end{equation}
for all $i\in\{0,\ldots,N_x\}$ and $k\in\{0,\ldots,N_q\}$. These rows
replace the implicit $(x,y)$ equation at $j=0$ and $j=N_y$.
 {{After specifying the numerical boundary conditions, we prove the
monotonicity of the implicit $(x,y)$ step in
Proposition~\ref{prop:xy_mmatrix}.}

\begin{proposition}[$M$-matrix property of the implicit $(x,y)$ step]
\label{prop:xy_mmatrix}
Assume $\rho<0$. Assume that the drift of $x$ satisfies
\eqref{eq:production_drift_boundary} and that, for every
$i=1,\dots,N_x-1$, the integer width $s_i$ satisfies
\eqref{eq:coefficient_nonnegativity_condition} and
\eqref{eq:oblique_stencil_condition}. Then, for every ${\Delta\tau}>0$
and each fixed inventory node $q_k$, the coefficient matrix $A_{xy}^{n+1}$ defined by
\eqref{eq:xy_matrix_row}--\eqref{eq:xy_matrix_coefficients}, boundary of $x$ in \eqref{eq:production_boundary_rows}, and boundary of $y$ in \eqref{eq:affine_boundary_rows} is a
nonsingular $M$-matrix. Consequently, the implicit map
${U^{n+1}_{\cdot,\cdot,k}\mapsto W^{n+1}_{\cdot,\cdot,k}}$ is monotone.
\end{proposition}

\begin{proof}
{The proof is given in Appendix~\ref{app:proof_xy_mmatrix}.}
\end{proof}}
}

{Figure~\ref{fig:wide_stencil_width_rho_m08} illustrates the choice of
the integer stencil widths satisfying the conditions of
Proposition~\ref{prop:xy_mmatrix}. For $\rho=-0.8$, the values
$s_i=\left\lceil |\rho|R_i\right\rceil$ lie between the bounds in
\eqref{eq:coefficient_nonnegativity_condition} at every interior production
node $x_i$ and satisfy the containment condition
\eqref{eq:oblique_stencil_condition}.}

\begin{figure}[htbp]
  \centering
  \includegraphics[width=0.60\textwidth]
  {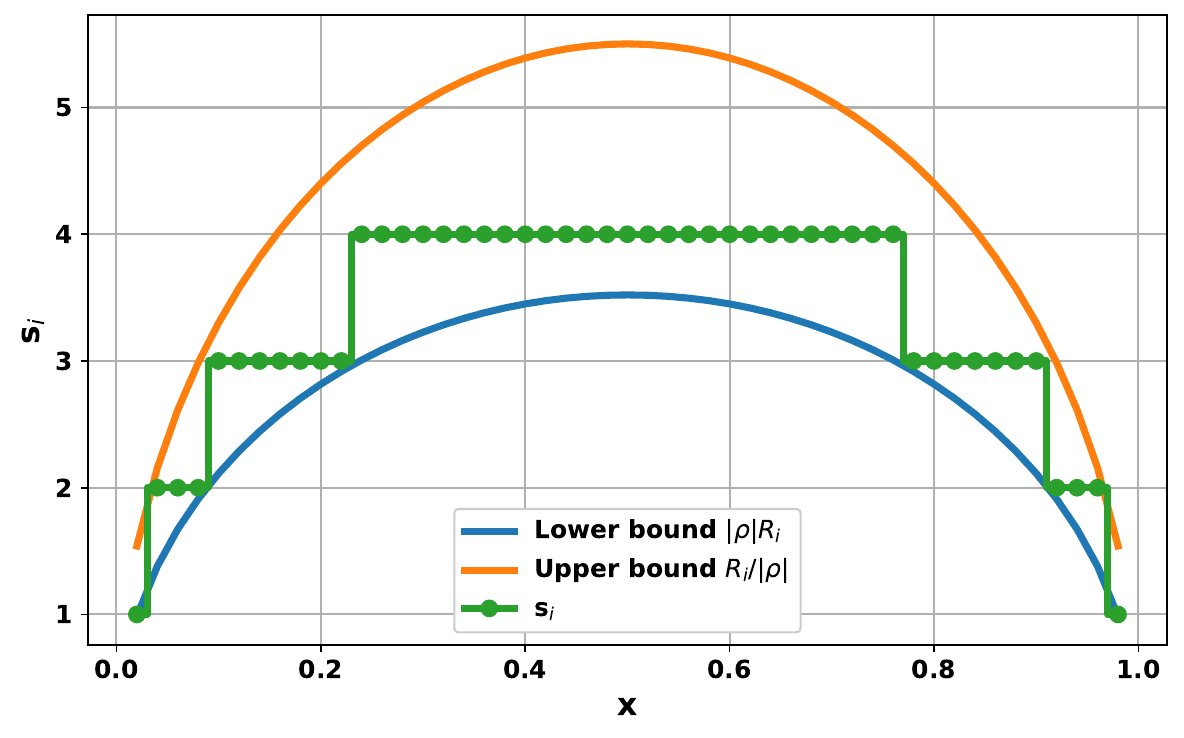}
  \caption{{The bounds $|\rho|R_i$ and $R_i/|\rho|$ in
  \eqref{eq:coefficient_nonnegativity_condition}, together with the selected
  integer stencil widths
  $s_i=\lceil|\rho|R_i\rceil$ from
  \eqref{eq:si_choice}, for $\rho=-0.8$ and $N_x=50$.}}
  \label{fig:wide_stencil_width_rho_m08}
\end{figure}
\FloatBarrier


\subsection{Explicit Discretization of the Jump Operator}
\label{sec:jump_discretization}

After completing the implicit $(x,y)$-update in Step~2, the final step of
the Lie-IMEX scheme treats the nonlocal jump term explicitly. Given
${W^{n+1}_{i,j,k}}$ from the previous substep, we advance to
${V^{n+1}_{i,j,k}}$ via
\begin{equation}\label{eq:jump_update}
{V^{n+1}_{i,j,k}}
=
{W^{n+1}_{i,j,k}}
+
{\Delta\tau}\;\mathcal{I}_y[{W^{n+1}}]_{i,j,k}.
\end{equation}

The standard approach to discretizing the nonlocal jump operator
\cite{cont2005finite} evaluates $\mathcal{I}_y[{W^{n+1}}]_{i,j,k}$ by numerical
quadrature, approximating
$$
\mathcal{I}_y[{W^{n+1}}]_{i,j,k}
\approx
\lambda\sum_{m=1}^{M} w_m
\Bigl(
{W^{n+1}_I}\bigl(x_i,\,y_j+z_m,\,q_k\bigr)
-
{W^{n+1}_{i,j,k}}
\Bigr),
$$
where $\{(z_m,w_m)\}_{m=1}^M$ are quadrature nodes and weights and
${W^{n+1}_I}$ denotes the piecewise linear interpolant of ${W^{n+1}}$ in~$y$. Since the
shifted points $\{y_j+z_m\}$ generally do not coincide with the grid
$\{y_j\}_{j=0}^{N_y}$, the quadrature requires evaluating ${W^{n+1}}$ at
{off grid} locations, which raises two issues. First, when a shifted point
$y_j+z_m$ falls outside $[y_{\min},y_{\max}]$, the interpolant must be
replaced by extrapolation. Extrapolation beyond the grid boundary may
introduce negative coefficients in the discrete jump stencil, violating the
monotonicity property of the numerical scheme. Second, even for interior shifted points,
piecewise linear interpolation introduces an interpolation error that must
be controlled separately.

{For jumps leaving the truncated price domain, we use the same
{asymptotic boundary condition} as in the implicit $(x,y)$ substep. In particular,
\eqref{eq:price_farfield_condition} prescribes the slopes
$a_-(\tau)$ and $a_+(\tau)$ at $y_{\min}$ and $y_{\max}$, respectively. For jumps whose shifted price remains
inside the truncated price domain, we exploit the double exponential
jump size density and use the differential formulation inspired by
\cite{carr2007numerical}. This replaces the semi-infinite jump integrals
by first order ODEs in~$y$, evaluated efficiently via  recurrence relations on the
existing grid, and therefore avoids interpolation at interior off-grid
points.}


By substituting the asymmetric {double exponential} density into the nonlocal
operator \eqref{eq:jump_operator_def}, the jump term can be written as
\begin{equation}\label{eq:jump_decomp}
\mathcal{I}_y[v](t,x,y,q)
=
\lambda\Bigl(
p_+\,\mathcal{J}_+[v](t,x,y,q)
+
p_-\,\mathcal{J}_-[v](t,x,y,q)
-
v(t,x,y,q)
\Bigr),
\end{equation}
where the {semi-infinite} integral operators are defined by
\begin{align}
\mathcal{J}_+[v](t,x,y,q)
&:=
\eta_+\int_0^\infty v(t,x,y+z,q)\,e^{-\eta_+ z}\,\mathrm{d}z,
\label{eq:Jplus_def}\\[4pt]
\mathcal{J}_-[v](t,x,y,q)
&:=
\eta_-\int_{-\infty}^0 v(t,x,y+z,q)\,e^{\eta_- z}\,\mathrm{d}z.
\label{eq:Jminus_def}
\end{align}

{The parts of these integrals for
which the shifted price remains inside the truncated domain are}
\begin{align}
\mathcal{J}_{+,\mathrm{loc}}[v](t,x,y,q)
&:=
\eta_+\int_0^{y_{\max}-y}
v(t,x,y+z,q)\,e^{-\eta_+ z}\,\mathrm{d}z,
\label{eq:Jplus_loc}\\[4pt]
\mathcal{J}_{-,\mathrm{loc}}[v](t,x,y,q)
&:=
\eta_-\int_{y_{\min}-y}^{0}
v(t,x,y+z,q)\,e^{\eta_- z}\,\mathrm{d}z.
\label{eq:Jminus_loc}
\end{align}

{To approximate the remaining parts,
we use the linear continuation implied by the
{asymptotic boundary condition}
\eqref{eq:price_farfield_condition}. Recall that
$\tau=T_{\mathrm{gc}}-t$. For an exterior price
$\widetilde y>y_{\max}$, we use}
$$
{
v(t,x,\widetilde y,q)
\approx
v(t,x,y_{\max},q)
+
a_+(\tau)(\widetilde y-y_{\max}),
}
$$
{whereas for $\widetilde y<y_{\min}$ we use}
$$
{
v(t,x,\widetilde y,q)
\approx
v(t,x,y_{\min},q)
+
a_-(\tau)(\widetilde y-y_{\min}).
}
$$
{Let}
$
{
d_+:=y_{\max}-y,
\qquad
d_-:=y-y_{\min}
}
$. {The corresponding contributions outside the truncated price domain
are then}
$$
{
\eta_+\int_{d_+}^{\infty}
\left[
v(t,x,y_{\max},q)
+
a_+(\tau)(z-d_+)
\right]
e^{-\eta_+z}\,\mathrm{d}z
=
e^{-\eta_+d_+}
\left(
v(t,x,y_{\max},q)
+
\frac{a_+(\tau)}{\eta_+}
\right),
}
$$
{and}
$$
{
\eta_-\int_{-\infty}^{-d_-}
\left[
v(t,x,y_{\min},q)
+
a_-(\tau)(z+d_-)
\right]
e^{\eta_-z}\,\mathrm{d}z
=
e^{-\eta_-d_-}
\left(
v(t,x,y_{\min},q)
-
\frac{a_-(\tau)}{\eta_-}
\right).
}
$$

{Accordingly, on the truncated price domain we approximate the two
semi-infinite integral operators by}
\begin{align}
{\widehat{\mathcal{J}}_+[v](t,x,y,q)}
&{:=
\mathcal{J}_{+,\mathrm{loc}}[v](t,x,y,q)
+
e^{-\eta_+d_+}
\left(
v(t,x,y_{\max},q)
+
\frac{a_+(\tau)}{\eta_+}
\right),}
\label{eq:Jplus_tr}\\[4pt]
{\widehat{\mathcal{J}}_-[v](t,x,y,q)}
&{:=
\mathcal{J}_{-,\mathrm{loc}}[v](t,x,y,q)
+
e^{-\eta_-d_-}
\left(
v(t,x,y_{\min},q)
-
\frac{a_-(\tau)}{\eta_-}
\right).}
\label{eq:Jminus_tr}
\end{align}
{The corresponding approximation of the jump operator is}
\begin{equation}\label{eq:jump_localized}
{
\widehat{\mathcal{I}}_y[v](t,x,y,q)
:=
\lambda\Bigl(
p_+\,\widehat{\mathcal{J}}_+[v](t,x,y,q)
+
p_-\,\widehat{\mathcal{J}}_-[v](t,x,y,q)
-
v(t,x,y,q)
\Bigr).
}
\end{equation}

\begin{proposition}[{Truncation error}]
\label{prop:loc_error}
{Fix $(t,x,q)$ and define, for $u\geq0$,}
$$
{
R_+(u)
:=
v(t,x,y_{\max}+u,q)
-
v(t,x,y_{\max},q)
-
a_+(\tau)u,
}
$$
{and}
$$
{
R_-(u)
:=
v(t,x,y_{\min}-u,q)
-
v(t,x,y_{\min},q)
+
a_-(\tau)u.
}
$$
{Assume that
$R_+(u)e^{-\eta_+u}$ and $R_-(u)e^{-\eta_-u}$ are integrable on
$[0,\infty)$. Then, for every $y\in[y_{\min},y_{\max}]$,}
\begin{align}
{
\mathcal{I}_y[v](t,x,y,q)
-
\widehat{\mathcal{I}}_y[v](t,x,y,q)}
&{
=
\lambda p_+e^{-\eta_+d_+}\eta_+
\int_0^\infty
R_+(u)e^{-\eta_+u}\,\mathrm{d}u}
\nonumber\\
&{\quad
+
\lambda p_-e^{-\eta_-d_-}\eta_-
\int_0^\infty
R_-(u)e^{-\eta_-u}\,\mathrm{d}u.}
\label{eq:loc_error_bound}
\end{align}
{Consequently,}
\begin{align*}
{
\bigl|
\mathcal{I}_y[v](t,x,y,q)
-
\widehat{\mathcal{I}}_y[v](t,x,y,q)
\bigr|}
&{
\le
\lambda p_+e^{-\eta_+d_+}\eta_+
\int_0^\infty
|R_+(u)|e^{-\eta_+u}\,\mathrm{d}u}
\\
&{\quad
+
\lambda p_-e^{-\eta_-d_-}\eta_-
\int_0^\infty
|R_-(u)|e^{-\eta_-u}\,\mathrm{d}u.}
\end{align*}
\end{proposition}

\begin{proof}
The proof is given in Appendix~\ref{app:proof_loc_error}.
\end{proof}

\subsubsection{Differential Formulation}
\label{sec:diff_jump}

{The approximations
$\widehat{\mathcal{J}}_+$ and $\widehat{\mathcal{J}}_-$ of the
semi-infinite integral operators admit first order ODE characterizations
in the price variable $y$, with $(t,x,q)$ treated as parameters.} 

\begin{proposition}[ODE characterization]\label{prop:ode_jump}
Fix $(t,x,q)$ and assume that the map
$y \mapsto v(t,x,y,q)$ is continuous on $[y_{\min},y_{\max}]$. Then the maps
{$y \mapsto \widehat{\mathcal{J}}_+[v](t,x,y,q)$ and
$y \mapsto \widehat{\mathcal{J}}_-[v](t,x,y,q)$}
belong to $C^1([y_{\min},y_{\max}])$ and satisfy
\begin{align}
{\partial_y \widehat{\mathcal{J}}_+[v](t,x,y,q)}
&=
\eta_+\Bigl(
{\widehat{\mathcal{J}}_+[v](t,x,y,q)}
-
v(t,x,y,q)
\Bigr),
\nonumber\\[-2pt]
{\widehat{\mathcal{J}}_+[v](t,x,y_{\max},q)}
&=
{
v(t,x,y_{\max},q)
+
\frac{a_+(\tau)}{\eta_+},
}
\label{eq:ode_Jplus_loc}\\[4pt]
{\partial_y \widehat{\mathcal{J}}_-[v](t,x,y,q)}
&=
\eta_-\Bigl(
v(t,x,y,q)
-
{\widehat{\mathcal{J}}_-[v](t,x,y,q)}
\Bigr),
\nonumber\\[-2pt]
{\widehat{\mathcal{J}}_-[v](t,x,y_{\min},q)}
&=
{
v(t,x,y_{\min},q)
-
\frac{a_-(\tau)}{\eta_-}.
}
\label{eq:ode_Jminus_loc}
\end{align}
\end{proposition}

\begin{proof}
{The proof is given in Appendix~N.}
\end{proof}

We discretize \eqref{eq:ode_Jplus_loc}-\eqref{eq:ode_Jminus_loc} on the
uniform price grid $\{y_j\}_{j=0}^{N_y}$ with step size $\Delta y$. For the
positive part, {we approximate   $W^{n+1}_{i,\cdot,k}$ by the constant
value $W^{n+1}_{i,j,k}$ on each cell $[y_j,y_{j+1}]$.} The ODE
$$
{\partial_y \widehat{\mathcal{J}}_+}
=
\eta_+\bigl(
{\widehat{\mathcal{J}}_+}
-
{W^{n+1}_{i,j,k}}
\bigr),
$$
then admits the exact relation
$$
{\widehat{\mathcal{J}}_{+,j}}
=
e^{-\eta_+\Delta y}\,
{\widehat{\mathcal{J}}_{+,j+1}}
+
\bigl(1-e^{-\eta_+\Delta y}\bigr)\,
{W^{n+1}_{i,j,k}}.
$$
For the negative part, {we approximate $W^{n+1}_{i,\cdot,k}$ by the
constant value $W^{n+1}_{i,j,k}$ on each cell $[y_{j-1},y_j]$.} The ODE
$$
{\partial_y \widehat{\mathcal{J}}_-}
=
\eta_-\bigl(
{W^{n+1}_{i,j,k}}
-
{\widehat{\mathcal{J}}_-}
\bigr),
$$
then has the exact solution
$$
{\widehat{\mathcal{J}}_{-,j}}
=
e^{-\eta_-\Delta y}\,
{\widehat{\mathcal{J}}_{-,j-1}}
+
\bigl(1-e^{-\eta_-\Delta y}\bigr)\,
{W^{n+1}_{i,j,k}}.
$$

Defining $r_+ := e^{-\eta_+\Delta y}\in(0,1)$ and
$r_- := e^{-\eta_-\Delta y}\in(0,1)$, we obtain the recurrence relations
\begin{align}
{\widehat{\mathcal{J}}_{+,j}}
&=
r_+\,{\widehat{\mathcal{J}}_{+,j+1}}
+
(1-r_+)\,{W^{n+1}_{i,j,k}},
\quad j=N_y-1,\ldots,0,
\nonumber\\[-2pt]
{\widehat{\mathcal{J}}_{+,N_y}}
&=
{
W^{n+1}_{i,N_y,k}
+
\frac{a_+(\tau_{n+1})}{\eta_+},
}
\label{eq:rec_plus}\\[4pt]
{\widehat{\mathcal{J}}_{-,j}}
&=
r_-\,{\widehat{\mathcal{J}}_{-,j-1}}
+
(1-r_-)\,{W^{n+1}_{i,j,k}},
\quad j=1,\ldots,N_y,
\nonumber\\[-2pt]
{\widehat{\mathcal{J}}_{-,0}}
&=
{
W^{n+1}_{i,0,k}
-
\frac{a_-(\tau_{n+1})}{\eta_-}.
}
\label{eq:rec_minus}
\end{align}

After computing both positive and negative parts, the {discrete jump
operator on the truncated price domain} is assembled for each interior price node
$j=1,\ldots,N_y-1$ as
\begin{equation}\label{eq:discrete_jump}
{\widehat{\mathcal{I}}_y[{W^{n+1}}]_{i,j,k}}
=
\lambda\Bigl(
p_+\,{\widehat{\mathcal{J}}_{+,j}}
+
p_-\,{\widehat{\mathcal{J}}_{-,j}}
-
{W^{n+1}_{i,j,k}}
\Bigr),
\end{equation}
and the explicit jump substep updates the interior values via
\begin{equation}\label{eq:step3}
{V^{n+1}_{i,j,k}}
=
{W^{n+1}_{i,j,k}}
+
{\Delta\tau}\,
{\widehat{\mathcal{I}}_y[{W^{n+1}}]_{i,j,k}},
\qquad j=1,\ldots,N_y-1.
\end{equation}

{After the interior jump update, we impose the same discrete
{asymptotic boundary condition} as in \eqref{eq:affine_boundary_rows}:}
\begin{equation}
\label{eq:post_jump_affine_refill}
\begin{aligned}
{V^{n+1}_{i,N_y,k}-V^{n+1}_{i,N_y-1,k}}
&=
{a_+(\tau_{n+1})\Delta y},
\\
{V^{n+1}_{i,0,k}-V^{n+1}_{i,1,k}}
&=
-{a_-(\tau_{n+1})\Delta y},
\end{aligned}
\end{equation}
{for all $i\in\{0,\ldots,N_x\}$ and
$k\in\{0,\ldots,N_q\}$.}

\begin{proposition}[Monotonicity of the explicit jump step]
\label{prop:cfl}
For every fixed pair of indices $(i,k)$ and every interior price variable node
$j=1,\dots,N_y-1$, {with the numerical boundary conditions
$a_-(\tau_{n+1})$ and $a_+(\tau_{n+1})$ prescribed by
\eqref{eq:price_farfield_coefficients},} the explicit jump update
$$
{W^{n+1}_{i,j,k}}
+
{\Delta\tau}\,
{\widehat{\mathcal{I}}_y[{W^{n+1}}]_{i,j,k}}
$$
is {an affine function} of the  values
${\{W^{n+1}_{i,\ell,k}\}_{\ell=0}^{N_y}}$ with nonnegative coefficients,
provided that
\begin{equation}\label{eq:cfl}
{\Delta\tau} \le \frac{1}{\lambda}.
\end{equation}
Consequently, the explicit jump step is monotone.
\end{proposition}

\begin{proof}
The proof is given in Appendix~\ref{app:monotonicity_jump}.
\end{proof}


\section{Numerical Experiments and Results}\label{sec:num_exp_results}
\label{sec:numerical_experiments}
This section presents numerical experiments that validate the proposed three-stage stochastic control framework and quantify the  economic impact of key model features on trading performance. We evaluate the optimal trading (OT) strategy obtained from solving the optimal control problem against benchmark strategies on German market data. The OT strategy is computed by solving the three-stage sequence of P(I)DEs using the monotone finite-difference scheme described in Section~\ref{sec:numerical_methodology}, with parameter values summarized in Table~\ref{tab:parameters} unless otherwise specified. The underlying wind production and price data are described in Section~\ref{sec:data_description_parameters}. The experiments are organized as follows. Section~\ref{sec:oos_evaluation} compares the OT strategy against a time-weighted average price (TWAP) benchmark and a perfect-foresight (PF) upper bound over {30 randomly selected trading days}, reporting realized inventory trajectories, trading rates, and P\&L distributions. Section~\ref{sec:jump_sensitivity} examines the sensitivity of the trading performance to the jump intensity in the price process. Section~\ref{sec:delivery_window_sensitivity} investigates the effect of the  length of the delivery window on the P\&L distribution. Section~\ref{sec:trading_horizon_sensitivity} studies the impact of the trading horizon on the P\&L of the optimal  trading strategy. Section~\ref{sec:sensitivity_gamma_beta} analyzes the influence of the liquidity cost parameter and the imbalance charges on the optimal trading policy. Finally, Section~\ref{sec:value_function_structure} provides a visualization of the computed value function, illustrating its dependence on the state variables.

{Unless otherwise specified, all numerical results in this section are computed
on the spatial grid
$
(N_x,N_y,N_q)=(50,50,200)
$,
with $N_t=N_{\mathrm{gc}}=N_{\mathrm d}=300$, denoting the numbers of Stage~III, Stage~II, and Stage~I time steps, respectively.
A maximum of $15$ Picard iterations with stopping tolerance
$10^{-6}$ and a fixed damping
parameter $\omega_c=0.5$ is used during the update step of the Picard iteration in \eqref{eq:damping_picard_qsub}.}

\subsection{Data Description}
	\label{sec:data_description_parameters}
The numerical experiments are based on publicly available data from the German Transmission System Operator (TSO) Amprion, retrieved from the SMARD.de platform\footnote{ \href{https://www.smard.de/home/downloadcenter/download-marktdaten}{https://www.smard.de}}. The dataset covers a two-year period from June 2022 to June 2024. For wind power production, it includes (i) the day-ahead forecasted production $p_X(t)$, (ii) the realized production $X_t$, and (iii) the maximum installed capacity over the Amprion control zone $P_{\max}$, all provided in {15-minute resolution. The} realized production is used both for calibration of the parameters $\alpha$ and $\theta_0$ in Equation~\eqref{eq:wind_sde}, using approximate maximum likelihood method described in \cite{caballero2021quantifying}, and in Section~\ref{sec:oos_evaluation} for the comparison of different trading strategies. {All model parameters are calibrated once from the full dataset (June 2022 to June 2024), which includes the 30 days used in the strategy comparison, and are then kept fixed across the numerical experiments.} The truncation parameter in Equation \eqref{eq:X_forecast_truncation} is set to $\varepsilon_{\mathrm{tr}}  = 0.01$, this value has also been used during  the calibration. The use of aggregated regional data, rather than plant-specific data is a limitation due to the absence of public availability of more local production units. 

As for the intraday price, publicly available data do not provide continuous intraday prices or intraday price forecasts. Consequently, we use the publicly available day-ahead hourly prices as a proxy. This choice is justified by the empirical evidence that intraday continuous prices are strongly coupled to day-ahead auction prices as documented in \cite{han2022complexity}. The day-ahead series is interpolated linearly to a 15-minute resolution, which is treated as the realized quoted price $Y_t$. {The deterministic forecast values are then constructed on the
15-minute grid as a rolling moving average of the interpolated series,
$$
p_Y(t_j)
=
\frac{1}{n_w}
\sum_{k=1}^{n_w}
Y(t_j-k\Delta),
$$
where $\Delta=15$ minutes and $n_w$ denotes the smoothing window length.
Unless specified otherwise, we set $n_w=4$. These values are subsequently
extended to $[0,T_{\mathrm{gc}}]$ using the shape-preserving cubic Hermite
interpolation described in Section~2.2. This construction yields a
continuously differentiable deterministic trajectory with a lag induced
by the backward-looking averaging window.} The development of intraday price forecasting methods is beyond the scope of this work and constitutes a separate research direction with its own extensive literature. The present framework treats the forecast trajectories $p_X(t)$ and $p_Y(t)$ as inputs to the stochastic optimal control problem formulated in Section~\ref{sec:SOC_formulation}. Any forecasting methodology that produces such trajectories can be used in place of the proxy adopted here. We highlight, however, that when the stochastic optimal control problem is solved offline over the full trading horizon, this forecast trajectory is not available before the start of trading. {A fully out-of-sample evaluation would require the model parameters to be estimated from a preceding calibration sample and $p_Y$ to be constructed using only information available before trading begins. Such an evaluation requires a separate forecasting procedure, which is outside the scope of the present work and is left for future research.} Although this approximation cannot capture some stylized intraday price features, like for instance, the increase in volatility near delivery, also known as the Samuelson effect \cite{benth2008stochastic}, the day-ahead price provides a reasonable approximation of the intraday price dynamics. The methodology remains completely reproducible on more problem-specific data when available.

    

	\begin{remark}[Preprocessing]
		Prior to use, days with missing observations are removed, and the production series is normalized by the installed capacity.
	\end{remark}

\paragraph{Evaluation Methodology}
\label{sec:evaluation_methodology}

The
procedure consists of three stages: an offline PDE solve that computes
the value function on a discrete grid, the construction of a feedback
control via interpolation and finite differencing, and a forward
evaluation of the resulting policy on realized or simulated market
data, as specified for each of the numerical experiments. The same forward evaluation loop is applied to the TWAP and PF benchmarks on identical price and production paths, ensuring a fair
comparison.

{The realized numerical P\&L is defined by
\[
\mathrm{P\&L}
=
\sum_{n=0}^{N_t-1}
\left(
\psi_nY_{t_n}
-
\frac{\gamma}{2}\psi_n^2
\right)\Delta t
-
g\!\left(
Q_{N_t}-P_{\max}M_T
\right),
\qquad
M_T
=
\sum_{n=0}^{N_{\mathrm d}-1}
X_{t_n^{\mathrm d}}\,\Delta t_{\mathrm d}.
\]}

Table~\ref{tab:parameters} summarizes the model, economic, and numerical parameters used in the simulations.  
\FloatBarrier
\begin{table}[h!]
	\centering
	\small
	\renewcommand{\arraystretch}{1.2}
	\begin{tabular}{@{}llll@{}}
		\toprule
		\textbf{Parameter} & \textbf{Unit} & \textbf{Description} & \textbf{Value} \\
		\midrule
		\multicolumn{4}{@{}l}{\textit{Wind process}}\\
		$\alpha$        & dimensionless                & Diffusion scale of wind production        & 0.012 \\
		$\theta_0$      & h$^{-1}$                     & Mean-reversion speed of wind process      & 0.0933 \\[2pt]
		
		\multicolumn{4}{@{}l}{\textit{Price process}}\\
		$\sigma$        & EUR/(MWh$\sqrt{\text{h}}$)   & Diffusive volatility of price             & 4.70 \\
		$\kappa$        & h$^{-1}$                     & Mean-reversion speed of price             & 0.2083 \\
		$\rho$          & dimensionless                & Wind-price correlation                    & $-0.3$ \\[2pt]
		
		\multicolumn{4}{@{}l}{\textit{Jump parameters}}\\
		$\lambda$       & h$^{-1}$                     & Jump intensity                            & 0.4167 \\
		$p_+$           & dimensionless                & Probability of upward jump                & 0.65 \\
		$1/\eta_+$      & EUR/MWh                      & Mean size of positive jumps               & 15 \\
		$1/\eta_-$      & EUR/MWh                      & Mean size of negative jumps               & 30 \\[2pt]
		
		\multicolumn{4}{@{}l}{\textit{Economic parameters}}\\
		$\gamma$        & {EUR$\cdot$h/MWh$^{2}$} & Liquidity-cost coefficient                & 0.02 \\
		$\beta$         & EUR/MWh                      & Terminal imbalance penalty                & Day-specific \\[2pt]
		
		\multicolumn{4}{@{}l}{\textit{Time horizons}}\\
		$T_{\text{gc}}$        & h                            & Trading horizon                           & 22.9167 \\
		$L$             & h                            & Delivery-window length                    & 1.0 \\
		$h$             & h                            & Lead time (5 minutes)                     & 0.0833 \\

		\bottomrule
	\end{tabular}
	\caption{Model and numerical parameters. Time units are expressed in hours.}
	\label{tab:parameters}
\end{table}
\FloatBarrier
From Table~\ref{tab:parameters}, a liquidity coefficient of $\gamma = 0.02$
implies that trading $1\,\mathrm{MWh}$ over one hour incurs an execution cost of
approximately $0.01\,\mathrm{EUR}$, while trading $100\,\mathrm{MWh}$ over the
same horizon results in a cost of about $100\,\mathrm{EUR}$. The choice of
$\gamma$ is guided by the empirical calibration in \cite{glas2020intraday}, who
estimate intraday execution costs from EPEX order book data and report
time-of-day-dependent liquidity slopes on the order of
$10^{-2}\,\mathrm{EUR/(MWh)}^{2}$. We select $\gamma$ within
this range to model temporary execution frictions. While
\cite{glas2020intraday} account for intraday liquidity variation, we
adopt a constant baseline coefficient and interpret it as a representative
average over the trading window.

In the existing literature, the imbalance penalty parameter $\beta$ is
typically fixed at an ad-hoc value that is fixed independently of market
conditions \cite{aid2016optimal, tan2018optimal, glas2020intraday}. However, a fixed $\beta$
can lead to economically inconsistent trading behavior: if $\beta$ is
too small relative to the prevailing price level, the optimal policy
effectively disregards the imbalance penalty and trades aggressively to
maximize short-term revenue, which is not consistent with realistic
market dynamics where producers face material balancing charges. 
Conversely, an excessively large $\beta$ induces overly conservative
strategies that forgo profitable trading opportunities. To address
this, we calibrate $\beta$ in a data-driven manner by setting it equal
to the maximum of the forecasted price over the delivery day,
$\beta = \max_{t \in [0,T_{\text{gc}}]} |p_Y(t)|$. This ensures that the marginal
cost of imbalance is of the same order as the marginal revenue from
trading, so that the producer has a clear incentive to trade on the
intraday market while avoiding excessive imbalance exposure. In this
sense, $\beta$ plays the role of a shadow price of electricity at the
balancing stage: it reflects the opportunity cost of failing to deliver
committed energy, scaled to the actual price level of the trading day.
The sensitivity of the optimal policy to the choice of $\beta$ is
illustrated in Section~\ref{sec:sensitivity_gamma_beta}.

\FloatBarrier
\subsection{{Benchmark Strategies and Performance Comparison}}
\label{sec:oos_evaluation}

We compare the optimal trading (OT) strategy, obtained by solving the
full three-stage HJB-KBE P(I)DEs incorporating gate closure, lead
time, and delivery constraints, against two representative benchmark
policies:
\begin{itemize}
\item \textbf{Time-Weighted Average Price (TWAP).}
The TWAP benchmark trades at a constant rate over the trading
interval $[0, T_{\mathrm{gc}}]$ in order to match the forecasted
delivered energy over the delivery window [T-L,T]. We define the forecasted
delivered energy as
\begin{equation}
\label{eq:delivered_energy}
E^{\mathrm{TWAP}}
:=
P_{\max}\int_{T-L}^{T} p_X(s)\,\mathrm{d}s,
\end{equation}
{which we approximate numerically using a simple composite rectangle quadrature rule on the
uniform 15-minute partition,}
\begin{equation}
\label{eq:delivered_energy_discrete}
E^{\mathrm{TWAP}}
\approx
P_{\max}\sum_{\ell=1}^{n_L} p_X(t_\ell)\,\Delta,
\qquad
\Delta = 15 \text{ min},
\end{equation}
{where $t_1,\ldots,t_{n_L}$ are the quadrature points, one for each
15-minute subinterval of $[T-L,T]$.}
For the standard delivery window $L = 1$~h, this corresponds to
$n_L = 4$ quadrature points, at which the forecasts are provided by the TSO as explained in Section \ref{sec:data_description_parameters}. The TWAP target inventory is set to
$Q_{T_{\mathrm{gc}}}^{\mathrm{TWAP}} = E^{\mathrm{TWAP}}$, and the
corresponding constant trading rate is then set as
$$
\psi_t^{\mathrm{TWAP}}
=
\frac{Q_{T_{\mathrm{gc}}}^{\mathrm{TWAP}} - Q_0}{T_{\mathrm{gc}}},
\qquad t \in [0, T_{\mathrm{gc}}],
$$
with $\psi_t^{\mathrm{TWAP}} = 0$ for $t \in (T_{\mathrm{gc}}, T]$.
TWAP is purely based on the energy delivery forecast and allocates the required inventory uniformly over the  trading horizon to serve as a naive benchmark.

\item \textbf{Perfect Foresight (PF).}
The PF benchmark assumes full knowledge of the realized price
path $Y(t)$ over $[0, T_{\mathrm{gc}}]$ and the realized delivered
energy
\begin{equation}
\label{eq:delivered_energy_PF}
E^{\mathrm{PF}}
:=
P_{\max}\int_{T-L}^{T} X_s\,\mathrm{d}s,
\end{equation}
which is approximated on the uniform 15-minute grid by
\begin{equation}
\label{eq:delivered_energy_PF_discrete}
E^{\mathrm{PF}}
\approx
P_{\max}\sum_{\ell=1}^{n_L} X_{t_\ell}\,\Delta,
\qquad
\Delta = 15 \text{ min},
\end{equation}
where $t_1, \ldots, t_{n_L}$ are the grid points in $[T-L, T]$.
Under perfect foresight, the only remaining state variable is the
inventory $q$, and the value function
$v^{\mathrm{PF}}(t,q)$ solves the one-dimensional deterministic HJB equation
{\begin{equation}
\label{eq:PF_HJB}
-\partial_t v^{\mathrm{PF}}(t,q)
=\inf_{\psi\in[\psi_{\min},\psi_{\max}]}
\bigl\{-\psi\,Y(t)+\tfrac{\gamma}{2}\psi^2
+\psi\,\partial_q v^{\mathrm{PF}}(t,q)\bigr\},
\quad t\in[0, T_{\mathrm{gc}}],
\end{equation}}
with terminal condition
$v^{\mathrm{PF}}(T_{\mathrm{gc}},q) = g(q - E^{\mathrm{PF}})$.
PF serves as an upper bound on achievable profit.
\end{itemize}
The central objective of the proposed stochastic optimal control
framework is to generate trading strategies that outperform naive
benchmarks on unseen market data. To rigorously assess this, we
conduct an out-of-sample (OOS) evaluation over 30 randomly selected
trading days from the dataset%
. For each trading day, the
three-stage P(I)DEs sequence of equations is solved offline and {the resulting optimal
policy is evaluated on the price approximation and realized production paths.} {For the comparison with OT, the PF problem~\eqref{eq:PF_HJB} is also solved numerically, using the same time and inventory grids, admissible trading rate bounds, and numerical treatment of the inventory boundaries.} The PF
strategy, which assumes full knowledge of future price and production
trajectories, serves as an upper bound on achievable
performance. {The realized wind production is not used in the offline solve, nor are the realized price fluctuations and jumps in experiments based on simulated data. However, the price input $p_Y$ is constructed from realized prices for the same day and therefore contains ex-post information. Consequently, the comparison is conducted ex-post and does not constitute a fully out-of-sample backtest.} Since the framework accepts any deterministic $p_Y(t)$ as input, replacing the moving-average proxy by a more sophisticated ex-ante forecast does not modify the adopted methodology. The design of such forecasts is beyond the scope of this work and belongs to the extensive intraday electricity price forecasting literature, we refer to \cite{narajewski2020econometric} and references therein.

To illustrate the behavior of the three strategies on a single trading
day, Figures~\ref{fig:oos_trajectory_nocontrol}
and~\ref{fig:oos_trajectory_withcontrol} display the {realized
trajectories} for 2023-04-14. Figure~\ref{fig:oos_trajectory_nocontrol}
shows the inventory paths $Q_t$ under OT, TWAP, and PF alongside the
realized wind power $X_t$ and price $Y_t$. The TWAP inventory grows
linearly by construction, while the OT inventory adapts dynamically
to the evolving market state, closely tracking the PF trajectory.
Figure~\ref{fig:oos_trajectory_withcontrol} additionally displays the
trading rates $\psi_t$. The OT trading rate exhibits clear
state-dependent adjustments: it increases when prices are favorable
and reduces exposure when the inventory is sufficiently aligned with
expected production. By contrast, the TWAP rate remains constant
throughout. The PF trading rate, computed with full knowledge of future
paths, represents the best achievable response at each instant. We note that, under the price forecast trajectory defined in Section \ref{sec:data_description_parameters}, the optimal trading rate $\psi_t^*$ tracks price movements with a lag determined by the smoothing window $n_w$. This is revealed in Figure \ref{fig:oos_trajectory_withcontrol}, where the OT trading rate responds to price changes with a delay. This lagged response is a direct consequence of the forecast construction.

Figure~\ref{fig:oos_boxplot} presents the aggregated P\&L
distributions of the three strategies as box plots. The OT
distribution achieves a substantially higher median P\&L than TWAP,
reflecting the economic value of adaptive, forecast-driven trading.
The wider spread of the OT distribution relative to TWAP is a natural
consequence of actively adjusting positions in response to stochastic
price and production signals, whereas TWAP trades at a constant rate
regardless of market conditions and therefore exhibits minimal
variability by construction. The PF benchmark achieves the highest
median P\&L, but the gap between OT and PF is considerably smaller
than the gap between OT and TWAP, confirming that the proposed
framework captures a substantial fraction of the theoretically
available profitability from using forecast information.

A distinctive feature of {the results} is the proximity of
the OT performance to the perfect-foresight upper bound. The optimal policy is computed offline
by solving the HJB-PIDE using the exogenous forecasts, yet it is evaluated on realized paths
that may deviate from these forecasts. Using the adopted data (see Section \ref{sec:data_description_parameters}), the principal source of forecast error is the production forecast $p_X(t)$, which is the ex-ante day-ahead TSO forecast and therefore deviates from realized production $X_t$, whereas the price forecast $p_Y(t)$ is constructed as a rolling moving average of realized prices and is by construction close to $Y_t$. The fact that the
resulting policy remains competitive with the perfect-foresight bound under such deviations supports the role of the state-feedback
structure of the optimal control, which continuously corrects the
trading rate through $\psi^*(y, \partial_q v)$ in response to the
observed state $(X_t, Y_t, Q_t)$, in absorbing production forecast
errors without requiring recomputation of the value function during the trading session.

{Figure~\ref{fig:oos_perday_violin} shows the empirical distribution
of realized P\&L for OT and TWAP across the 30 trading days, based on one
realized P\&L value for each day and strategy.}

Finally, Table~\ref{tab:oos_benchmarking} summarizes {the results}. The 30 evaluation days are drawn uniformly and randomly from the provided data set (see Section \ref{sec:data_description_parameters}). For each trading day $d = 1, \ldots, 30$ and for the comparison of strategy $A$ against strategy $B$, we compute the daywise absolute and relative gains
\begin{equation}
\label{eq:gain_definitions}
G_a^d(A,B) := \mathrm{P\&L}^d_A - \mathrm{P\&L}^d_B,
\qquad
G_r^d(A,B) := \frac{G_a^d(A,B)}{|\mathrm{P\&L}^d_B|}\times 100\%.
\end{equation}

The reported statistics are then
$\mathrm{mean}(G_a^d(A,B))$, $\mathrm{median}(G_a^d(A,B))$,
$\max_d G_a^d(A,B)$, $\min_d G_a^d(A,B)$ for the absolute gains, and
$\mathrm{mean}(G_r^d(A,B))$, $\mathrm{median}(G_r^d(A,B))$,
$\max_d G_r^d(A,B)$, $\min_d G_r^d(A,B)$ for the relative gains,
each taken over the 30 days. Table \ref{tab:oos_benchmarking} shows that OT outperforms TWAP on 28 out of 30
days, with a mean absolute gain of approximately $+737\,000$~EUR. On
the two remaining days, OT underperforms TWAP by about $15\%$. The
PF benchmark yields a higher {P\&L} than OT on all 30
days, with a mean daywise relative gain of $+23\%$ (median $+19\%$).
The PF vs OT gap remains modest relative to the OT vs TWAP gain,
confirming that OT captures most of the out-of-sample improvement
over TWAP achieved by the PF benchmark.  The PF vs OT gap is also informative from a forecast-valuation perspective: the mean daywise gap of $+111000 \sim$ EUR provides an indicative upper bound on the incremental profit attainable by improving the forecasts used as inputs to the control problem.

\begin{figure}[!ht]
	\centering
	\subfloat[Inventory paths]{%
		\includegraphics[width=0.5\textwidth]{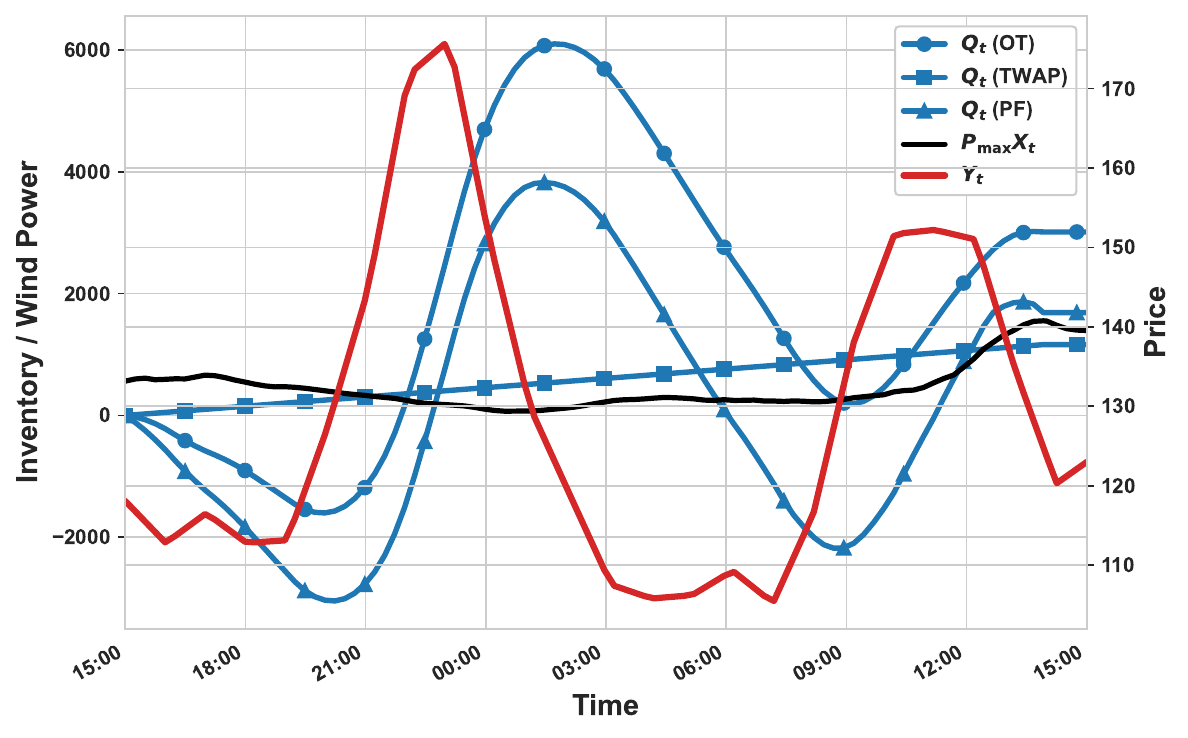}%
		\label{fig:oos_trajectory_nocontrol}}
	\hfill
	\subfloat[Inventory paths and trading rates]{%
		\includegraphics[width=0.5\textwidth]{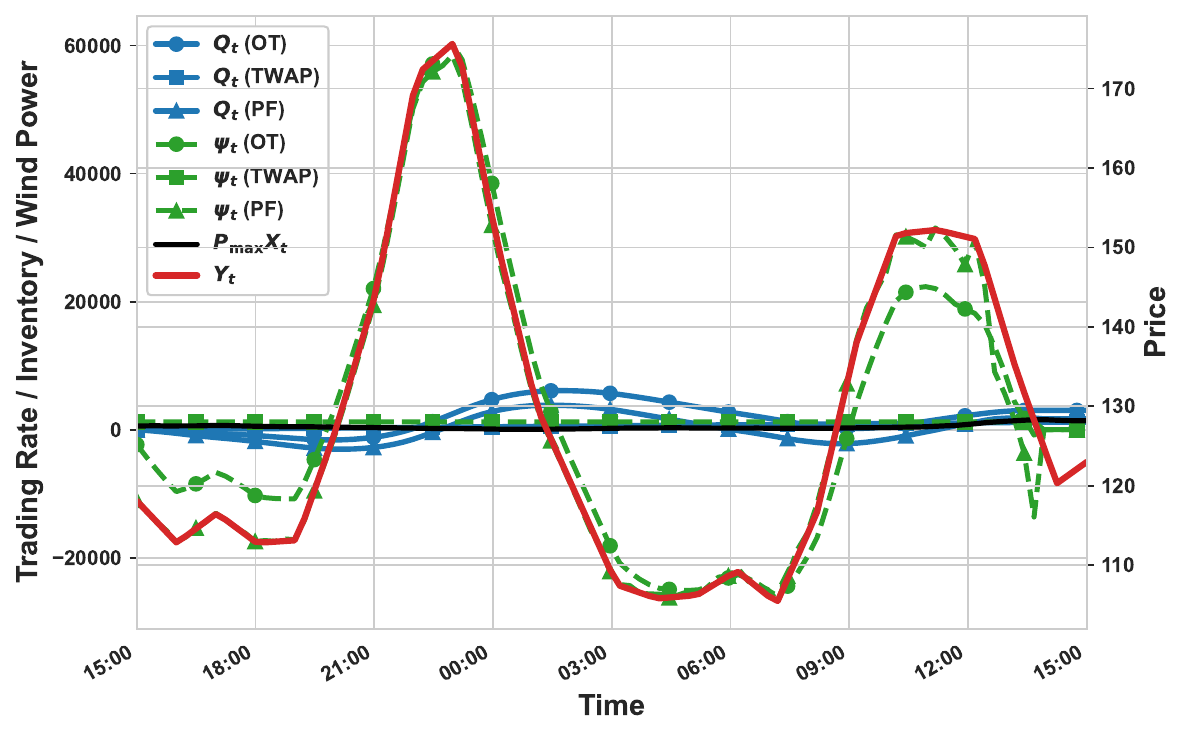}%
		\label{fig:oos_trajectory_withcontrol}}
	\caption{{Trajectory comparison for 2023-04-14.}}
	\label{fig:oos_trajectories}
\end{figure}

\begin{figure}[!ht]
	\centering
	\includegraphics[width=0.55\linewidth]{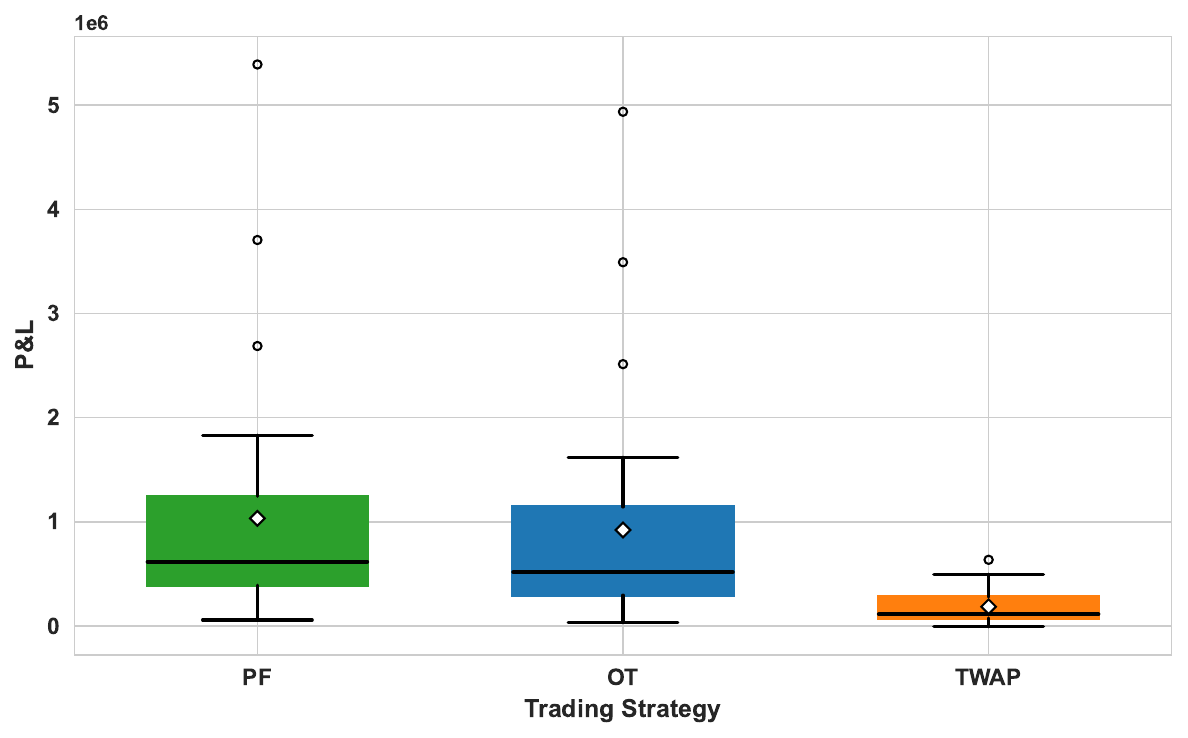}
	\caption{{Aggregated P\&L box plots over 30 randomly selected
		trading days for PF, OT, and TWAP.}}
	\label{fig:oos_boxplot}
\end{figure}

\begin{figure}[H]
	\centering
	\includegraphics[width=0.65\linewidth]{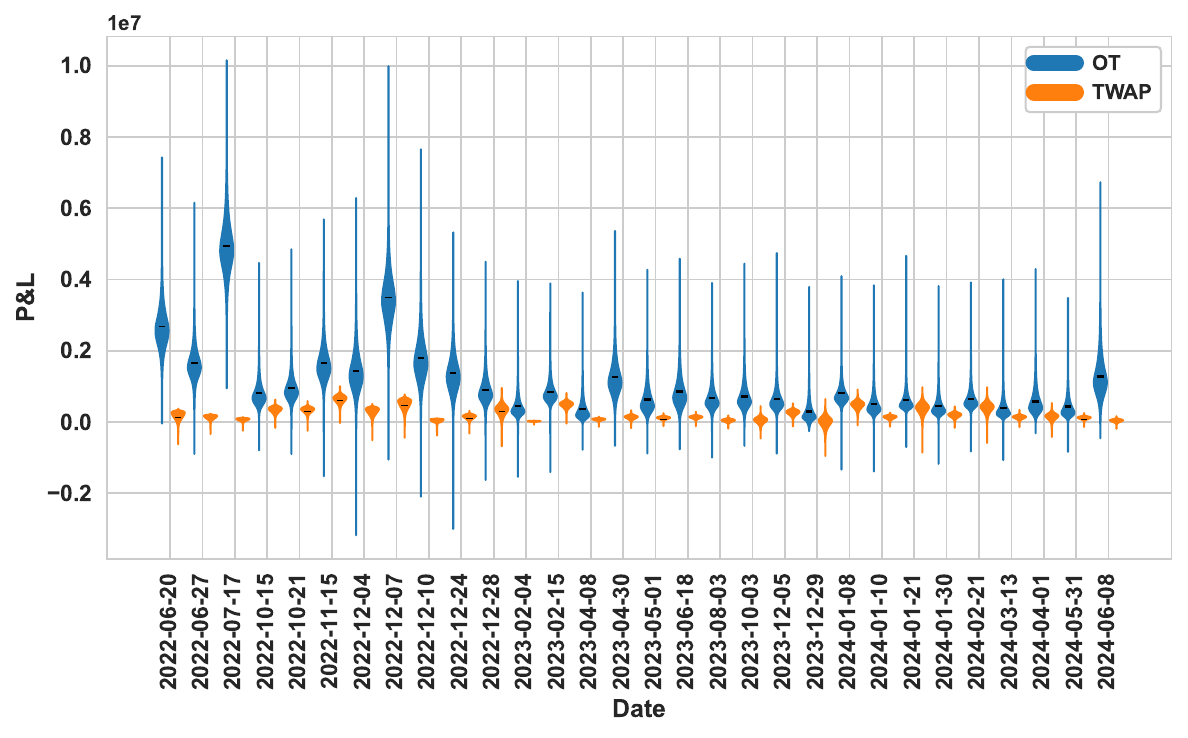}
	\caption{{Empirical distributions of realized P\&L for OT and TWAP
		across the 30 sampled trading days.}}
	\label{fig:oos_perday_violin}
\end{figure}

\begin{table}[!ht]
	\centering
	\small
	\begin{tabular}{l r r}
		\toprule
		Metric & OT vs TWAP & PF vs OT \\
		\midrule
		Mean gain     & +737\,459 EUR \;(+891\%)     & +111\,179 EUR \;(+23\%)  \\
		Median gain   & +333\,676 EUR \;(+224\%)     & +93\,016 EUR \;(+19\%)   \\
		Max gain      & +4\,857\,007 EUR \;(+5983\%) & +454\,383 EUR \;{(+68\%)}   \\
		Min gain      & $-52\,368$ EUR \;($-15\%$)   & +22\,752 EUR \;{(+9\%)}   \\
		Win rate      & 28/30 \;(93.3\%)             & 30/30 \;(100\%)          \\
		\bottomrule
	\end{tabular}
	\caption{{Benchmark comparison over 30 randomly selected trading
		days.}}
	\label{tab:oos_benchmarking}
\end{table}

\FloatBarrier

\subsection{Sensitivity to Jump Intensity}
\label{sec:jump_sensitivity}

We examine the sensitivity of trading performance to the jump
intensity $\lambda$ in the price process. The optimal control is
solved for each value of $\lambda$ and evaluated on simulated
trajectories of $(X_t, Y_t)$ with $10^5$ Monte Carlo paths per day
over 30 randomly selected trading days, {under
$\lambda \in \{0,0.4167,2.083\}\,\mathrm{h}^{-1}$,
corresponding to approximately $0$, $10$, and $50$ expected jumps per
24 hours.}

Figure~\ref{fig:jump_boxplot} displays the aggregated P\&L box plots
across the 30 days for each jump intensity. Several features are
visible. First, both the median and mean P\&L increase with $\lambda$,
indicating that higher jump activity provides additional trading
opportunities that the producer can take advantage of. The optimal
policy adapts its trading rate in response to price jumps, capturing
favorable movements while the imbalance penalty limits exposure to
adverse ones.

Second, the interquartile range and the extent of the whiskers grow
with $\lambda$, reflecting heavier tails in the P\&L distribution.
The number and magnitude of outliers also increase, confirming that
jump risk materially affects the tail behavior. The diffusion-only
case ($\lambda = 0$) produces the tightest distribution, suggesting that, a {pure diffusion} model underestimates both the upside potential and the downside risk of the simulated P\&L distribution. Within this setting, the choice of price model materially affects the assessment of tail risk.

Third, the gap between the mean (diamond markers) and the median
widens as $\lambda$ increases, which is characteristic of
right-skewed, heavy-tailed distributions. This confirms that the
inclusion of jumps in the price dynamics is not merely a modeling
refinement but has direct economic consequences for both the expected
profitability and the risk profile of the trading strategy.

\begin{figure}[!ht]
	\centering
	\includegraphics[width=0.55\linewidth]{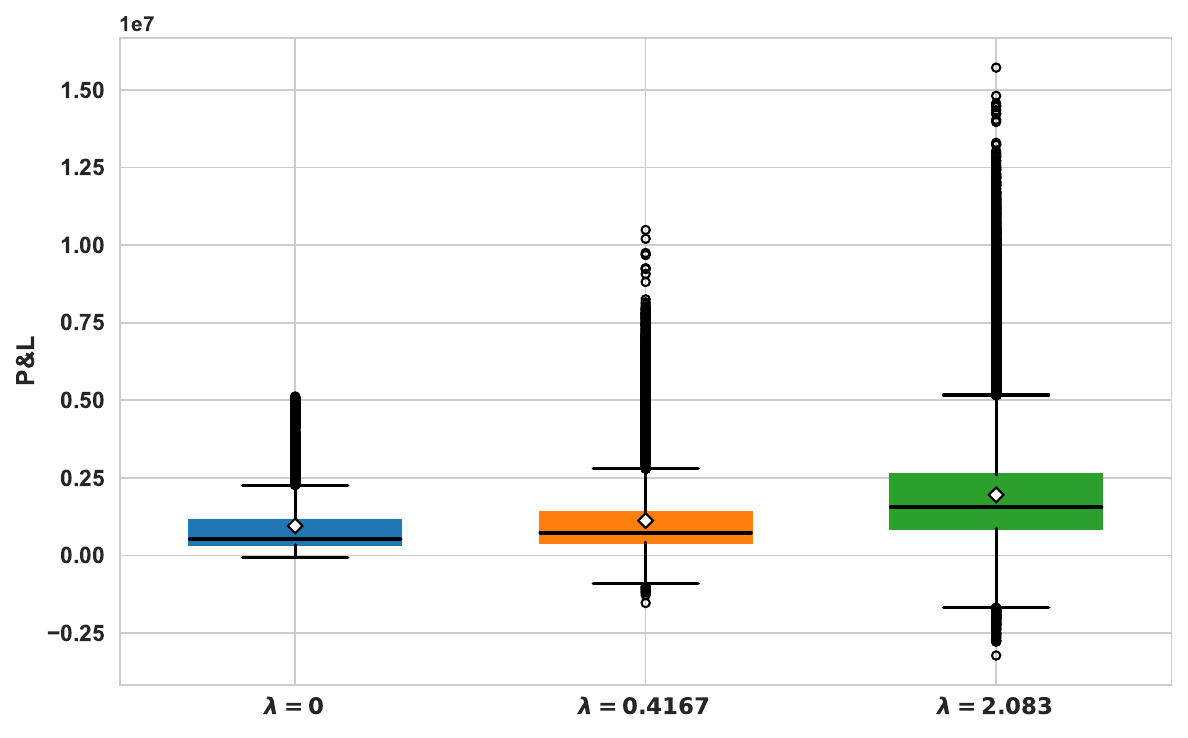}
	\caption{Box plots of simulated P\&L under varying jump intensities
$\lambda \in \{0,\,0.4167,\,2.083\}$, {corresponding respectively to approximately $0$, $10$, and $50$ expected
jumps per 24 hours,} aggregated over 30 randomly selected
trading days ($10^5$ MC paths per day). Diamond markers
indicate the mean.}
	\label{fig:jump_boxplot}
\end{figure}
\FloatBarrier

\subsection{Sensitivity to Delivery Window Length}
\label{sec:delivery_window_sensitivity}

A distinguishing feature of the proposed formulation is the explicit
modeling of the delivery window length $L$ through the energy-based
imbalance penalty integrated over $[T - L,\, T]$. To
assess the impact of this modeling choice on trading performance, we
solve the three-stage control problem for delivery window lengths
corresponding to 15-minute, 30-minute, and 1-hour products and
evaluate the resulting optimal policy on simulated trajectories with
$10^5$ Monte Carlo paths per day over 30 randomly selected trading
days.

Figure~\ref{fig:L_sensitivity} displays box plots of the simulated
P\&L for the three delivery windows. Both the median and mean P\&L
increase with $L$, indicating that longer delivery products yield
more favorable economic outcomes. A potential explanation is that the average energy produced over a
longer delivery window exhibits lower variance than the production
over shorter time periods, making the imbalance penalty more
predictable and thus easier to hedge against. As a result, the producer can commit to trading
positions with greater confidence when the settlement is based on
energy integrated over a longer interval. The interquartile range
also grows with $L$, reflecting the wider range of outcomes associated
with longer delivery windows.

From a practical standpoint, these results suggest that trading hourly
products may yield more favorable risk-adjusted outcomes, as the
longer delivery window provides a natural averaging effect that
reduces the exposure to short-term production fluctuations. This
further highlights that replacing the energy-based settlement over
$[T-L,\, T]$ with a pointwise penalty at a single instant, as is
common in the existing literature, may lead to a misspecified control
problem that overestimates imbalance risk and yields suboptimal
trading policies.

\FloatBarrier
\begin{figure}[!ht]
	\centering
	\includegraphics[width=0.55\linewidth]{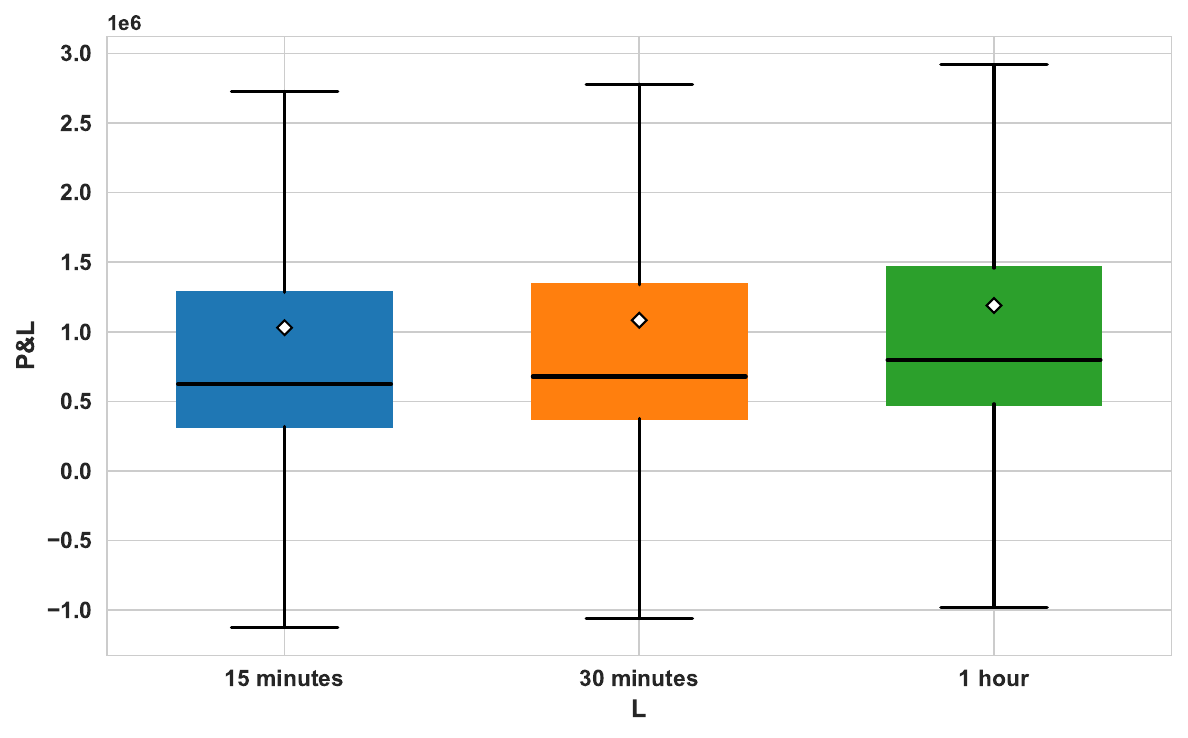}
	\caption{Box plots of simulated P\&L across 30 randomly selected
		trading days for delivery window lengths $L$ corresponding to
		15-minute, 30-minute, and 1-hour products. Diamond markers
		indicate the mean. }
	\label{fig:L_sensitivity}
\end{figure}
\FloatBarrier

\subsection{Sensitivity to Trading Horizon}
\label{sec:trading_horizon_sensitivity}

We investigate the impact of the total problem horizon $T$ on the
performance of the optimal strategy. Since the lead time $h = 5$~min
and delivery window $L = 1$~h are held fixed, varying $T$ directly
determines the length of the active trading window via
$T_{\mathrm{gc}} = T - h - L$. For instance, $T = 24$~h corresponds
to $T_{\mathrm{gc}} \approx 22$~h~$55$~min, while $T = 12$~h
corresponds to $T_{\mathrm{gc}} \approx 10$~h~$55$~min. The optimal
control is solved for each value of $T$ and evaluated on simulated
trajectories with $10^5$ Monte Carlo paths per day over 30 randomly
selected trading days.

Figure~\ref{fig:trading_horizon_sensitivity} displays box plots of the
simulated P\&L for $T \in \{12\text{h},\, 18\text{h},\, 24\text{h}\}$.
Both the median and mean P\&L increase with the horizon, confirming
that a longer trading window allows the optimal policy to spread
trades over a greater time span, reducing the effective liquidity cost
per unit of energy traded. The interquartile range also grows with
$T$, reflecting the increased variability that comes with a longer
exposure to stochastic price and production dynamics. The number of
outliers above the upper whisker increases for longer horizons,
indicating that the producer benefits from rare but profitable
trading opportunities that are only accessible when sufficient time
remains. Conversely, shorter horizons compress the trading activity
into a smaller interval, forcing more aggressive trading rates and
increasing exposure to temporary market impact, which results in a
tighter but lower P\&L distribution. These results suggest that a longer active trading window is beneficial to the producer to adjust their position.

\begin{figure}[!ht]
	\centering
	\includegraphics[width=0.55\linewidth]{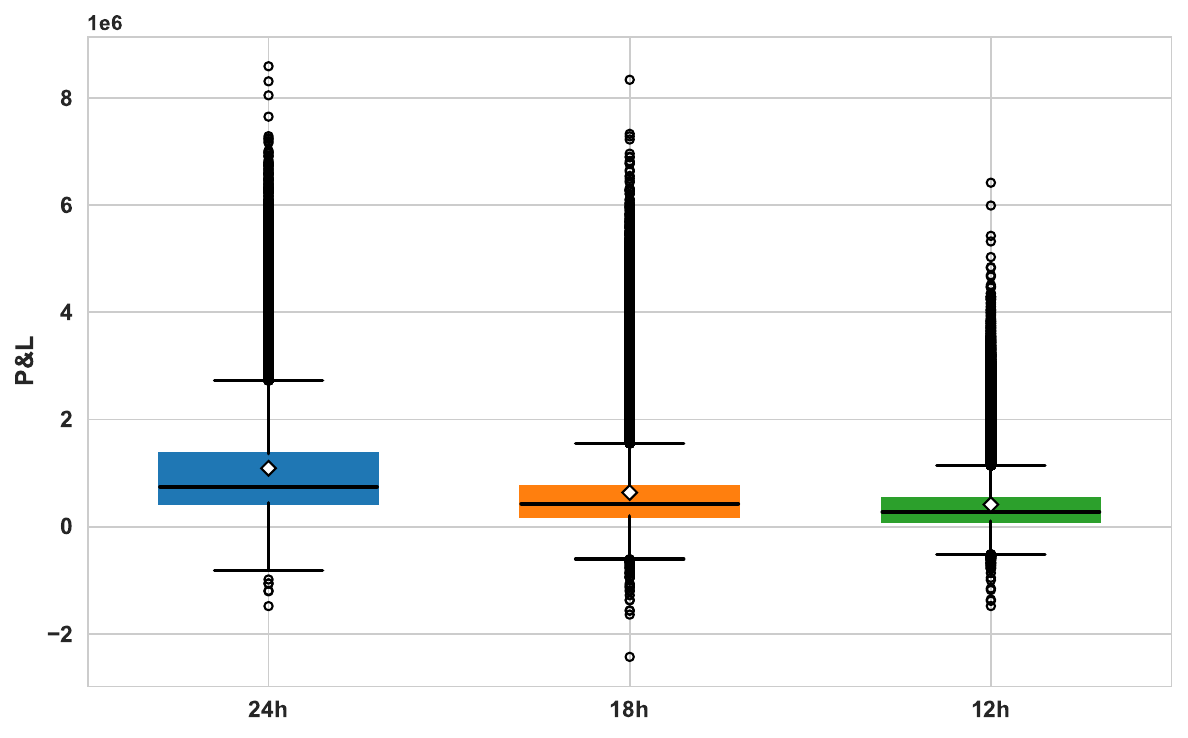}
	\caption{Box plots of simulated P\&L across 30 randomly selected
		trading days for problem horizons
		$T \in \{12\text{h},\, 18\text{h},\, 24\text{h}\}$ with
		$h = 5$~min and $L = 1$~h fixed. Diamond markers indicate the
		mean.}
	\label{fig:trading_horizon_sensitivity}
\end{figure}
\FloatBarrier

\subsection{Sensitivity to Liquidity and Imbalance Penalties}
\label{sec:sensitivity_gamma_beta}

We conclude the numerical experiments by analyzing the influence of the
liquidity parameter $\gamma$ and the imbalance penalty $\beta$ on the
optimal trading policy. The optimal
strategy is applied on the realized price and production
paths for the trading day 2023-04-12.

Figure~\ref{fig:gamma_sensitivity} displays the inventory and trading
rate trajectories for $\gamma \in \{0.005,\, 0.01,\, 0.05\}$ with the value of $\beta = 145$ fixed. Lower
values of $\gamma$ reduce the cost of rapid execution, enabling the
producer to trade more aggressively and to track the realized
wind production $X_t$ more closely. As $\gamma$ increases, the penalty
on trading speed forces the producer to spread its trades over longer
periods, resulting in smoother inventory trajectories that deviate
further from the production profile.

Figure~\ref{fig:beta_sensitivity} shows the corresponding trajectories
for $\beta \in \{116,\, 145,\, 174\}$ with the value of $\gamma = 0.02$ fixed. Increasing $\beta$ raises the
cost of overcommitting relative to actual delivery, which incentivizes
the producer to align inventory more tightly with the anticipated
production. As a result, higher $\beta$ values produce inventory paths
that converge earlier and more closely toward $X_t$, with the trading
rate adjusting accordingly. Conversely, lower $\beta$ values relax the
imbalance constraint, allowing the producer to prioritize
speculation of favorable price movements at the expense of larger
potential imbalances at delivery.
\FloatBarrier
\begin{figure}[!ht]
	\centering
	\includegraphics[width=0.65\linewidth]{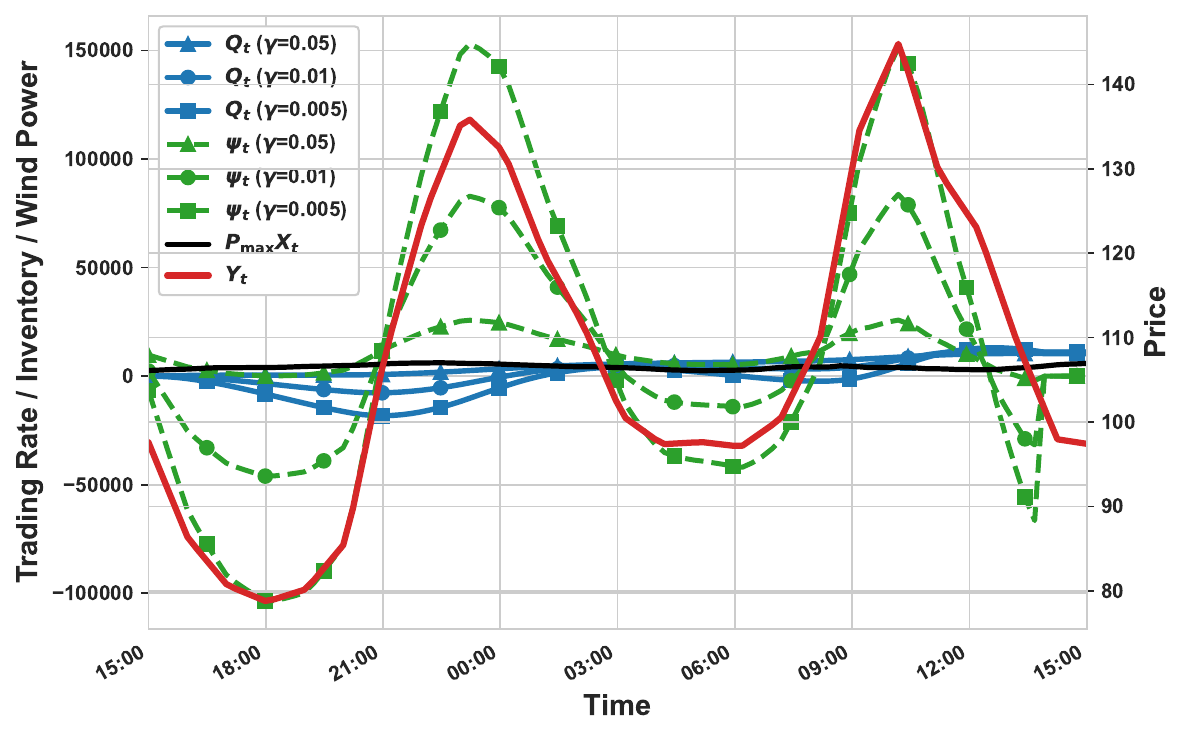}
	\caption{{Inventory} $Q_t$, trading rate $\psi_t$,
		realized production $X_t$, and realized price $Y_t$ for the
		trading day 2023-04-12 under varying liquidity parameter
		$\gamma \in \{0.005,\, 0.01,\, 0.05\}$.}
	\label{fig:gamma_sensitivity}
\end{figure}

\FloatBarrier
\begin{figure}[!ht]
	\centering
	\includegraphics[width=0.65\linewidth]{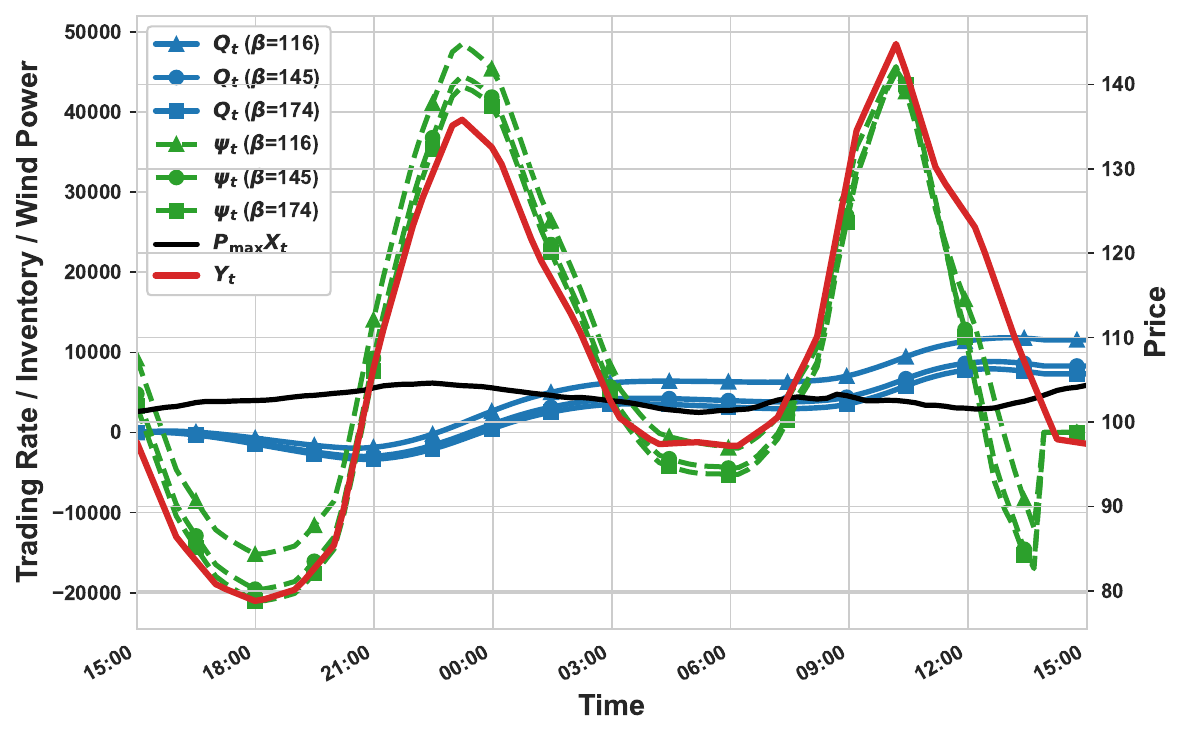}
	\caption{{Inventory} $Q_t$, trading rate $\psi_t$,
		realized production $X_t$, and realized price $Y_t$ for the
		trading day 2023-04-12 under varying imbalance penalty
		$\beta \in \{116,\, 145,\, 174\}$.}
	\label{fig:beta_sensitivity}
\end{figure}
\FloatBarrier

\subsection{Visualization of the Value Function}
\label{sec:value_function_structure}

To provide insight into the shape of the approximated value function,
Figures~\ref{fig:V_surfaces_tgc} and~\ref{fig:V_surfaces_t0} show
two-dimensional cross-sections of the Stage~III numerical value
function $V^{\mathrm{III}}$ near gate closure $t \approx T_{\mathrm{gc}}$
and at the initial time $t = 0$, respectively, for the trading day
2022-12-11. Each surface is obtained by fixing one state variable at
the midpoint of its computational domain and varying the remaining
two. The value function is computed on a uniform grid with
$(N_x, N_y, N_q) = (50, 50, 200)$.

Figure~\ref{fig:V_surfaces_tgc} also shows that near gate closure, the value
function inherits the non-smooth structure of the terminal penalty
$g(\xi) = \beta|\xi|$. The $(y,q)$ cross-section exhibits a sharp
ridge near the region where the inventory matches the expected
delivered energy, with steep gradients on both sides reflecting the
symmetric penalization of shortfall and surplus. The $(x,q)$
cross-section shows a similar kink along the $q$-direction, while the
$(x,y)$ cross-section reveals the dependence on production and price
levels close to terminal time.

Figure~\ref{fig:V_surfaces_t0} shows that the backward propagation
of the PDE has smoothed the value function at $t= 0$. The kink inherited from
the terminal penalty is no longer visible, as the remaining trading
horizon provides sufficient time for the optimal policy to absorb
imbalance risk. {The $(x,q)$ cross-section shows the dependence of the value function on
the inventory relative to the production level, reflecting the associated
imbalance cost.} The
$(x,y)$ cross-section exhibits a clear dependence on the price level
$y$, which determines the profitability of trading, and on the
production level $x$, which governs the expected imbalance at
delivery. The $(y,q)$ cross-section isolates the joint effect of
price and inventory on the cost-to-go, now smoother in both
directions.

\begin{figure}[!ht]
	\centering
	\subfloat[$(x, q)$ plane, $y = y_{\mathrm{mid}}$]{%
		\includegraphics[width=0.32\textwidth]{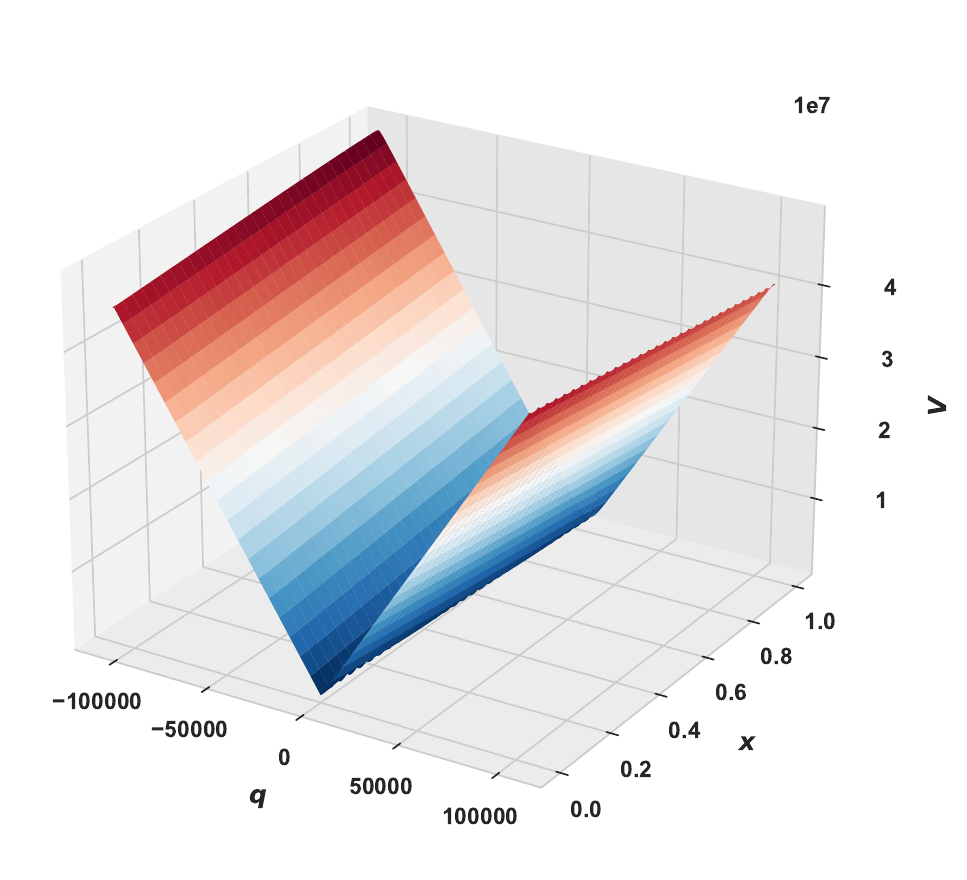}%
		\label{fig:V_xq_tgc}}
	\hfill
	\subfloat[$(x, y)$ plane, $q = q_{\mathrm{mid}}$]{%
		\includegraphics[width=0.32\textwidth]{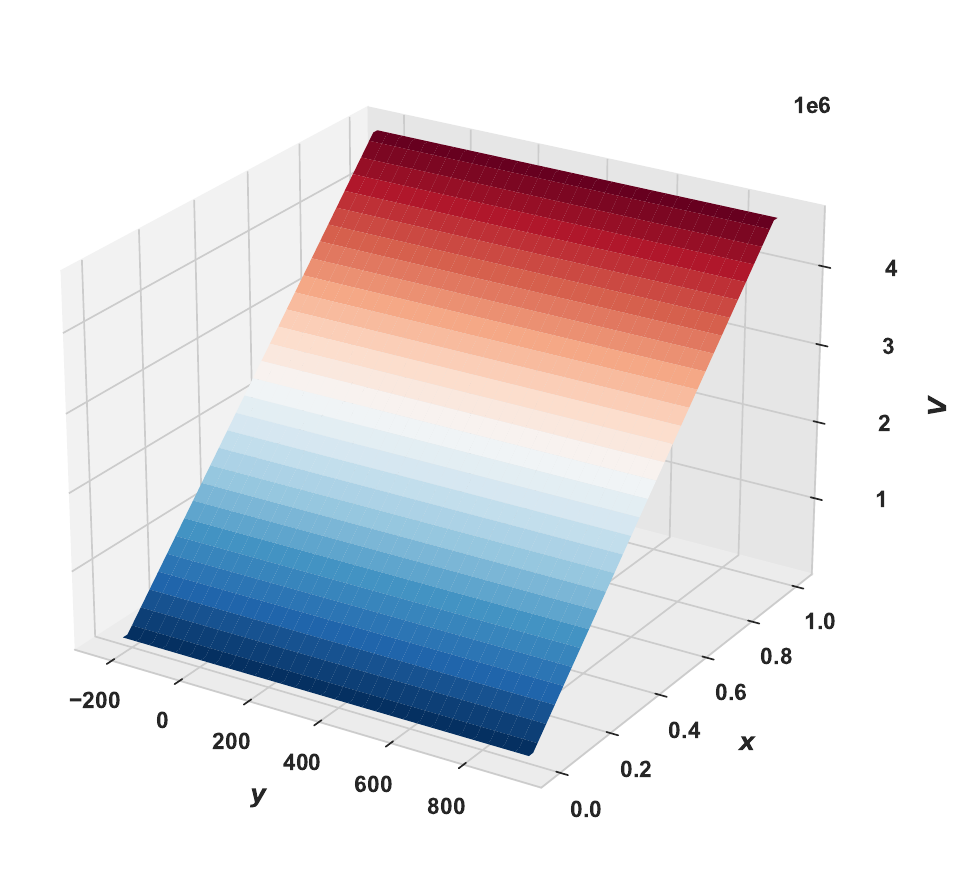}%
		\label{fig:V_xy_tgc}}
	\hfill
	\subfloat[$(y, q)$ plane, $x = x_{\mathrm{mid}}$]{%
		\includegraphics[width=0.32\textwidth]{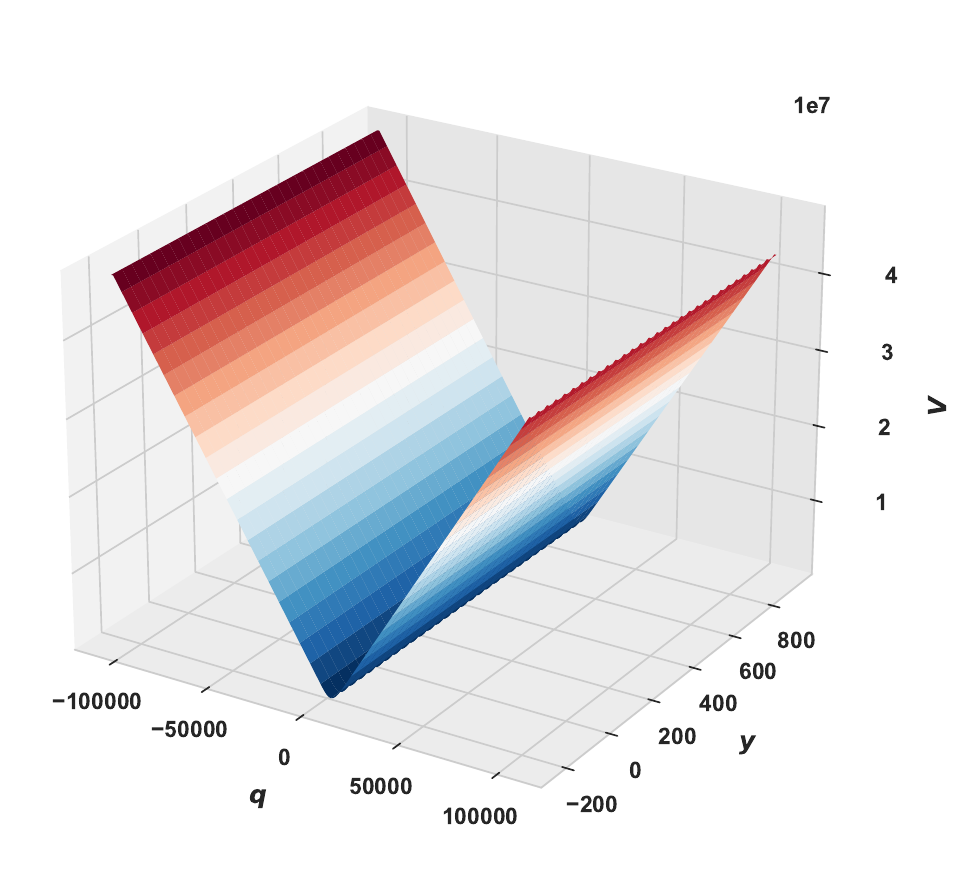}%
		\label{fig:V_yq_tgc}}
	\caption{Cross-sections of the numerical value function
		$V^{\mathrm{III}}(T_{\mathrm{gc}}, x_i, y_j, q_k)$ near gate
		closure, each with one state variable fixed at the midpoint of
		its computational domain. Trading day: 2022-12-11.}
	\label{fig:V_surfaces_tgc}
\end{figure}

\begin{figure}[!ht]
	\centering
	\subfloat[$(x, q)$ plane, $y = y_{\mathrm{mid}}$]{%
		\includegraphics[width=0.32\textwidth]{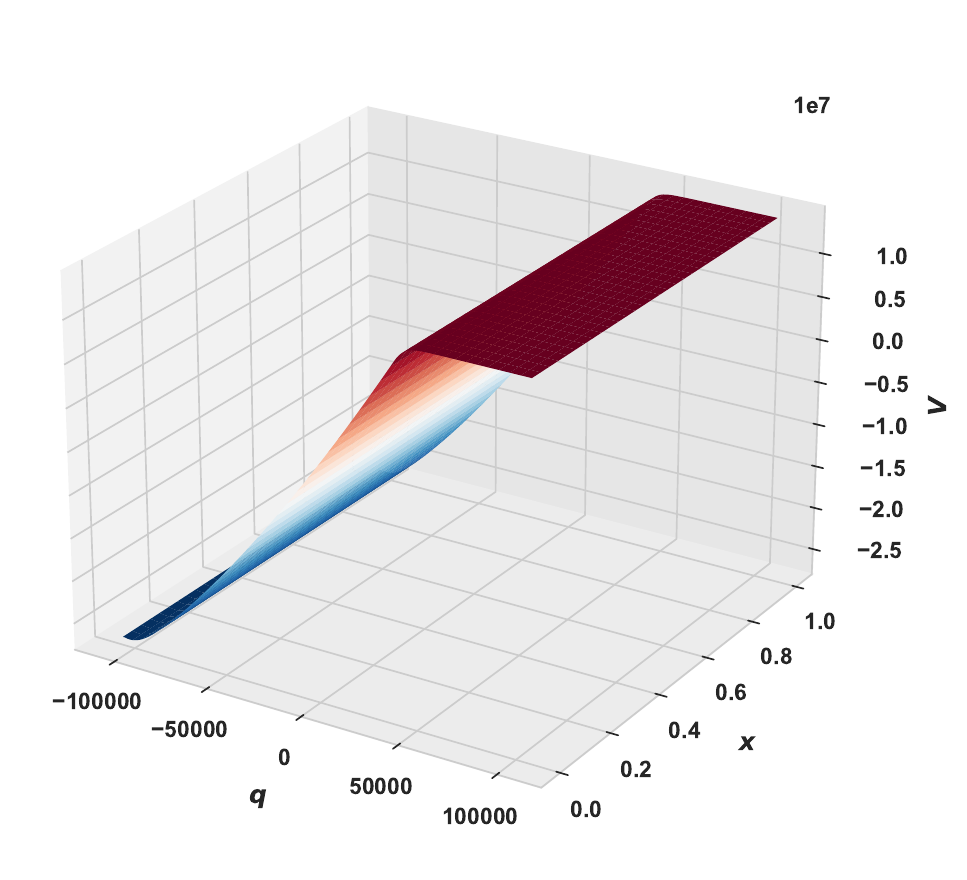}%
		\label{fig:V_xq_t0}}
	\hfill
	\subfloat[$(x, y)$ plane, $q = q_{\mathrm{mid}}$]{%
		\includegraphics[width=0.32\textwidth]{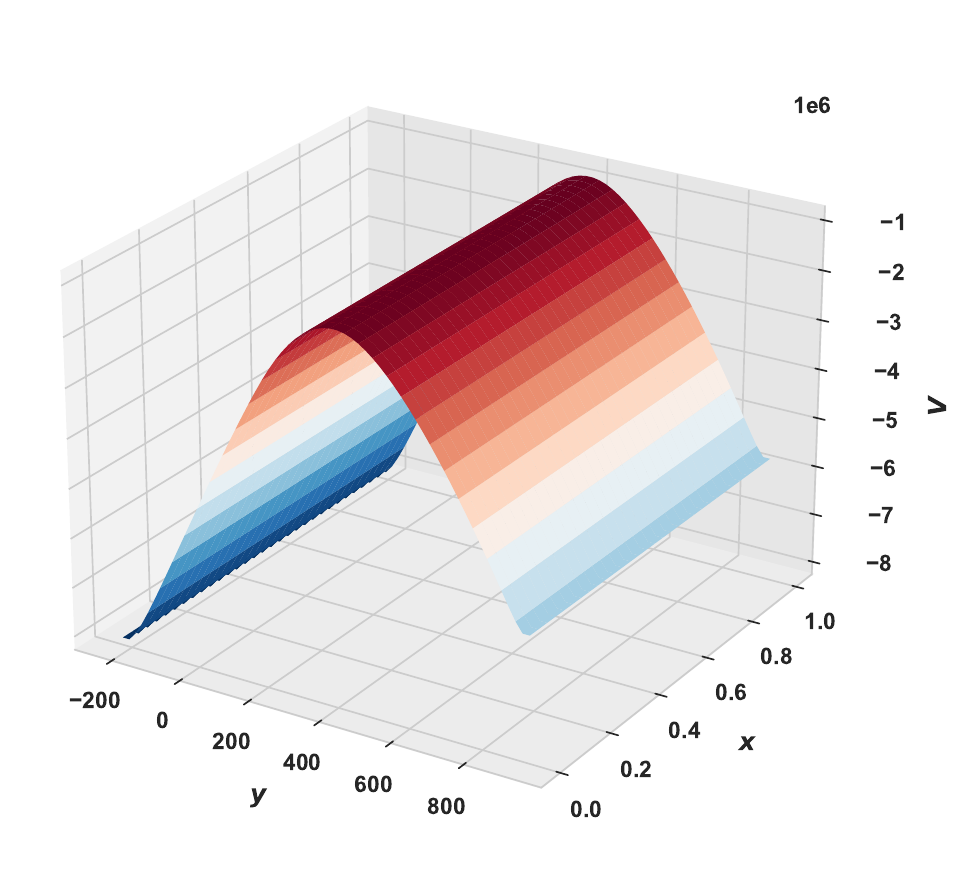}%
		\label{fig:V_xy_t0}}
	\hfill
	\subfloat[$(y, q)$ plane, $x = x_{\mathrm{mid}}$]{%
		\includegraphics[width=0.32\textwidth]{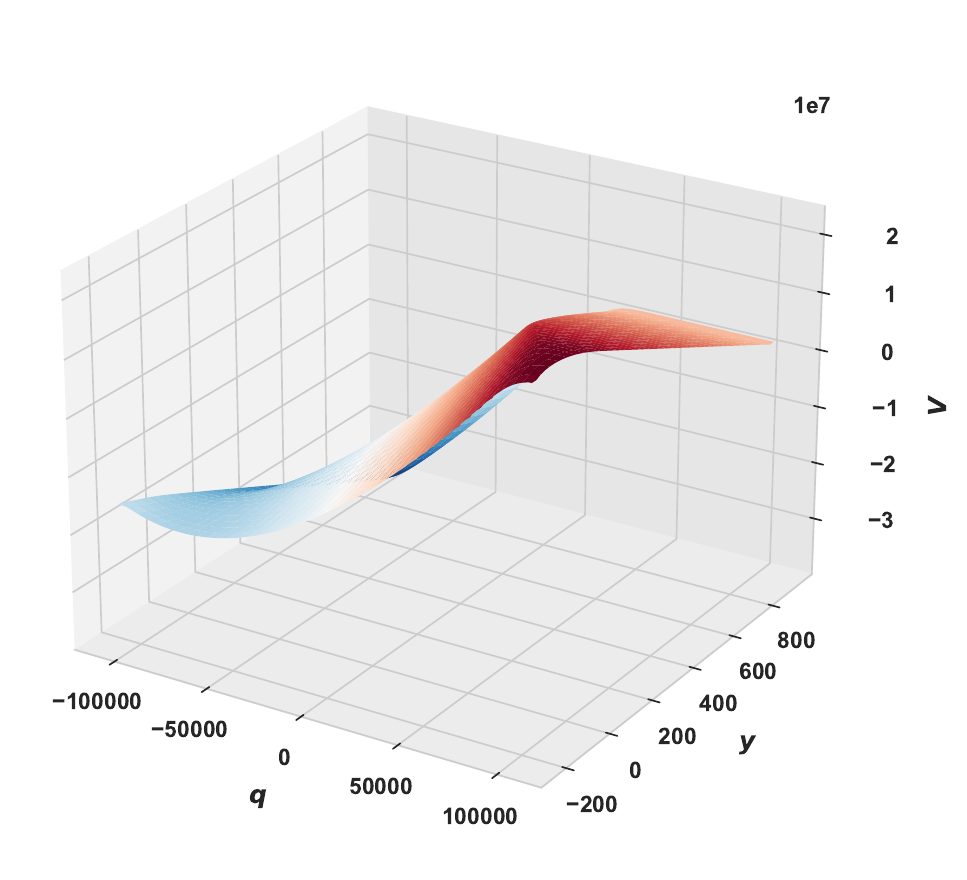}%
		\label{fig:V_yq_t0}}
	\caption{Cross-sections of the numerical value function
		$V^{\mathrm{III}}(0, x_i, y_j, q_k)$ at $t = 0$, each with one
		state variable fixed at the midpoint of its computational domain.
		Trading day: 2022-12-11.}
	\label{fig:V_surfaces_t0}
\end{figure}

{\section*{Declarations}}

\noindent{\textbf{Funding.} This work was supported by the Alexander von
Humboldt Foundation. Michael Samet acknowledges support by the Helmholtz
School for Data Science in Life, Earth and Energy (HDS-LEE).}

\noindent{\textbf{Competing interests.} The authors declare no competing
interests.}

\noindent{\textbf{Ethics approval.} Not applicable.}

\noindent{\textbf{Data availability.} The data used in this study are
publicly available from the SMARD platform at
\href{https://www.smard.de/en/downloadcenter/download-market-data}{https://www.smard.de}.
The data sources and preprocessing steps are described in
Section~\ref{sec:data_description_parameters}.}

\noindent{\textbf{Code availability.} The source code used to produce the
numerical results will be made publicly available in a GitHub repository
upon publication. Until then, it is available from the corresponding author
upon reasonable request.}

\noindent{\textbf{Use of artificial intelligence.} The authors edited the
manuscript with AI assistance and take full responsibility for the content
of the manuscript.}


{
\bibliographystyle{plain}
\bibliography{bibliography} 
}
\appendix
{
\section{Proof of Lemma~\ref{lem:stageII_terminal_lipschitz}}
\label{app:proof_stageII_terminal_lipschitz}

\begin{proof}
Let $x,x'\in[0,1]$, and let
$X^{\varepsilon,x}$ and $X^{\varepsilon,x'}$ denote the regularized
production processes driven by the same Brownian motion, with
$$
X_{T_{\mathrm{gc}}}^{\varepsilon,x}=x,
\qquad
X_{T_{\mathrm{gc}}}^{\varepsilon,x'}=x'.
$$
Under Assumptions~\ref{ass:lipschitz_muX}
and~\ref{ass:lipschitz_sigmaX}, existence and strong uniqueness follow
from \cite{karatzas2014brownian}.

We define
$$
\Delta X_s
:=
X_s^{\varepsilon,x}
-
X_s^{\varepsilon,x'},
\qquad
s\in[T_{\mathrm{gc}},T].
$$
Without loss of generality, assume that $x\geq x'$. By the comparison
theorem for one-dimensional stochastic differential equations,
$$
X_s^{\varepsilon,x}
\geq
X_s^{\varepsilon,x'}
\quad\text{a.s. for all }s\in[T_{\mathrm{gc}},T].
$$
Hence, $\Delta X_s\geq0$ almost surely. Since the production drift is
affine in the state variable,
$$
\mu_X(r,X_r^{\varepsilon,x})
-
\mu_X(r,X_r^{\varepsilon,x'})
=
-\theta(r)\Delta X_r.
$$
It follows that
$$
\begin{aligned}
\Delta X_s
={}&
x-x'
-
\int_{T_{\mathrm{gc}}}^s
\theta(r)\Delta X_r\,\mathrm{d}r
\\
&+
\int_{T_{\mathrm{gc}}}^s
\left(
\sigma_X^\varepsilon(X_r^{\varepsilon,x})
-
\sigma_X^\varepsilon(X_r^{\varepsilon,x'})
\right)
\,\mathrm{d}B_r^X.
\end{aligned}
$$
Under Assumption~\ref{ass:lipschitz_sigmaX}, the stochastic integral is
square integrable and therefore has zero expectation. Thus,
$$
\mathbb{E}[\Delta X_s]
=
x-x'
-
\int_{T_{\mathrm{gc}}}^s
\theta(r)\mathbb{E}[\Delta X_r]\,\mathrm{d}r
=
\exp\left(
-\int_{T_{\mathrm{gc}}}^s
\theta(r)\,\mathrm{d}r
\right)
(x-x').
$$
Since $\Delta X_s\geq0$ almost surely and
$\theta(r)\geq\theta_0>0$ by \eqref{eq:mean_reversion_condition},
\begin{equation}
\label{eq:regularized_production_l1_contraction}
\mathbb{E}\left[|\Delta X_s|\right]
=
\exp\left(
-\int_{T_{\mathrm{gc}}}^s
\theta(r)\,\mathrm{d}r
\right)
|x-x'|
\leq
|x-x'|,
\qquad
s\in[T_{\mathrm{gc}},T].
\end{equation}
The case $x<x'$ follows by interchanging $x$ and $x'$.

Combining the Stage~II interface condition with the Stage~I
representation and applying the tower property gives
\begin{equation}
\label{eq:regularized_stageII_terminal_representation}
v^{\mathrm{II},\varepsilon}(T_{\mathrm{gc}},x;\,q)
=
\mathbb{E}\left[
g\left(
q
-
P_{\max}
\int_{T_{\mathrm{gc}}+h}^{T}
X_s^{\varepsilon,x}\,\mathrm{d}s
\right)
\right].
\end{equation}
Let $q,q'\in\mathbb{R}$. Using
\eqref{eq:regularized_stageII_terminal_representation} for the initial
states $(x,q)$ and $(x',q')$, together with the Lipschitz continuity of
$g(\xi)=\beta|\xi|$, gives
$$
\begin{aligned}
&
\left|
v^{\mathrm{II},\varepsilon}(T_{\mathrm{gc}},x;\,q)
-
v^{\mathrm{II},\varepsilon}(T_{\mathrm{gc}},x';\,q')
\right|
\\
&\leq
\beta
\mathbb{E}\left[
\left|
(q-q')
-
P_{\max}
\int_{T_{\mathrm{gc}}+h}^{T}
\Delta X_s\,\mathrm{d}s
\right|
\right]
\\
&\leq
\beta|q-q'|
+
\beta P_{\max}
\int_{T_{\mathrm{gc}}+h}^{T}
\mathbb{E}\left[|\Delta X_s|\right]
\,\mathrm{d}s.
\end{aligned}
$$
Applying \eqref{eq:regularized_production_l1_contraction} and using
$T=T_{\mathrm{gc}}+h+L$, we obtain
$$
\begin{aligned}
&
\left|
v^{\mathrm{II},\varepsilon}(T_{\mathrm{gc}},x;\,q)
-
v^{\mathrm{II},\varepsilon}(T_{\mathrm{gc}},x';\,q')
\right|
\\
&\leq
\beta P_{\max}L|x-x'|
+
\beta|q-q'|
\\
&\leq
\max\{\beta P_{\max}L,\beta\}
\left(
|x-x'|
+
|q-q'|
\right).
\end{aligned}
$$
The estimate stated in the lemma therefore holds with
$$
C=\max\{\beta P_{\max}L,\beta\}.
$$
\end{proof}
}

\section{Proof of Theorem~\ref{thm:viscosity_reg}}
\label{app:proof_well_posedness_reg}
    
\begin{proof}
After the sign change
$
\tilde v^\varepsilon:=-v^{\mathrm{III},\varepsilon},
$
the regularized Stage~III problem is cast as a finite-horizon maximization
problem for a controlled jump-diffusion of the type studied
in~\cite{pham1998optimal}, with state variable
$
\boldsymbol{\xi}=(x,y,q)\in[0,1]\times\mathbb{R}^2,
$
drift components $\mu_X(t,x)$, $\mu_Y(t,y)$, and $\psi$,
diffusion coefficients $\sigma_X^\varepsilon(x)$ and $\sigma>0$,
running reward $-\ell(y,\psi)$, terminal
reward $-v^{\mathrm{II},\varepsilon}(T_{\mathrm{gc}},x;q)$, and jump amplitude
$
\Gamma(z)
$. In our considered framework there is no discount factor as in \cite{pham1998optimal}.
Our setting is a pure control problem with deterministic terminal time,
hence it falls within the framework of~\cite{pham1998optimal} as a
special case of their controlled jump-diffusion problem with no stopping. In the following paragraphs we verify that the conditions of Assumptions~\ref{ass:lipschitz}-\ref{ass:cost_lip} hold and hence fit the framework of~\cite{pham1998optimal}. Since the study of~\cite{pham1998optimal} is done on the state space $\mathbb{R}^d$ for some $d\in\mathbb{N}$, whereas the state variable $x$ in our setting takes values in $[0,1]$ (see Section~\ref{sec:wind_dynamics}), we reformulate our setting by a standard extension argument. We define the extended drift and diffusion coefficients and the extended terminal cost by
$$
\widetilde\mu_X(t,x):=
\begin{cases}
\mu_X(t,0), & x<0,\\
\mu_X(t,x), & x\in[0,1],\\
\mu_X(t,1), & x>1,
\end{cases}
\qquad
\widetilde\sigma_X^\varepsilon(x):=
\begin{cases}
0, & x<0,\\
\sigma_X^\varepsilon(x), & x\in[0,1],\\
0, & x>1,
\end{cases}
$$
and
$$
\widetilde v^{\mathrm{II},\varepsilon}(T_{\mathrm{gc}},x;q):=
\begin{cases}
v^{\mathrm{II},\varepsilon}(T_{\mathrm{gc}},0;q), & x<0,\\
v^{\mathrm{II},\varepsilon}(T_{\mathrm{gc}},x;q), & x\in[0,1],\\
v^{\mathrm{II},\varepsilon}(T_{\mathrm{gc}},1;q), & x>1.
\end{cases}
$$
The extended coefficients and extended terminal cost coincide with the original functions on $[0,1]$, and are set to be constant, equal to the values at the boundary outside $[0,1]$. Our aim in what follows is to verify Assumptions~\ref{ass:lipschitz}-\ref{ass:cost_lip}.

The continuity assertion in Assumption~\ref{ass:continuity} follows directly
from the construction of the model. More specifically, we have that
$$
p_X(\cdot),p_Y(\cdot)\in C^1([0,T_{\mathrm{gc}}]),
\qquad
\theta(\cdot)\in C([0,T_{\mathrm{gc}}]),
$$
Consequently, the drift functions $\mu_X(t,x)$ and $\mu_Y(t,y)$ are
continuous in $(t,x)$ and $(t,y)$, respectively. The regularized diffusion
coefficient $\sigma_X^\varepsilon(x)$ defined in~\eqref{eq:regularized_diffusion}
is continuous on $[0,1]$ by construction. The extended coefficients $\widetilde\mu_X(t,\cdot)$ and $\widetilde\sigma_X^\varepsilon$ are continuous on $\mathbb{R}$. For $\widetilde\mu_X(t,\cdot)$, this follows from the continuity of $\mu_X(t,\cdot)$ on $[0,1]$, while for $\widetilde\sigma_X^\varepsilon$ it follows from the continuity of $\sigma_X^\varepsilon$ on $[0,1]$ together with $\sigma_X^\varepsilon(0)=\sigma_X^\varepsilon(1)=0$.

We next verify Assumption~\ref{ass:lipschitz}. Since
$\theta(\cdot)\in C([0,T_{\mathrm{gc}}])$, it attains its maximum on
$[0,T_{\mathrm{gc}}]$. The drift term $\mu_X(t,x)$ satisfies
$$
|\mu_X(t,x)-\mu_X(t,x')|
=
\theta(t)|x-x'|
\le
\left(\max_{t\in[0,T_{\mathrm{gc}}]}\theta(t)\right)|x-x'|, \quad x,x'\in[0,1].
$$
Similarly, the drift term $\mu_Y(t,y)$ satisfies
$$
|\mu_Y(t,y)-\mu_Y(t,y')|
=
\kappa|y-y'|.
$$
For the regularized diffusion coefficient,
\eqref{eq:regularized_diffusion} shows that
$\sigma_X^\varepsilon(x)$ is piecewise differentiable and continuous on
$[0,1]$. On $[0,\varepsilon]$ and $[1-\varepsilon,1]$ it is linear
with slope
$$
\frac{\sigma_X(\varepsilon)}{\varepsilon}
=
\sqrt{\frac{2\alpha\theta_0(1-\varepsilon)}{\varepsilon}}.
$$
On $[\varepsilon,1-\varepsilon]$, one has
$\sigma_X^\varepsilon(x)=\sigma_X(x)=\sqrt{2\alpha\theta_0x(1-x)}$, and hence
$$
|(\sigma_X^\varepsilon)'(x)|
=
\frac{\sqrt{2\alpha\theta_0}|1-2x|}{2\sqrt{x(1-x)}}
\le
\frac{\sqrt{2\alpha\theta_0}(1-2\varepsilon)}{2\sqrt{\varepsilon(1-\varepsilon)}},
\quad x\in[\varepsilon,1-\varepsilon].
$$
Moreover, we have that
$$
\frac{1-2\varepsilon}{2\sqrt{\varepsilon(1-\varepsilon)}}
\le
\sqrt{\frac{1-\varepsilon}{\varepsilon}}, \quad \text{for all } \varepsilon\in(0,1/2).
$$ 
Hence, taking the maximum derivative for $x\in[0,1]$, we have that
$$
|\sigma_X^\varepsilon(x)-\sigma_X^\varepsilon(x')|
\le
\sqrt{\frac{2\alpha\theta_0(1-\varepsilon)}{\varepsilon}}|x-x'|,
\quad \text{for all } x,x'\in[0,1].
$$
We now prove that the same Lipschitz constants hold for the extended coefficients on $\mathbb{R}$. Set
$$
L_X:=\max_{t\in[0,T_{\mathrm{gc}}]}\theta(t),
\qquad
L_\sigma:=\sqrt{\frac{2\alpha\theta_0(1-\varepsilon)}{\varepsilon}}.
$$
Since $\widetilde\mu_X$ is constant on $(-\infty,0]$ and on $[1,\infty)$, only the cross-boundary cases must be checked. If $x\le 0\le x'\le 1$, then
$$
|\widetilde\mu_X(t,x)-\widetilde\mu_X(t,x')|
=
|\mu_X(t,0)-\mu_X(t,x')|
\le
L_X|x'|
\le
L_X|x'-x|.
$$
If $0\le x\le 1\le x'$, then similarly
$$
|\widetilde\mu_X(t,x)-\widetilde\mu_X(t,x')|
=
|\mu_X(t,x)-\mu_X(t,1)|
\le
L_X|1-x|
\le
L_X|x'-x|.
$$
If $x\le 0$ and $x'\ge 1$, then
$$
|\widetilde\mu_X(t,x)-\widetilde\mu_X(t,x')|
=
|\mu_X(t,0)-\mu_X(t,1)|
\le
L_X
\le
L_X|x'-x|,
$$
since $x'-x\ge 1$. Hence
$$
|\widetilde\mu_X(t,x)-\widetilde\mu_X(t,x')|
\le
L_X|x-x'|,
\qquad x,x'\in\mathbb{R}.
$$
The argument for $\widetilde\sigma_X^\varepsilon$ is identical, using that $\widetilde\sigma_X^\varepsilon$ vanishes on $(-\infty,0]\cup[1,\infty)$ and that $\sigma_X^\varepsilon$ is $L_\sigma$-Lipschitz on $[0,1]$. Thus
$$
|\widetilde\sigma_X^\varepsilon(x)-\widetilde\sigma_X^\varepsilon(x')|
\le
L_\sigma|x-x'|,
\qquad x,x'\in\mathbb{R}.
$$

The price volatility $\sigma_Y=\sigma$ is constant and therefore
trivially Lipschitz. The cross-diffusion term
$\rho\widetilde\sigma_X^\varepsilon(x)\sigma$ inherits Lipschitz continuity
in $x$ from $\widetilde\sigma_X^\varepsilon(x)$, with constant
$|\rho|\sigma\sqrt{\frac{2\alpha\theta_0(1-\varepsilon)}{\varepsilon}}$.
Hence Assumption~\ref{ass:lipschitz} holds for the extended coefficients on $\mathbb{R}$ with
$$
K_1
=
\max\left\{
\max_{t\in[0,T_{\mathrm{gc}}]}\theta(t),
\kappa,
\sqrt{\frac{2\alpha\theta_0(1-\varepsilon)}{\varepsilon}},
|\rho|\sigma\sqrt{\frac{2\alpha\theta_0(1-\varepsilon)}{\varepsilon}}
\right\}.
$$

We now verify Assumption~\ref{ass:levy}. The jump component is
a compound Poisson process with finite intensity $\lambda<\infty$,
hence $\nu(\mathbb{R})=\lambda<\infty$. For the
double-exponential random variable $Z$ we have that
the second moment is
$$
\mathbb{E}[Z^2]
= p_+\frac{2}{\eta_+^2} + p_-\frac{2}{\eta_-^2} < \infty.
$$
Consequently, $\int_{\mathbb{R}} |z|^2\nu(\mathrm{d}z)
= \lambda\mathbb{E}[Z^2] < \infty$, confirming
\eqref{eq:levy_finiteness}. In addition, since $\Gamma(z) = (0,z,0)^\top$ does
not depend on $(t,\boldsymbol{\xi},\psi)$ and
$|\Gamma(z)| = |z| = \varrho(z)$, the conditions
\eqref{eq:jump_amplitude_bound} follow trivially.

Finally, we verify Assumption~\ref{ass:cost_lip}. For the running
cost $\ell(y,\psi) = -\psi y + \tfrac{\gamma}{2}\psi^2$, we have
$$
|\ell(y,\psi)-\ell(y',\psi)|
= |\psi||y-y'|
\le \bar{\psi}|y-y'|,
$$
where $\bar{\psi}:=\max(|\psi_{\min}|,|\psi_{\max}|)$, so the
running cost is Lipschitz in $y$ with constant $\bar{\psi}$.

For the terminal cost
$v^{\mathrm{II},\varepsilon}(T_{\mathrm{gc}},x;q)$,
Lemma~\ref{lem:stageII_terminal_lipschitz} shows that there exists
$C>0$ such that
$$
|v^{\mathrm{II},\varepsilon}(T_{\mathrm{gc}},x;q)
- v^{\mathrm{II},\varepsilon}(T_{\mathrm{gc}},x';q')|
\le C(|x-x'|+|q-q'|), \quad x,x'\in[0,1].
$$
The argument for $\widetilde v^{\mathrm{II},\varepsilon}$ is identical, using the three cross-boundary cases $x\le 0\le x'\le 1$, $0\le x\le 1\le x'$, and $x\le 0\le 1\le x'$. Thus
$$
|\widetilde v^{\mathrm{II},\varepsilon}(T_{\mathrm{gc}},x;q)
-\widetilde v^{\mathrm{II},\varepsilon}(T_{\mathrm{gc}},x';q')|
\le
C(|x-x'|+|q-q'|),
\qquad x,x'\in\mathbb{R}.
$$

Consequently, Assumption~\ref{ass:cost_lip} holds with
$K_2=\max\{\bar\psi,C\}$.

The dynamic programming principle
(Proposition~3.1 of~\cite{pham1998optimal}), the viscosity
characterization (Theorem~3.1), and the comparison principle
(Theorem~4.1) then yield that $\tilde v^\varepsilon$ is the unique
viscosity solution of the corresponding HJB-PIDE on
$[0,T_{\mathrm{gc}})\times\mathbb{R}^3$ in the
class of continuous functions with at most linear growth and that are uniformly
continuous in $(x,y,q)$ uniformly in $t$. Moreover, by Proposition 3.3 in \cite{pham1998optimal}, the value
function is globally Lipschitz continuous in the  spatial variables, uniformly in $t$.

The process $\widetilde X_t^\varepsilon$ with extended coefficients and regularized diffusion has continuous sample paths. Under the mean-reversion condition~\eqref{eq:mean_reversion_condition}, one has
$$
\widetilde\mu_X(t,0)=\mu_X(t,0)=\dot p_X(t)+\theta(t)p_X(t)\ge 0,
\qquad
\widetilde\mu_X(t,1)=\mu_X(t,1)=\dot p_X(t)-\theta(t)(1-p_X(t))\le 0,
$$
for all $t\in[0,T_{\mathrm{gc}}]$. Moreover, by the definition of the piecewise extension and by~\eqref{eq:regularized_diffusion},
$$
\widetilde\sigma_X^\varepsilon(x)=0 \quad \text{for } x\notin(0,1),
\qquad
\widetilde\sigma_X^\varepsilon(0)=\widetilde\sigma_X^\varepsilon(1)=0.
$$
Therefore any exit from $[0,1]$ would have to occur through one of the boundary points. However, at $x=0$ and $x=1$ the diffusion vanishes, and outside $[0,1]$ the extended dynamics have drift directed inward, namely
$$
\widetilde\mu_X(t,x)=\widetilde\mu_X(t,0)\ge 0 \quad \text{for } x\le 0,
\qquad
\widetilde\mu_X(t,x)=\widetilde\mu_X(t,1)\le 0 \quad \text{for } x\ge 1.
$$
Consequently, since the diffusion vanishes at $\{0,1\}$, the boundary dynamics are purely deterministic outside $[0,1]$, and the inward-pointing drift prevents the process from crossing outward. Hence, by  the continuity of sample paths of $\tilde{X}^\varepsilon_t$, the process cannot leave $[0,1]$. Therefore $[0,1]$ is invariant for $\tilde{X}^\varepsilon_t$. Consequently, for every admissible control and every initial state $(x,y,q)\in[0,1]\times\mathbb{R}^2$, the extended process remains in the physical domain and therefore evolves according to the original regularized coefficients and terminal cost. The associated payoff functionals coincide, and the restriction of $\tilde v^\varepsilon$ to $[0,T_{\mathrm{gc}})\times[0,1]\times\mathbb{R}^2$ is exactly $-v^{\mathrm{III},\varepsilon}$, which concludes the proof.
\end{proof}

\section{Regularization of the Jacobi Diffusion}
\label{app:regularization}
Figure~\ref{fig:regularized_volatility} shows the global effect of the
regularization on the interval $x\in[0,1]$. For completeness, the localized
behavior near the boundary is illustrated in Figure~\ref{fig:regularized_zoom}.

\FloatBarrier
\begin{figure}[!ht]
	\centering
	\includegraphics[width=0.55\linewidth]{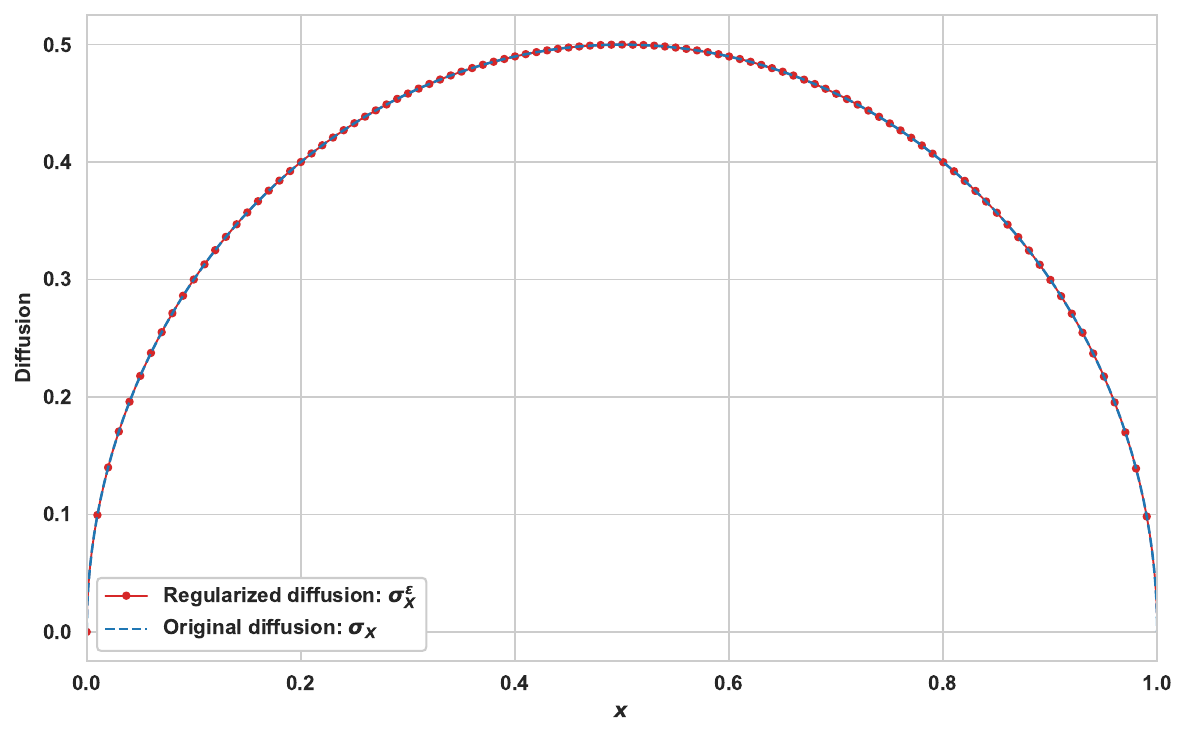}
	\caption{Original Jacobi diffusion $\sigma_X$ (dashed) and its regularized
		counterpart $\sigma_X^{\varepsilon}$ (solid) for $\varepsilon=10^{-4}$ on
		$x\in[0,1]$. By construction, $\sigma_X^{\varepsilon}=\sigma_X$ on
		$[\varepsilon,1-\varepsilon]$ and $\sigma_X^{\varepsilon}$ is linear on
		$[0,\varepsilon]$ and $[1-\varepsilon,1]$.}
	\label{fig:regularized_volatility}
\end{figure}

\FloatBarrier
\begin{figure}[!ht]
	\centering
	\includegraphics[width=0.55\linewidth]{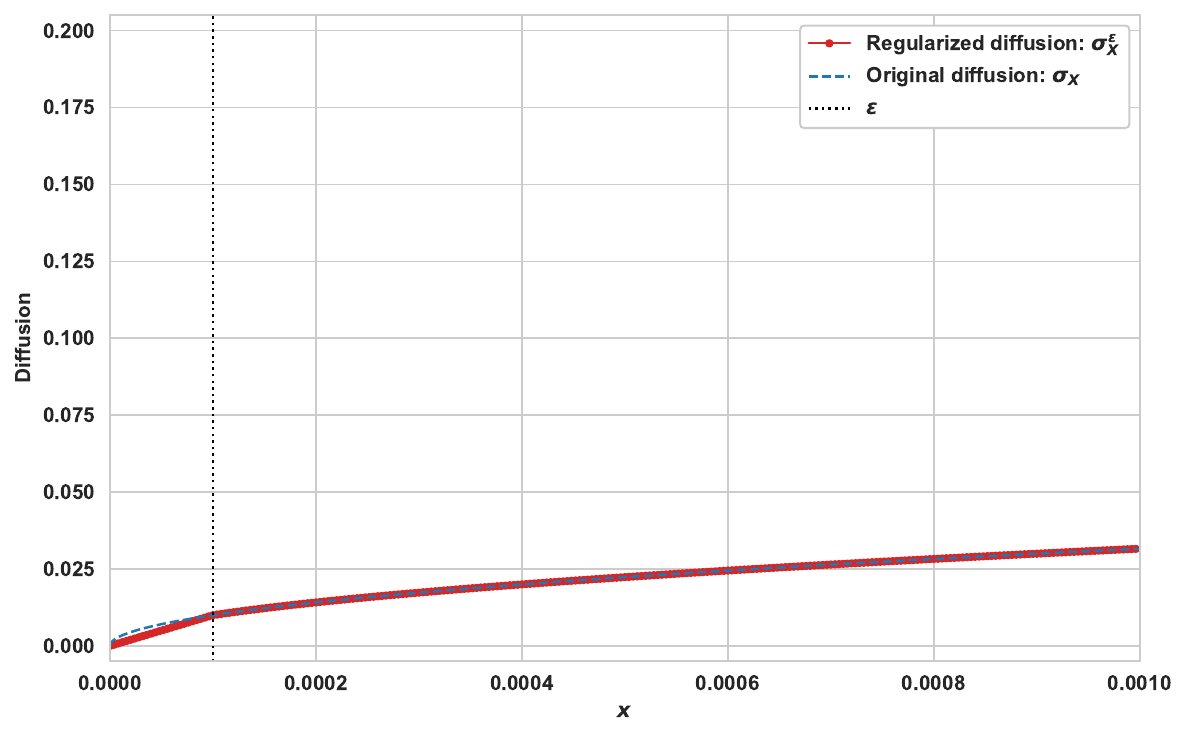}
	\caption{Zoom near the left boundary illustrating the linear behavior of the
		regularized volatility $\sigma_X^{\varepsilon}$ for
		$\varepsilon =10^{-4}$ near $x=0$.}
	\label{fig:regularized_zoom}
\end{figure}

Figure~\ref{fig:regularization_convergence} reports $\|V(0,x,y,q) - V^{\varepsilon}(0,x,y,q)\|_\infty$ as a function of $\varepsilon$. Both $V$ and $V^{\varepsilon}$ are computed on the same grid, with the same domain truncation, boundary treatment, and solver settings as used throughout Section~\ref{sec:numerical_experiments}, with the parameters of Table~\ref{tab:parameters}; the regularization parameter $\varepsilon$ is the only quantity varied.

The rapid decay observed in the
plot provides numerical evidence that the regularized solutions
$V^{\varepsilon}$ converge to $V$ as $\varepsilon \to 0$.

\begin{figure}[!ht]
	\centering
	\includegraphics[width=0.55\linewidth]{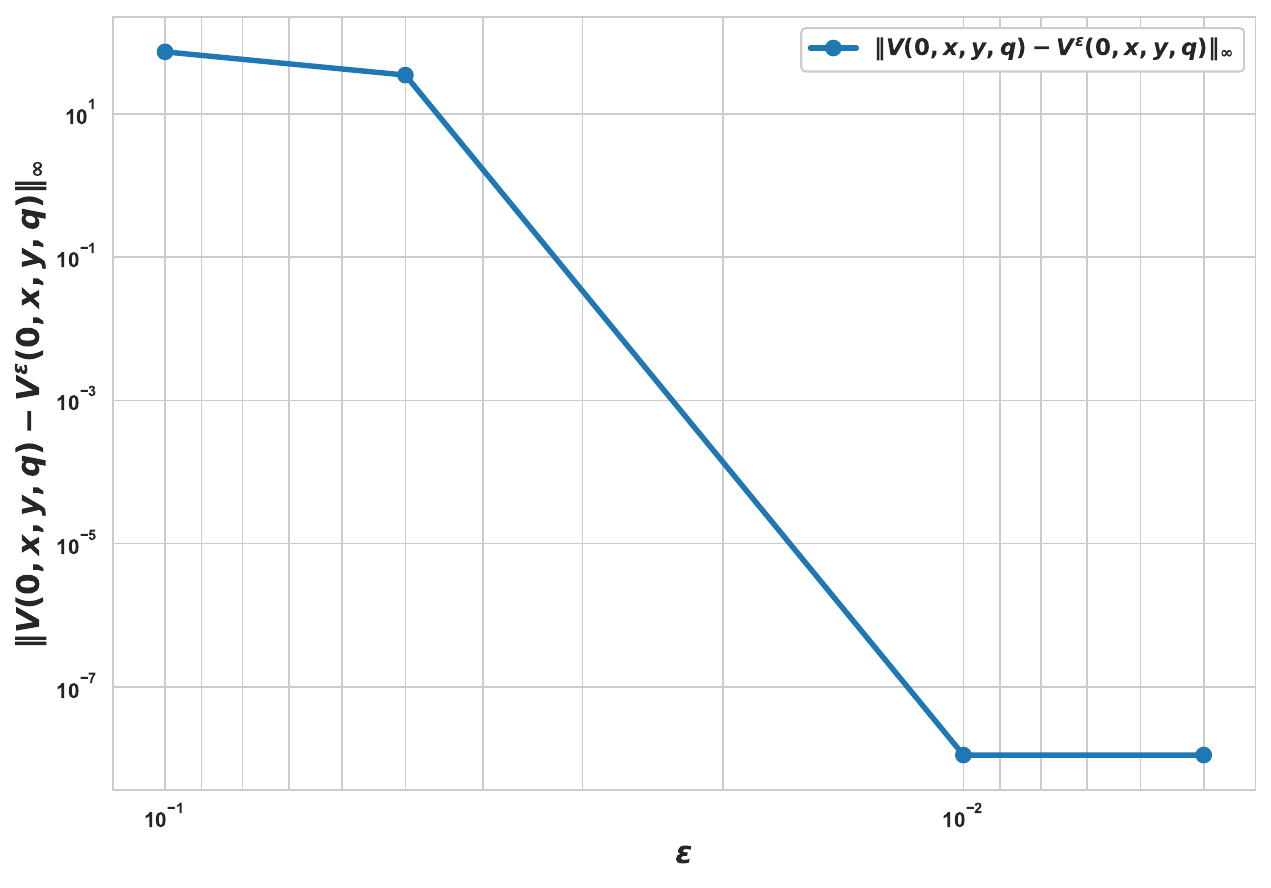}
	\caption{Sensitivity of $\|V(0,x,y,q) - V^{\varepsilon}(0,x,y,q)\|_\infty$ to the regularization parameter $\varepsilon$ for the trading day 2023-04-12. The reference solution $V$ is computed with the original diffusion coefficient $\sigma_X$, the regularized solution $V^{\varepsilon}$ with the coefficient $\sigma_X^{\varepsilon}$, and the difference is evaluated at $t=0$ over the full spatial grid.}
	\label{fig:regularization_convergence}
\end{figure}


{
\section{Proof of Lemma~\ref{lem:production_process_convergence}}
\label{app:proof_production_process_convergence}

\begin{proof}
For $a,b\geq0$, we have that $|\sqrt a-\sqrt b|^2\leq|a-b|$. Hence, for
$x,x'\in[0,1]$,
$$
\begin{aligned}
|\sigma_X(x)-\sigma_X(x')|^2
&\leq2\alpha\theta_0
\left|x(1-x)-x'(1-x')\right|\\
&=2\alpha\theta_0|x-x'|\,|1-x-x'|\\
&\leq2\alpha\theta_0|x-x'|.
\end{aligned}
$$
The coefficients $\sigma_X^\varepsilon$ and $\sigma_X$ agree on
$[\varepsilon,1-\varepsilon]$. On each boundary interval, their
difference is bounded by $\sigma_X(\varepsilon)$. Therefore,
$$
\left\|\sigma_X^\varepsilon-\sigma_X\right\|_{L^\infty(0,1)}
\leq\sqrt{2\alpha\theta_0\varepsilon}.
$$

Let $0\leq t\leq s\leq T$ and $x\in[0,1]$. For $r\in[t,s]$, set
$$
\Delta X_r^\varepsilon
:=
X_r^{\varepsilon,t,x}-X_r^{t,x}.
$$
Since the two equations have the same affine drift,
$$
\mu_X(r,X_r^{\varepsilon,t,x})
-
\mu_X(r,X_r^{t,x})
=
-\theta(r)\Delta X_r^\varepsilon.
$$

We use the Yamada-Watanabe approximation of the absolute value.
Let $\delta>0$ and $R>1$, and let
$\phi_{\delta,R}$ be a continuous nonnegative function supported
in $[\delta/R,\delta]$, satisfying
$$
\int_{\delta/R}^{\delta}
\phi_{\delta,R}(u)\,\mathrm{d}u
=
1,
\qquad
\phi_{\delta,R}(u)
\leq
\frac{2}{u\log R}.
$$
We further define $\varphi_{\delta,R}\in C^2(\mathbb{R})$ by
$$
\varphi_{\delta,R}(z)
:=
\int_0^{|z|}
\int_0^y
\phi_{\delta,R}(u)\,\mathrm{d}u\,\mathrm{d}y.
$$
Then
$$
|z|
\leq
\delta+\varphi_{\delta,R}(z),
\qquad
\left|\varphi_{\delta,R}'(z)\right|
\leq1,
\qquad
\varphi_{\delta,R}'(z)z\geq0,
$$
and
$$
0\leq
\varphi_{\delta,R}''(z)
\leq
\frac{2}{|z|\log R}
\mathbf{1}_{[\delta/R,\delta]}(|z|),
\qquad z\in\mathbb{R}\setminus\{0\}.
$$

Since $\Delta X_t^\varepsilon=0$, applying It\^o's formula and taking
expectations yields
$$
\begin{aligned}
\mathbb{E}\left[
\varphi_{\delta,R}(\Delta X_s^\varepsilon)
\right]
={}&-
\mathbb{E}\left[
\int_t^s\theta(r)
\varphi_{\delta,R}'(\Delta X_r^\varepsilon)
\Delta X_r^\varepsilon\,\mathrm{d}r
\right]\\
&+\frac12\mathbb{E}\left[
\int_t^s
\varphi_{\delta,R}''(\Delta X_r^\varepsilon)
\left|
\sigma_X^\varepsilon(X_r^{\varepsilon,t,x})
-\sigma_X(X_r^{t,x})
\right|^2\mathrm{d}r
\right].
\end{aligned}
$$
The stochastic integral has zero expectation because the diffusion
coefficients are bounded on $[0,1]$. Since $\theta(r)\geq0$, the first
term on the right hand side is nonpositive. Moreover,
$$
\begin{aligned}
&\left|
\sigma_X^\varepsilon(X_r^{\varepsilon,t,x})
-\sigma_X(X_r^{t,x})
\right|^2\\
&\qquad\leq
2\left\|\sigma_X^\varepsilon-\sigma_X
\right\|_{L^\infty(0,1)}^2
+4\alpha\theta_0|\Delta X_r^\varepsilon|.
\end{aligned}
$$
On the support of $\varphi_{\delta,R}''$,
$|\Delta X_r^\varepsilon|\geq\delta/R$. Consequently,
$$
\begin{aligned}
\mathbb{E}|\Delta X_s^\varepsilon|
\leq\delta+(s-t)\left(
\frac{2R}{\delta\log R}
\left\|\sigma_X^\varepsilon-\sigma_X
\right\|_{L^\infty(0,1)}^2
+\frac{4\alpha\theta_0}{\log R}
\right).
\end{aligned}
$$
Thus
$$
\begin{aligned}
&\sup_{\substack{0\leq t\leq s\leq T\\x\in[0,1]}}
\mathbb{E}\left[
\left|X_s^{\varepsilon,t,x}-X_s^{t,x}\right|
\right]\\
&\qquad\leq
\delta+T\left(
\frac{2R}{\delta\log R}
\left\|\sigma_X^\varepsilon-\sigma_X
\right\|_{L^\infty(0,1)}^2
+\frac{4\alpha\theta_0}{\log R}
\right).
\end{aligned}
$$
We first let $\varepsilon\downarrow0$, then
$R\uparrow\infty$, and finally $\delta\downarrow0$.
\end{proof}

\section{Proof of Proposition~\ref{prop:value_function_regularization}}
\label{app:proof_value_function_regularization}

\begin{proof}
The imbalance penalty $g(\xi)=\beta|\xi|$ is globally Lipschitz with
constant $\beta$. Let $t\in[T-L,T]$. By the conditional expectation
representation of the Stage~I value function and the Lipschitz
continuity of $g$,
\begin{equation}
\label{eq:stageI_regularization_estimate}
\begin{aligned}
&\left|v^{\mathrm{I},\varepsilon}(t,x,m;\,q)
-v^{\mathrm{I}}(t,x,m;\,q)\right|\\
&\qquad\leq
\beta P_{\max}\int_t^T
\mathbb{E}\left[
\left|X_s^{\varepsilon,t,x}-X_s^{t,x}\right|
\right]\mathrm{d}s\\
&\qquad\leq
\beta P_{\max}L
\sup_{\substack{0\leq r\leq u\leq T\\z\in[0,1]}}
\mathbb{E}\left[
\left|X_u^{\varepsilon,r,z}-X_u^{r,z}\right|
\right].
\end{aligned}
\end{equation}

Let $t\in[T_{\mathrm{gc}},T-L]$. By the Markov property of the
production process and the tower property, together with the
conditional expectation representations of Stages~I and~II,
\begin{equation}
\label{eq:stageII_regularization_estimate}
\begin{aligned}
&\left|v^{\mathrm{II},\varepsilon}(t,x;\,q)
-v^{\mathrm{II}}(t,x;\,q)\right|\\
&\qquad\leq
\beta P_{\max}\int_{T-L}^T
\mathbb{E}\left[
\left|X_s^{\varepsilon,t,x}-X_s^{t,x}\right|
\right]\mathrm{d}s\\
&\qquad\leq
\beta P_{\max}L
\sup_{\substack{0\leq r\leq u\leq T\\z\in[0,1]}}
\mathbb{E}\left[
\left|X_u^{\varepsilon,r,z}-X_u^{r,z}\right|
\right].
\end{aligned}
\end{equation}

Let $t\in[0,T_{\mathrm{gc}}]$ and fix
$\psi\in\mathcal{A}_{t,T}$. Consider the original and regularized
Stage~III problems driven by the same Brownian motions and Poisson
random measure. Since the regularization affects only the production
diffusion coefficient, the price processes coincide. Moreover, in both
systems,
\begin{equation}
\label{eq:inventory_common_control}
Q_s
=
q+\int_t^s\psi_r\,\mathrm{d}r,
\qquad s\in[t,T_{\mathrm{gc}}],
\end{equation}
so the inventory processes coincide as well. Consequently, the
running costs are identical. By the Markov property of the production
process and the tower property, together with the Stage~II and
Stage~I representations,
\begin{equation}
\label{eq:cost_functional_regularized_estimate}
\begin{aligned}
&\left|J^{\mathrm{III},\varepsilon}(t,x,y,q;\psi)
-J^{\mathrm{III}}(t,x,y,q;\psi)\right|\\
&\qquad\leq
\beta P_{\max}\int_{T-L}^T
\mathbb{E}\left[
\left|X_s^{\varepsilon,t,x}-X_s^{t,x}\right|
\right]\mathrm{d}s\\
&\qquad\leq
\beta P_{\max}L
\sup_{\substack{0\leq r\leq u\leq T\\z\in[0,1]}}
\mathbb{E}\left[
\left|X_u^{\varepsilon,r,z}-X_u^{r,z}\right|
\right].
\end{aligned}
\end{equation}
The bound in~\eqref{eq:cost_functional_regularized_estimate} is uniform
in $\psi\in\mathcal{A}_{t,T}$. Since the two value functions are
defined by taking the infimum over the same admissible control set
$\mathcal{A}_{t,T}$,
\begin{equation}
\label{eq:stageIII_regularization_estimate}
\begin{aligned}
&\left|v^{\mathrm{III},\varepsilon}(t,x,y,q)
-v^{\mathrm{III}}(t,x,y,q)\right|\\
&\qquad\leq
\sup_{\psi\in\mathcal{A}_{t,T}}
\left|J^{\mathrm{III},\varepsilon}(t,x,y,q;\psi)
-J^{\mathrm{III}}(t,x,y,q;\psi)\right|\\
&\qquad\leq
\beta P_{\max}L
\sup_{\substack{0\leq r\leq u\leq T\\z\in[0,1]}}
\mathbb{E}\left[
\left|X_u^{\varepsilon,r,z}-X_u^{r,z}\right|
\right].
\end{aligned}
\end{equation}

Taking the suprema over the corresponding time and state domains
in~\eqref{eq:stageI_regularization_estimate},
\eqref{eq:stageII_regularization_estimate}, and
\eqref{eq:stageIII_regularization_estimate}, and applying
Lemma~\ref{lem:production_process_convergence}, proves
Proposition~\ref{prop:value_function_regularization}.
\end{proof}

\section{Proof of Proposition~\ref{prop:price_bound}}
\label{sec:proof_prob_bound_Y}

\begin{proof}
Let
$$
D_t:=Y_t-p_Y(t),\qquad t\in [0,T_{\mathrm{gc}}].
$$
Subtracting $\dot p_Y(t)\,\mathrm{d}t$ from \eqref{eq:price_sde} and using
$D_0=Y_0-p_Y(0)=0$ yields
$$
\mathrm{d}D_t
=
-\kappa D_t\,\mathrm{d}t
+\sigma\,\mathrm{d}B_t^Y
+\mathrm{d}J_t,
\qquad
D_0=0.
$$

Hence, we can define the following decomposition of the process $D_t$
$$
D_t
=
\underbrace{\sigma\int_0^t e^{-\kappa(t-s)}\,\mathrm{d}B_s^Y}_{=:U_t}
\;+\;
\underbrace{\int_0^t e^{-\kappa(t-s)}\,\mathrm{d}J_s}_{=:\widetilde J_t}.
$$

By subadditivity of the supremum norm, we have that
$$
\mathbb{P}\!\left(
\sup_{0\le t\le T_{\mathrm{gc}}} |D_t| > K_U+K_J
\right)
\le
\mathbb{P}\!\left(
\sup_{0\le t\le T_{\mathrm{gc}}} |U_t| > K_U
\right)
+
\mathbb{P}\!\left(
\sup_{0\le t\le T_{\mathrm{gc}}} |\widetilde J_t| > K_J
\right).
$$

Since each finite linear combination of $(U_t)_{t\in[0,T_{\mathrm{gc}}]}$ is an It\^o integral with a deterministic integrand, it is Gaussian. Hence $U_t$ is a centered Gaussian process. Moreover, $U_t$ is the unique strong solution of the following SDE
$$
\mathrm{d}U_t=-\kappa U_t\,\mathrm{d}t+\sigma\,\mathrm{d}B_t^Y,
\qquad U_0=0,
$$
whose coefficients are globally Lipschitz, therefore $U_t$ has continuous sample paths \cite{karatzas2014brownian}, and its variance is given by
$$
\mathrm{Var}[U_t]
=
\sigma^2\int_0^t e^{-2\kappa(t-s)}\,\mathrm{d}s
=
\frac{\sigma^2}{2\kappa}\bigl(1-e^{-2\kappa t}\bigr),
$$
which is non-decreasing in $t$. Therefore,
$$
\sigma_{\sup}^2
:=
\sup_{0\le t\le T_{\mathrm{gc}}}\mathrm{Var}[U_t]
=
\frac{\sigma^2}{2\kappa}\bigl(1-e^{-2\kappa T_{\mathrm{gc}}}\bigr).
$$

Since $U_t$ is a continuous Gaussian process on the compact interval $[0,T_{\mathrm{gc}}]$, it is almost surely bounded, and we define
$$
m_U:=\mathbb{E}\!\left[\sup_{0\le t\le T_{\mathrm{gc}}} U_t\right]
$$

By the Borell-TIS inequality (see, e.g., \cite[Theorem~2.1.1]{adler2007random}), for every $u>0$, we have
$$
\mathbb{P}\!\left(
\sup_{0\le t\le T_{\mathrm{gc}}} U_t > m_U+u
\right)
\le
\exp\!\left(-\frac{u^2}{2\sigma_{\sup}^2}\right).
$$

Since $-U$ has the same law as $U$, the same estimate holds for $-U$ i.e.,
$$
\mathbb{P}\!\left(
\sup_{0\le t\le T_{\mathrm{gc}}} (-U_t) > m_U+u
\right)
\le
\exp\!\left(-\frac{u^2}{2\sigma_{\sup}^2}\right).
$$

Hence, by the union bound, we obtain
$$
\mathbb{P}\!\left(
\sup_{0\le t\le T_{\mathrm{gc}}} |U_t| > m_U+u
\right)
\le
2\exp\!\left(-\frac{u^2}{2\sigma_{\sup}^2}\right).
$$

Setting $u=\sigma_{\sup}\sqrt{2\ln(4/\varepsilon)}$ yields
$$
\mathbb{P}\!\left(
\sup_{0\le t\le T_{\mathrm{gc}}} |U_t| > K_U(\varepsilon)
\right)
\le
\frac{\varepsilon}{2}.
$$

The second part of the proof aims to derive an estimate for $\widetilde J_t$. Since $J_t$ is a compensated compound Poisson process, it has finite variation on compact intervals. Hence the integration by parts formula applies to
$$
\widetilde J_t=\int_0^t e^{-\kappa(t-s)}\,\mathrm{d}J_s,
$$
and yields
$$
\widetilde J_t
=
J_t-\kappa\int_0^t e^{-\kappa(t-s)}J_s\,\mathrm{d}s.
$$

Therefore,
$$
|\widetilde J_t|
\le
|J_t|
+\kappa\int_0^t e^{-\kappa(t-s)}|J_s|\,\mathrm{d}s
\le
\sup_{0\le s\le t}|J_s|
\left(1+\kappa\int_0^t e^{-\kappa(t-s)}\,\mathrm{d}s\right)
=
\sup_{0\le s\le t}|J_s|\,(2-e^{-\kappa t}),
$$
and hence using $2-e^{-\kappa t}\le 2$,
$$
\sup_{0\le t\le T_{\mathrm{gc}}}|\widetilde J_t|
\le
2\sup_{0\le t\le T_{\mathrm{gc}}}|J_t|.
$$

For every $\alpha\in(-\eta_-,\eta_+)$, we define
$$
{\mathcal M_t^{(\alpha)}}
:=
\exp\!\bigl(\alpha J_t-t\Lambda_c(\alpha)\bigr),
\qquad
\Lambda_c(\alpha)
=
\lambda\bigl(\mathbb{E}[e^{\alpha Z}]-1-\alpha\,\mathbb{E}[Z]\bigr).
$$

Since $e^x\ge 1+x$ for $x \in \mathbb{R}$, we have $\Lambda_c(\alpha)\ge 0$. Moreover, for $\alpha_+\in(0,\eta_+)$ and $x>0$, the inequality
$$
e^{\alpha_+J_t}\le {\mathcal M_t^{(\alpha_+)}}\,e^{T_{\mathrm{gc}}\Lambda_c(\alpha_+)}
$$
implies
$$
\left\{
\sup_{0\le t\le T_{\mathrm{gc}}} J_t>x
\right\}
\subset
\left\{
\sup_{0\le t\le T_{\mathrm{gc}}} {\mathcal M_t^{(\alpha_+)}}
\ge e^{\alpha_+x-T_{\mathrm{gc}}\Lambda_c(\alpha_+)}
\right\}.
$$

Since ${\mathcal M^{(\alpha_+)}}$ is a nonnegative càdlàg martingale, hence a right-continuous submartingale, the submartingale inequality in \cite[Theorem~3.8(i)]{karatzas2014brownian} yields
$$
\mathbb{P}\!\left(
\sup_{0\le t\le T_{\mathrm{gc}}} J_t>x
\right)
\le
e^{-\alpha_+x+T_{\mathrm{gc}}\Lambda_c(\alpha_+)}.
$$

Applying the same argument to $-J$ with $\alpha_-\in(0,\eta_-)$ gives
$$
\mathbb{P}\!\left(
\sup_{0\le t\le T_{\mathrm{gc}}} (-J_t)>x
\right)
\le
e^{-\alpha_-x+T_{\mathrm{gc}}\Lambda_c(-\alpha_-)}.
$$

Therefore, by the union bound,
$$
\mathbb{P}\!\left(
\sup_{0\le t\le T_{\mathrm{gc}}} |J_t|>x
\right)
\le
e^{-\alpha_+x+T_{\mathrm{gc}}\Lambda_c(\alpha_+)}
+
e^{-\alpha_-x+T_{\mathrm{gc}}\Lambda_c(-\alpha_-)}.
$$

Set $x=\frac{K_J(\varepsilon)}{2}$, the definition of $K_J(\varepsilon)$ ensures that each exponential term is at most $\varepsilon/4$. Hence
$$
\mathbb{P}\!\left(
\sup_{0\le t\le T_{\mathrm{gc}}} |J_t|>\frac{K_J(\varepsilon)}{2}
\right)
\le
\frac{\varepsilon}{2}.
$$

Combining this with the bound on $\widetilde J_t$ yields
$$
\mathbb{P}\!\left(
\sup_{0\le t\le T_{\mathrm{gc}}} |\widetilde J_t|>K_J(\varepsilon)
\right)
\le
\frac{\varepsilon}{2}.
$$

To conclude, combining the bounds for $U_t$ and $\widetilde J_t$ gives
$$
\mathbb{P}\!\left(
\sup_{0\le t\le T_{\mathrm{gc}}} |D_t|>K(\varepsilon)
\right)
\le
\varepsilon.
$$
\end{proof}

\section{{Proof of the characteristic equations in the inventory variable}}
\label{sec:characteristics}
\begin{proof}
To facilitate the method of characteristics analysis, we
introduce a regularization of the original terminal cost
$
g(\xi)=\beta|\xi|
$, defined by
$$
g^\delta(\xi):=\beta\left(\sqrt{\xi^2+\delta^2}-\delta\right),
\qquad \delta>0.
$$
This regularization approximates the original terminal cost from below, with an
 error of order $\delta$, 
\begin{equation}
\label{eq:gdelta_uniform_error}
0\le g(\xi)-g^\delta(\xi)\le \beta\delta
\qquad\text{for all }\xi\in\mathbb{R}.
\end{equation}
Hence $g^\delta\to g$ uniformly on $\mathbb{R}$ as $\delta\downarrow 0$. In
particular, the regularized problem differs from the original one by an arbitrarily small perturbation of the terminal cost, while the
advantage of $g^\delta$ is that it is differentiable and therefore allows us to
justify taking derivatives in calculations that follow. Moreover, the characteristic
bound derived later is independent of $\delta$, so the resulting truncation
interval does not depend on the regularization and remains valid as
$\delta\downarrow 0$. Thus the regularization is used only as a technical device, without affecting the generality of the determined inventory bounds.

{Let $\varepsilon\in(0,1/2)$ and $\delta>0$} and let
$v^{\mathrm{II},\varepsilon,\delta}$ denote the corresponding regularized
Stage~II value function. Let $(x,y)\in[0,1]\times[y_{\min},y_{\max}]$ be fixed, we study the reduced
time-reversed PDE in $(\tau,q)$, with
$
\tau=T_{\mathrm{gc}}-t
$, which reads as
\begin{equation}
\label{eq:char_reduced_PDE}
\frac{\partial \tilde v}{\partial \tau}
+
\frac{1}{2\gamma}
\left(
y-\frac{\partial \tilde v}{\partial q}
\right)^2
=
0,
\qquad
\tau\in[0,T_{\mathrm{gc}}],
\end{equation}
with terminal cost
$$
\tilde v(0,q)=v^{\mathrm{II},\varepsilon,\delta}(T_{\mathrm{gc}},x;q).
$$

We first justify the differentiability of the terminal cost with respect to the
inventory variable $q$. Since the control is inactive on
$[T_{\mathrm{gc}},T]$, the inventory remains equal to $q$. Moreover, the
delivery mismatch is determined by the future accumulated production over the
interval $[T-L,T]$. Hence
\begin{equation}
\label{eq:stageII_terminal_representation}
v^{\mathrm{II},\varepsilon,\delta}(T_{\mathrm{gc}},x;q)
=
\mathbb{E}\left[
g^\delta\left(q-P_{\max}\int_{T-L}^{T} X_s\,ds\right)
\,\middle|\, X_{T_{\mathrm{gc}}}=x
\right].
\end{equation}
Moreover, we have that
\begin{equation}
\label{eq:gdelta_derivative}
(g^\delta)'(\xi)=\beta\frac{\xi}{\sqrt{\xi^2+\delta^2}},
\qquad
\left|(g^\delta)'(\xi)\right|\le\beta
\qquad\text{for all }\xi\in\mathbb{R}.
\end{equation}

By the Leibniz rule for differentiation under the integral sign, applied to
\eqref{eq:stageII_terminal_representation}, we obtain
\begin{equation}
\label{eq:terminal_q_derivative}
\partial_q v^{\mathrm{II},\varepsilon,\delta}(T_{\mathrm{gc}},x;q)
=
\mathbb{E}\left[
(g^\delta)'\left(q-P_{\max}\int_{T-L}^{T} X_s\,ds\right)
\,\middle|\, X_{T_{\mathrm{gc}}}=x
\right].
\end{equation}
Indeed, for each realization of the path $(X_s)_{s\in[T-L,T]}$, the map
$$
q\mapsto
g^\delta\left(q-P_{\max}\int_{T-L}^{T} X_s\,ds\right)
$$
is differentiable, with derivative
$$
(g^\delta)'\left(q-P_{\max}\int_{T-L}^{T} X_s\,ds\right).
$$

To justify differentiation under the conditional expectation in
\eqref{eq:stageII_terminal_representation}, fix $x\in[0,1]$ and let
$\nu_X$ denote the conditional law of
$(X_s)_{s\in[T-L,T]}$ given $X_{T_{\mathrm{gc}}}=x$. Then
\eqref{eq:stageII_terminal_representation} can be written as
\begin{equation}
\label{eq:stageII_terminal_representation_measure}
v^{\mathrm{II},\varepsilon,\delta}(T_{\mathrm{gc}},x;q)
=
\int
g^\delta\left(q-P_{\max}\int_{T-L}^{T} X_s(\omega)\,ds\right)
\,\nu_X(d\omega).
\end{equation}
For each fixed $\omega$, the map
$$
q\mapsto
g^\delta\left(q-P_{\max}\int_{T-L}^{T} X_s(\omega)\,ds\right)
$$
is differentiable, with derivative
$$
(g^\delta)'\left(q-P_{\max}\int_{T-L}^{T} X_s(\omega)\,ds\right).
$$
Moreover, by \eqref{eq:gdelta_derivative},
$$
\left|
(g^\delta)'\left(q-P_{\max}\int_{T-L}^{T} X_s(\omega)\,ds\right)
\right|
\le \beta
\qquad\text{for all }q\in\mathbb{R}\text{ and all }\omega.
$$
Since $\nu_X$ is a probability measure, the constant $\beta$ is
$\nu_X$-integrable. Therefore the assumptions of the Leibniz rule for
differentiation under the integral sign are satisfied, and we obtain
\begin{equation}
\label{eq:terminal_q_derivative}
\partial_q v^{\mathrm{II},\varepsilon,\delta}(T_{\mathrm{gc}},x;q)
=
\mathbb{E}\left[
(g^\delta)'\left(q-P_{\max}\int_{T-L}^{T} X_s\,ds\right)
\,\middle|\, X_{T_{\mathrm{gc}}}=x
\right].
\end{equation}
Consequently, using \eqref{eq:gdelta_derivative},
\begin{equation}
\label{eq:psi2_bound}
\left|
\partial_q v^{\mathrm{II},\varepsilon,\delta}(T_{\mathrm{gc}},x;q)
\right|
\le \beta
\qquad\text{for every }q\in\mathbb{R}.
\end{equation}

We define
$$
p_1:=\partial_\tau \tilde v,
\qquad
p_2:=\partial_q \tilde v,
$$
Hence, the reduced PDE can be rewritten as
$$
p_1+\frac{1}{2\gamma}(y-p_2)^2=0.
$$
The characteristic system is then given by
$$
\begin{cases}
\dot\tau(s)=1,\\[4pt]
\dot q(s)=-\dfrac{1}{\gamma}\bigl(y-p_2\bigr),\\[8pt]
\dot p_1(s)=0,\\[4pt]
\dot p_2(s)=0,
\end{cases}
\qquad
\tau(0)=0,
\qquad
q(0)=r,
\qquad
p_2(0)=\psi_2(r),
$$
where
$
\psi_2(r)
=
\partial_q v^{\mathrm{II},\varepsilon,\delta}(T_{\mathrm{gc}},x;r).
$
By \eqref{eq:psi2_bound}, we have that
$$
|\psi_2(r)|\le\beta
\qquad\text{for every }r\in\mathbb{R}.
$$

Since $\dot p_2=0$, the slope is constant along each characteristic, consequently
$$
p_2(s)\equiv\psi_2(r).
$$
Integrating $\dot q$ gives
\begin{equation}
\label{eq:char_q}
q(s)=r-\frac{s}{\gamma}\bigl(y-\psi_2(r)\bigr),
\qquad s\in[0,T_{\mathrm{gc}}].
\end{equation}

Passing back to forward time by
$
t=T_{\mathrm{gc}}-s,
$
the corresponding forward characteristic speed is
$
\frac{y-\psi_2(r)}{\gamma}.
$
By setting $y\in[y_{\min},y_{\max}]$ together with \eqref{eq:psi2_bound}, we obtain
\begin{equation}
\label{eq:slope_envelope}
\frac{y_{\min}-\beta}{\gamma}
\le
\frac{y-\psi_2(r)}{\gamma}
\le
\frac{y_{\max}+\beta}{\gamma}.
\end{equation}

Hence, for any affine  trajectory starting from
$
Q_0=0
$
we have using \eqref{eq:slope_envelope}
$$
Q(t)=ct,
\qquad
c\in
\left[
\frac{y_{\min}-\beta}{\gamma},
\frac{y_{\max}+\beta}{\gamma}
\right],
\qquad
t\in[0,T_{\mathrm{gc}}],
$$
and therefore
\begin{equation}
\label{eq:Q_bounds_t}
\frac{t}{\gamma}(y_{\min}-\beta)
\le
Q(t)
\le
\frac{t}{\gamma}(y_{\max}+\beta),
\qquad
t\in[0,T_{\mathrm{gc}}].
\end{equation}

Since the two bounds in \eqref{eq:Q_bounds_t} are linear in $t$, their extrema
over $[0,T_{\mathrm{gc}}]$ are attained at the endpoints. Hence
\begin{equation}
\label{eq:Q_bounds_uniform}
T_{\mathrm{gc}}
\min\left\{\frac{y_{\min}-\beta}{\gamma},0\right\}
\le
Q(t)
\le
T_{\mathrm{gc}}
\max\left\{\frac{y_{\max}+\beta}{\gamma},0\right\},
\qquad
t\in[0,T_{\mathrm{gc}}].
\end{equation}
Accordingly, for any $\epsilon>0$, we define
\begin{equation}
q_{\min}
=
T_{\mathrm{gc}}
\min\left\{\frac{y_{\min}-\beta}{\gamma},0\right\}
-\epsilon,
\qquad
q_{\max}
=
T_{\mathrm{gc}}
\max\left\{\frac{y_{\max}+\beta}{\gamma},0\right\}
+\epsilon.
\end{equation}
\end{proof}
\vspace{-1cm}
{\section{{Proof of the Asymptotic Price Boundary Condition}}
\label{app:farfield}
\begin{proof}

In this section, we derive the affine condition imposed at the artificial
price endpoints $y \in \{y_{\min}, y_{\max}\}$. In the asymptotic regime
$|y| \to \infty$, the optimal control saturates, i.e., $\psi^* = \psi_{\max}$
as $y \to +\infty$ and $\psi^* = \psi_{\min}$ as $y \to -\infty$. If we let
$\bar\psi$ denote the saturated value, then the Hamiltonian reads as
\begin{equation}
\label{eq:app_saturated_hamiltonian}
H(y, \partial_q v) = -\bar\psi\,y + \tfrac{\gamma}{2}\bar\psi^2 + \bar\psi\,\partial_q v.
\end{equation}
Consequently, by linearity of the Hamiltonian in $y$  we assume the  following asymptotic ansatz 
\begin{equation}
\label{eq:app_affine_ansatz}
{v^{III}(T_{\mathrm{gc}}-\tau, x, y, q)} = a(\tau)\,y + r(\tau, x, q),
\qquad {a \in C^1([0,T_{\rm gc}])},
\quad r : [0, T_{\rm gc}] \times [0,1] \times [q_{\min}, q_{\max}] \to \mathbb{R}.
\end{equation}
{Substituting \eqref{eq:app_affine_ansatz} into the HJB
equation} and identifying the coefficient of $y$, we obtain
\begin{equation}
\label{eq:app_drift_contribution}
{\mu_Y(T_{\mathrm{gc}}-\tau, y)}\,\partial_y v
= {\bigl(\dot p_Y(T_{\mathrm{gc}}-\tau) - \kappa(y - p_Y(T_{\mathrm{gc}}-\tau))\bigr)} a(\tau)
= -\kappa\,a(\tau)\,y + \mathcal{O}(1).
\end{equation}
The coefficient of $y$  thus satisfies
\begin{equation}
\label{eq:app_farfield_slope_ode}
a'(\tau) + \kappa\,a(\tau) + \bar\psi = 0,
\qquad a(0) = 0,
\end{equation}
where $a(0) = 0$ follows from the terminal condition, which is
    independent of $y$. The unique solution of \eqref{eq:app_farfield_slope_ode} on $[0, T_{\rm gc}]$ is
\begin{equation}
\label{eq:app_farfield_slope_solution}
a(\tau) = -\bar\psi\,\frac{1 - e^{-\kappa\tau}}{\kappa}.
\end{equation}
Taking $\bar\psi = \psi_{\max}$ for $y \to +\infty$ and
$\bar\psi = \psi_{\min}$ for $y \to -\infty$ gives the two possible solutions
\begin{equation}
\label{eq:app_farfield_slopes}
a_+(\tau) = -\psi_{\max}\,\frac{1 - e^{-\kappa\tau}}{\kappa},
\qquad
a_-(\tau) = -\psi_{\min}\,\frac{1 - e^{-\kappa\tau}}{\kappa}.
\end{equation}
Motivated by \eqref{eq:app_farfield_slopes}, we impose the slopes
$a_\pm(\tau)$ at the truncation endpoints as the boundary condition
\eqref{eq:price_farfield_condition}.
\end{proof}
}

{\section{Proof of Lemma~\ref{lem:coefficient_nonnegativity}}
\label{app:proof_coefficient_nonnegativity}

\begin{proof}
Equating the off-diagonal entries in
\eqref{eq:positive_type_representation} gives
\begin{equation}
\label{eq:wxy_from_offdiagonal}
w_i^{xy}
=
\frac{|\rho|\,\sigma_X(x_i)\sigma}
{2s_i\Delta x\,\Delta y},
\end{equation}
which is nonnegative because $\rho<0$, $\sigma_X(x_i)>0$, $\sigma>0$,
$s_i\ge1$, $\Delta x>0$, and $\Delta y>0$. Equating the two diagonal
entries gives
\begin{equation}
\label{eq:wx_wy_from_diagonal}
w_i^x
=
\frac{\sigma_X(x_i)^2}{2(\Delta x)^2}
-
s_i^2 w_i^{xy},
\qquad
w_i^y
=
\frac{\sigma^2}{2(\Delta y)^2}
-
w_i^{xy}.
\end{equation}
Substituting \eqref{eq:wxy_from_offdiagonal} into
\eqref{eq:wx_wy_from_diagonal}, the inequality $w_i^x\ge0$ is equivalent
to $s_i \le R_i/|\rho|$, and the inequality $w_i^y\ge0$ is equivalent to
$s_i \ge |\rho|R_i$. Thus
\eqref{eq:coefficient_nonnegativity_condition} is necessary and
sufficient for all coefficients in
\eqref{eq:positive_type_representation} to be nonnegative.
\end{proof}}

{\section{Proof of Lemma~\ref{lem:jacobi_coefficient_nonnegativity}}
\label{app:proof_jacobi_coefficient_nonnegativity}

\begin{proof}
Since $x_i=i\Delta x$ and $\Delta x=1/N_x$, the Jacobi coefficient gives
\begin{equation}
\label{eq:jacobi_R_identity}
|\rho|R_i
=
\frac{|\rho|\sqrt{2\alpha\theta_0}}{\sigma}\,
\Delta y\,\sqrt{i(N_x-i)}.
\end{equation}
The lower bound in \eqref{eq:jacobi_lower_mesh_condition} is equivalent
to $R_1\ge|\rho|$. Since $R_i\ge R_1$ for all
$i=1,\dots,N_x-1$, we have $R_i\ge|\rho|$ for all such $i$. Therefore
$|\rho|R_i\le s_i$ follows from $s_i=\lceil|\rho|R_i\rceil$.

It remains to prove $s_i\le R_i/|\rho|$. If
$|\rho|\le R_i\le1/|\rho|$, then $|\rho|R_i\le1$, so $s_i=1$, and
$s_i\le R_i/|\rho|$. If $R_i>1/|\rho|$, then
$$
s_i=\left\lceil|\rho|R_i\right\rceil\le|\rho|R_i+1.
$$
Since $0<|\rho|\le1/\sqrt2$ and $R_i>1/|\rho|$,
$$
R_i\left(\frac{1}{|\rho|}-|\rho|\right)
>
\frac{1}{|\rho|^2}-1
\ge1.
$$
Thus $|\rho|R_i+1\le R_i/|\rho|$, and hence
$s_i\le R_i/|\rho|$. The claim follows from
Lemma~\ref{lem:coefficient_nonnegativity}.
\end{proof}}

{\section{Proof of Lemma~\ref{lem:jacobi_interior_oblique_stencil}}
\label{app:proof_jacobi_interior_oblique_stencil}

\begin{proof}
By \eqref{eq:jacobi_R_identity}, condition
\eqref{eq:jacobi_upper_mesh_condition} is equivalent to
$|\rho|R_1\le1$. If \eqref{eq:oblique_stencil_condition} holds
for all $i=1,\dots,N_x-1$, then in particular $s_1\le1$. Since
$s_1=\lceil|\rho|R_1\rceil$, this implies $|\rho|R_1\le1$, and hence
\eqref{eq:jacobi_upper_mesh_condition} is satisfied.

We now prove the sufficiency of the condition. For $1\le i\le N_x/2$, the inequality
$|\rho|R_i\le i$ is equivalent to
$$
\frac{|\rho|\sqrt{2\alpha\theta_0}}{\sigma}\,
\Delta y
\sqrt{\frac{N_x-i}{i}}
\le1.
$$
The left-hand side is decreasing in $i$, and at $i=1$ it is bounded by
\eqref{eq:jacobi_upper_mesh_condition}. Hence $|\rho|R_i\le i$ for
$1\le i\le N_x/2$. Since $i \in \mathbb{N}$ and
$s_i=\lceil|\rho|R_i\rceil$, we get $s_i\le i$. Replacing $i$ by
$N_x-i$ gives $s_i\le N_x-i$ for $N_x/2<i\le N_x-1$. Therefore
$s_i\le\min\{i,N_x-i\}$ for all $i=1,\dots,N_x-1$.
\end{proof}}

\section{{Proof of Proposition~\ref{prop:xy_mmatrix}}}
\label{app:proof_xy_mmatrix}

{
\begin{proof}
We first address the monotonicity of the interior part of the stencil. Let $1\le i\le N_x-1$ and $2\le j\le N_y-2$. By
\eqref{eq:xy_transport_rates}, we have that the coefficients satisfy
\[
a_i^{x,W},\ a_i^{x,E},\ a_j^{y,S},\ a_j^{y,N}\ge 0,
\]
and by \eqref{eq:xy_diffusion_rates} together with
\eqref{eq:coefficient_nonnegativity_condition},
\[
\nu_i^x,\ \nu_i^y,\ \nu_i^{xy}\ge 0.
\]
Hence, the oblique-stencil indices $i\pm s_i,\ j\mp 1$ are admissible by
\eqref{eq:oblique_stencil_condition}. Substitution in
\eqref{eq:xy_matrix_coefficients} yields
\begin{equation*}
C^W_{i,j},\ C^E_{i,j},\ C^S_{i,j},\ C^N_{i,j},\
C^{SE}_{i,j},\ C^{NW}_{i,j}\le 0,
\end{equation*}
and
\begin{equation*}
C^C_{i,j}+C^W_{i,j}+C^E_{i,j}
+C^S_{i,j}+C^N_{i,j}
+C^{SE}_{i,j}+C^{NW}_{i,j}
=1.
\end{equation*}
Hence
\begin{equation*}
C^C_{i,j}-\bigl(|C^W_{i,j}|+|C^E_{i,j}|
+|C^S_{i,j}|+|C^N_{i,j}|
+|C^{SE}_{i,j}|+|C^{NW}_{i,j}|\bigr)
=1,
\end{equation*}
Consequently
\begin{equation*}
C^C_{i,j}=1+|C^W_{i,j}|+|C^E_{i,j}|
+|C^S_{i,j}|+|C^N_{i,j}|
+|C^{SE}_{i,j}|+|C^{NW}_{i,j}|
>0.
\end{equation*}
Solving \eqref{eq:xy_matrix_row} for ${W^{n+1}_{i,j,k}}$,
\begin{equation*}
{W^{n+1}_{i,j,k}}
=
\frac{1}{C^C_{i,j}}
\Bigl(
{U^{n+1}_{i,j,k}}
-C^W_{i,j}\,{W^{n+1}_{i-1,j,k}}
-C^E_{i,j}\,{W^{n+1}_{i+1,j,k}}
-C^S_{i,j}\,{W^{n+1}_{i,j-1,k}}
-C^N_{i,j}\,{W^{n+1}_{i,j+1,k}}
\end{equation*}
\begin{equation*}
-C^{SE}_{i,j}\,{W^{n+1}_{i+s_i,j-1,k}}
-C^{NW}_{i,j}\,{W^{n+1}_{i-s_i,j+1,k}}
\Bigr).
\end{equation*}
Since $C^C_{i,j}>0$ and $-C^W_{i,j},-C^E_{i,j},-C^S_{i,j},-C^N_{i,j},
-C^{SE}_{i,j},-C^{NW}_{i,j}\ge 0$, the right-hand side is nondecreasing
in ${U^{n+1}_{i,j,k}}$ and in each neighboring value of ${W^{n+1}}$. The
interior stencil is therefore monotone.

\smallskip
The verification for the numerical boundary conditions
consists of:
\begin{itemize}
\item the one-sided equations at $x=0$ and $x=1$, using
\eqref{eq:production_boundary_rows},
\eqref{eq:production_drift_boundary}, and
\eqref{eq:production_diffusion_boundary};
\item the affine boundary condition in $y$ at $j=0$ and $j=N_y$, using
\eqref{eq:affine_boundary_rows};
\item the sign and diagonal-dominance properties of the assembled
coefficient matrix $\mathbf A$, from which $\mathbf A^{-1}\ge 0$ and
the monotonicity of the full implicit $(x,y)$ solve follow.
\end{itemize}
\end{proof}
}


\section{Proof of Proposition~\ref{prop:loc_error}}
\label{app:proof_loc_error}

\begin{proof}
{Let $(t,x,q)$ be fixed, let $y\in[y_{\min},y_{\max}]$, and set
$\tau=T_{\mathrm{gc}}-t$. Define}
$$
{
\Delta\mathcal{J}_+(y)
:=
\mathcal{J}_+[v](t,x,y,q)
-
\widehat{\mathcal{J}}_+[v](t,x,y,q),
\qquad
\Delta\mathcal{J}_-(y)
:=
\mathcal{J}_-[v](t,x,y,q)
-
\widehat{\mathcal{J}}_-[v](t,x,y,q).
}
$$
{From \eqref{eq:jump_decomp} and \eqref{eq:jump_localized}, the terms
$-\lambda v(t,x,y,q)$ cancel, and therefore}
\begin{equation}\label{eq:loc_decomp_app}
{
\mathcal{I}_y[v](t,x,y,q)
-
\widehat{\mathcal{I}}_y[v](t,x,y,q)
=
\lambda\Bigl(
p_+\,\Delta\mathcal{J}_+(y)
+
p_-\,\Delta\mathcal{J}_-(y)
\Bigr).
}
\end{equation}

{We first consider the positive semi-infinite integral. Let
$d_+:=y_{\max}-y$. From \eqref{eq:Jplus_def} and
\eqref{eq:Jplus_tr},}
\begin{align*}
{\Delta\mathcal{J}_+(y)}
&{
=
\eta_+
\int_{d_+}^{\infty}
\Bigl[
v(t,x,y+z,q)
-
v(t,x,y_{\max},q)
-
a_+(\tau)(z-d_+)
\Bigr]
e^{-\eta_+z}\,\mathrm{d}z.}
\end{align*}
{With the change of variable $z=d_++u$, we have
$y+z=y_{\max}+u$, and hence}
\begin{equation}
\label{eq:positive_exterior_error_app}
{
\Delta\mathcal{J}_+(y)
=
e^{-\eta_+d_+}\eta_+
\int_0^\infty
R_+(u)e^{-\eta_+u}\,\mathrm{d}u.
}
\end{equation}

{For the negative semi-infinite integral, let
$d_-:=y-y_{\min}$. From \eqref{eq:Jminus_def} and
\eqref{eq:Jminus_tr},}
\begin{align*}
{\Delta\mathcal{J}_-(y)}
&{
=
\eta_-
\int_{-\infty}^{-d_-}
\Bigl[
v(t,x,y+z,q)
-
v(t,x,y_{\min},q)
-
a_-(\tau)(z+d_-)
\Bigr]
e^{\eta_-z}\,\mathrm{d}z.}
\end{align*}
{Setting $z=-d_--u$ gives $y+z=y_{\min}-u$, and therefore}
\begin{equation}
\label{eq:negative_exterior_error_app}
{
\Delta\mathcal{J}_-(y)
=
e^{-\eta_-d_-}\eta_-
\int_0^\infty
R_-(u)e^{-\eta_-u}\,\mathrm{d}u.
}
\end{equation}

{Substituting
\eqref{eq:positive_exterior_error_app} and
\eqref{eq:negative_exterior_error_app} into
\eqref{eq:loc_decomp_app} gives the identity
\eqref{eq:loc_error_bound}. Applying the triangle inequality yields the
stated bound.}
\end{proof}

\section{{Proof of Proposition~\ref{prop:ode_jump}}}
\begin{proof}
{Let $(t,x,q)$ be fixed and set
$\tau=T_{\mathrm{gc}}-t$. We first derive the differential relation for
the interior contribution $\mathcal{J}_{+,\mathrm{loc}}[v]$ and then add
the contribution from the linear continuation.}

Let $y\in[y_{\min},y_{\max}]$.  The change of variable $u=y+z$ in
\eqref{eq:Jplus_loc} gives
\begin{equation}\label{eq:Jplus_loc_factored_app}
\mathcal{J}_{+,\mathrm{loc}}[v](t,x,y,q)
=
\eta_+\, e^{\eta_+ y}
\int_y^{y_{\max}} v(t,x,u,q)\,e^{-\eta_+ u}\,\mathrm{d}u.
\end{equation}
Let
$$
F_+(y):=\int_y^{y_{\max}} v(t,x,u,q)\,e^{-\eta_+ u}\,\mathrm{d}u.
$$
Since $y\mapsto v(t,x,y,q)$ is continuous on $[y_{\min},y_{\max}]$ by
assumption, the integrand $u\mapsto v(t,x,u,q)\,e^{-\eta_+ u}$ is continuous.
By the fundamental theorem of calculus, $F_+\in C^1([y_{\min},y_{\max}])$ and
$$
F_+'(y)=-v(t,x,y,q)\,e^{-\eta_+ y}.
$$
Since $y\mapsto \eta_+\,e^{\eta_+ y}$ belongs to $C^\infty(\mathbb{R})$, it follows from
\eqref{eq:Jplus_loc_factored_app} that
$y\mapsto \mathcal{J}_{+,\mathrm{loc}}[v](t,x,y,q)$ belongs to
$C^1([y_{\min},y_{\max}])$.  Applying the product rule to
\eqref{eq:Jplus_loc_factored_app}, we have that
\begin{align*}
\partial_y \mathcal{J}_{+,\mathrm{loc}}[v](t,x,y,q)
&=
\eta_+^2\, e^{\eta_+ y}\, F_+(y)
\;+\;
\eta_+\, e^{\eta_+ y}\, F_+'(y) \\
&=
\eta_+\,\mathcal{J}_{+,\mathrm{loc}}[v](t,x,y,q)
\;-\;
\eta_+\, v(t,x,y,q) \\
&=
\eta_+\Bigl(
\mathcal{J}_{+,\mathrm{loc}}[v](t,x,y,q)-v(t,x,y,q)
\Bigr).
\end{align*}
At $y=y_{\max}$, the integral $F_+(y_{\max})=0$. Consequently,
$\mathcal{J}_{+,\mathrm{loc}}[v](t,x,y_{\max},q)=0$.

{By \eqref{eq:Jplus_tr},}
$$
{
\widehat{\mathcal{J}}_+[v](t,x,y,q)
=
\mathcal{J}_{+,\mathrm{loc}}[v](t,x,y,q)
+
B_+(y),
}
$$
{where}
$$
{
B_+(y)
:=
e^{-\eta_+(y_{\max}-y)}
\left(
v(t,x,y_{\max},q)
+
\frac{a_+(\tau)}{\eta_+}
\right).
}
$$
{Since $\partial_y B_+(y)=\eta_+B_+(y)$, it follows that}
\begin{align*}
{
\partial_y\widehat{\mathcal{J}}_+[v](t,x,y,q)}
&{
=
\eta_+\Bigl(
\mathcal{J}_{+,\mathrm{loc}}[v](t,x,y,q)
-
v(t,x,y,q)
+
B_+(y)
\Bigr)}
\\
&{
=
\eta_+\Bigl(
\widehat{\mathcal{J}}_+[v](t,x,y,q)
-
v(t,x,y,q)
\Bigr).}
\end{align*}
{At $y=y_{\max}$, we have
$\mathcal{J}_{+,\mathrm{loc}}[v](t,x,y_{\max},q)=0$ and
$B_+(y_{\max})
=
v(t,x,y_{\max},q)+a_+(\tau)/\eta_+$, which gives
\eqref{eq:ode_Jplus_loc}.}

{We now proceed analogously for the negative semi-infinite integral.}
Let $y\in[y_{\min},y_{\max}]$. The
change of variable $u=y+z$ in \eqref{eq:Jminus_loc} gives
\begin{equation}\label{eq:Jminus_loc_factored_app}
\mathcal{J}_{-,\mathrm{loc}}[v](t,x,y,q)
=
\eta_-\, e^{-\eta_- y}
\int_{y_{\min}}^{y} v(t,x,u,q)\,e^{\eta_- u}\,\mathrm{d}u.
\end{equation}
Let
$$
F_-(y):=\int_{y_{\min}}^{y} v(t,x,u,q)\,e^{\eta_- u}\,\mathrm{d}u.
$$
As before, continuity of $y\mapsto v(t,x,y,q)$ implies that
$F_-\in C^1([y_{\min},y_{\max}])$ with
$F_-'(y)=v(t,x,y,q)\,e^{\eta_- y}$, and it follows from
\eqref{eq:Jminus_loc_factored_app} that
$y\mapsto \mathcal{J}_{-,\mathrm{loc}}[v](t,x,y,q)$ belongs to
$C^1([y_{\min},y_{\max}])$. Applying the product rule to
\eqref{eq:Jminus_loc_factored_app}, we have that
\begin{align*}
\partial_y \mathcal{J}_{-,\mathrm{loc}}[v](t,x,y,q)
&=
-\eta_-^2\, e^{-\eta_- y}\, F_-(y)
\;+\;
\eta_-\, e^{-\eta_- y}\, F_-'(y) \\
&=
-\eta_-\,\mathcal{J}_{-,\mathrm{loc}}[v](t,x,y,q)
\;+\;
\eta_-\, v(t,x,y,q) \\
&=
\eta_-\Bigl(
v(t,x,y,q)-\mathcal{J}_{-,\mathrm{loc}}[v](t,x,y,q)
\Bigr).
\end{align*}
At $y=y_{\min}$, the integral $F_-(y_{\min})=0$. Consequently,
$\mathcal{J}_{-,\mathrm{loc}}[v](t,x,y_{\min},q)=0$.

{By \eqref{eq:Jminus_tr},}
$$
{
\widehat{\mathcal{J}}_-[v](t,x,y,q)
=
\mathcal{J}_{-,\mathrm{loc}}[v](t,x,y,q)
+
B_-(y),
}
$$
{where}
$$
{
B_-(y)
:=
e^{-\eta_-(y-y_{\min})}
\left(
v(t,x,y_{\min},q)
-
\frac{a_-(\tau)}{\eta_-}
\right).
}
$$
{Since $\partial_y B_-(y)=-\eta_-B_-(y)$, we obtain}
\begin{align*}
{
\partial_y\widehat{\mathcal{J}}_-[v](t,x,y,q)}
&{
=
\eta_-\Bigl(
v(t,x,y,q)
-
\mathcal{J}_{-,\mathrm{loc}}[v](t,x,y,q)
-
B_-(y)
\Bigr)}
\\
&{
=
\eta_-\Bigl(
v(t,x,y,q)
-
\widehat{\mathcal{J}}_-[v](t,x,y,q)
\Bigr).}
\end{align*}
{At $y=y_{\min}$,
$\mathcal{J}_{-,\mathrm{loc}}[v](t,x,y_{\min},q)=0$ and
$B_-(y_{\min})
=
v(t,x,y_{\min},q)-a_-(\tau)/\eta_-$, which gives
\eqref{eq:ode_Jminus_loc}.}
\end{proof}

\section{{Proof of Proposition~\ref{prop:cfl}}}
\label{app:monotonicity_jump}

\begin{proof}[Proof of Proposition~\ref{prop:cfl}]
Let $(i,k)$ be fixed and let $j\in\{1,\dots,N_y-1\}$ be an interior price variable
node. {Unrolling the recurrence in \eqref{eq:rec_plus} from}
$$
{
\widehat{\mathcal{J}}_{+,N_y}
=
W^{n+1}_{i,N_y,k}
+
\frac{a_+(\tau_{n+1})}{\eta_+}
}
$$
{gives}
\begin{equation}\label{eq:unrolled_plus}
{
\widehat{\mathcal{J}}_{+,j}
=
\sum_{\ell=j}^{N_y-1}
r_+^{\,\ell-j}(1-r_+)\,
W^{n+1}_{i,\ell,k}
+
r_+^{\,N_y-j}W^{n+1}_{i,N_y,k}
+
r_+^{\,N_y-j}
\frac{a_+(\tau_{n+1})}{\eta_+}.
}
\end{equation}
{Similarly, unrolling \eqref{eq:rec_minus} from}
$$
{
\widehat{\mathcal{J}}_{-,0}
=
W^{n+1}_{i,0,k}
-
\frac{a_-(\tau_{n+1})}{\eta_-}
}
$$
{yields}
\begin{equation}\label{eq:unrolled_minus}
{
\widehat{\mathcal{J}}_{-,j}
=
\sum_{\ell=1}^{j}
r_-^{\,j-\ell}(1-r_-)\,
W^{n+1}_{i,\ell,k}
+
r_-^{\,j}W^{n+1}_{i,0,k}
-
r_-^{\,j}
\frac{a_-(\tau_{n+1})}{\eta_-}.
}
\end{equation}

Substituting \eqref{eq:unrolled_plus} and \eqref{eq:unrolled_minus} into
\eqref{eq:discrete_jump} and isolating the term $\ell=j$, we obtain
\begin{equation}\label{eq:jump_update_expanded}
\begin{aligned}
{
W^{n+1}_{i,j,k}
+
\Delta\tau\,
\widehat{\mathcal{I}}_y[W^{n+1}]_{i,j,k}}
&{
=
\Delta\tau\,\lambda\,p_+
\sum_{\ell=j+1}^{N_y-1}
r_+^{\,\ell-j}(1-r_+)\,
W^{n+1}_{i,\ell,k}}
\\
&{\quad
+
\Delta\tau\,\lambda\,p_-
\sum_{\ell=1}^{j-1}
r_-^{\,j-\ell}(1-r_-)\,
W^{n+1}_{i,\ell,k}}
\\
&{\quad
+
\Delta\tau\,\lambda\,p_+
r_+^{\,N_y-j}W^{n+1}_{i,N_y,k}}
\\
&{\quad
+
\Delta\tau\,\lambda\,p_-
r_-^{\,j}W^{n+1}_{i,0,k}}
\\
&{\quad
+
\Bigl(
1-\Delta\tau\,\lambda\,
(p_+r_+ + p_-r_-)
\Bigr)W^{n+1}_{i,j,k}}
\\
&{\quad
+
\Delta\tau\,\lambda
\left(
p_+r_+^{\,N_y-j}
\frac{a_+(\tau_{n+1})}{\eta_+}
-
p_-r_-^{\,j}
\frac{a_-(\tau_{n+1})}{\eta_-}
\right).}
\end{aligned}
\end{equation}
where we used $p_+(1-r_+)+p_-(1-r_-)-1 = -(p_+r_++p_-r_-)$, which holds
since $p_++p_-=1$.

{All coefficients multiplying values of $W^{n+1}$ in
\eqref{eq:jump_update_expanded} are nonnegative except possibly the
diagonal coefficient. In particular, the coefficients of the boundary
values $W^{n+1}_{i,N_y,k}$ and $W^{n+1}_{i,0,k}$ are
$\Delta\tau\lambda p_+r_+^{N_y-j}$ and
$\Delta\tau\lambda p_-r_-^j$, respectively, and are nonnegative.
The final line of \eqref{eq:jump_update_expanded} depends only on the
prescribed slopes and is an additive term.}

It remains to verify that the diagonal coefficient is nonnegative. Since
$r_\pm\in(0,1)$ and $p_++p_-=1$, we have that
$$
0<p_+r_++p_-r_-<1.
$$
Under the CFL condition ${\Delta\tau}\le 1/\lambda$, it follows that
$$
0
\;\le\;
1-{\Delta\tau}\,\lambda
\;\le\;
1-{\Delta\tau}\,\lambda\,(p_+r_+ + p_-r_-).
$$
Consequently, {all coefficients multiplying the  values in
\eqref{eq:jump_update_expanded} are nonnegative.}

If we let ${\widetilde W^{n+1}}$ be another grid function satisfying
${W^{n+1}_{i,\ell,k}\le \widetilde W^{n+1}_{i,\ell,k}}$ for all
$\ell=0,\dots,N_y$, {and use the same prescribed slopes for both
grid functions, then the additive slope terms in
\eqref{eq:jump_update_expanded} are identical.} Applying
\eqref{eq:jump_update_expanded} to both ${W^{n+1}}$
and ${\widetilde W^{n+1}}$, we obtain
\begin{equation}
\label{eq:jump_mononotonicity_end}
{
W^{n+1}_{i,j,k}
+
\Delta\tau\,
\widehat{\mathcal{I}}_y[W^{n+1}]_{i,j,k}
}
\;\le\;
{
\widetilde W^{n+1}_{i,j,k}
+
\Delta\tau\,
\widehat{\mathcal{I}}_y[\widetilde W^{n+1}]_{i,j,k},
}
\qquad j=1,\dots,N_y-1.
\end{equation}

{Finally, after the interior update, the boundary values satisfy
\eqref{eq:post_jump_affine_refill}. Each boundary value is the adjacent
interior value plus a prescribed term independent of the input grid
function. Hence the numerical boundary condition is also order-preserving.}
Equation \eqref{eq:jump_mononotonicity_end} {together with
\eqref{eq:post_jump_affine_refill}} proves the monotonicity of the explicit
jump step under the provided CFL condition
${\Delta\tau}\le 1/\lambda$.
\end{proof}

\end{document}